\documentclass[aos]{imsart}
\usepackage{amsmath}
\usepackage{graphicx,psfrag,epsf}
\usepackage{enumerate}
\usepackage{natbib}
\usepackage{calc}
\usepackage{upgreek}
\usepackage[noend]{algpseudocode}
\usepackage{algorithmicx, algorithm}
\PassOptionsToPackage{svgnames}{xcolor}
\usepackage{url} % not crucial - just used below for the URL 
\usepackage{soul}
\usepackage{dsfont}
\usepackage{pgf,tikz}
\usepackage{forest}
\usepackage{tabularx}
\usepackage{extarrows}
\usepackage{mathrsfs}
\usepackage{graphics,epsfig,epsf,psfrag}
\usepackage{url}
\usepackage{bm}
\usepackage{color, appendix}
\usepackage{amsmath, amsfonts, amsthm,amssymb}
\usepackage{epstopdf}
\usepackage{hyperref}
\usepackage{enumerate, framed}
\usepackage{graphicx,lscape,rotate}
\usepackage[noend]{algpseudocode}
\usepackage{algorithmicx, algorithm}
\usepackage{comment}
\usepackage{booktabs}
\usepackage{multirow}
\usepackage{CJK}
\usepackage{tcolorbox}
\usepackage{lipsum}
\usepackage{adjustbox}
\usepackage[margin=1.5in]{geometry}
\tcbuselibrary{skins,breakable}
\usetikzlibrary{shadings,shadows}

\usepackage[acronym,nomain]{glossaries}
\setacronymstyle{long-sc-short}
\makeglossaries

\newacronym
{inform}{inform}{informative illocutionary force}
\newacronym
{point}{point}{demonstrative illocutionary force}
\usepackage{multicol}
\newglossarystyle{glossabbr}{%
    {\begin{multicols}{2}\raggedright}
        {\end{multicols}}

    \renewcommand*{\glsgroupheading}[1]{}

}

\usepackage{amssymb}% http://ctan.org/pkg/amssymb
\usepackage{tikz}
\usetikzlibrary{arrows.meta}
\usepackage{pifont}% http://ctan.org/pkg/pifont

\numberwithin{equation}{section}
\startlocaldefs
\theoremstyle{plain}

\newtheorem{theorem}{Theorem}[section]
\newtheorem{lemma}[theorem]{Lemma}
\newtheorem{corollary}[theorem]{Corollary}
\newtheorem{pro}[theorem]{Proposition}

\usepackage{thmtools}

\newtheorem{definition}[theorem]{Definition}
\theoremstyle{remark}

\newtheorem{ass}[theorem]{Assumption}
\endlocaldefs

\newcommand{\E}{\mathrm{E}}
\newcommand{\pr}{\mathrm{P}}
\newcommand{\pn}{\mathbb{P}_n}
\newcommand{\X}{\mathbf{X}}
\newcommand{\plim}{\operatorname{plim}}
\newcommand{\pa}{\operatorname{pa}}
\newcommand{\Pa}{\operatorname{Pa}}

\newcommand{\var}{\operatorname{var}}
\newcommand{\cov}{\operatorname{cov}}

\newcommand{\bcdot}{\boldsymbol{\cdot}}

\usepackage{mathtools}
\makeatletter
\newcommand{\vast}{\bBigg@{4}}
\newcommand{\Vast}{\bBigg@{5}}

\makeatother

\makeatletter
\newcommand*{\indep}{%
  \mathbin{%
    \mathpalette{\@indep}{}%
  }%
}
\newcommand*{\nindep}{%
  \mathbin{%                   % The final symbol is a binary math operator
    \mathpalette{\@indep}{\not}% \mathpalette helps for the adaptation
  }%
}
\newcommand*{\@indep}[2]{%
  \sbox0{$#1\perp\m@th$}%        box 0 contains \perp symbol
  \sbox2{$#1=$}%                 box 2 for the height of =
  \sbox4{$#1\vcenter{}$}%        box 4 for the height of the math axis
  \rlap{\copy0}%                 first \perp
  \dimen@=\dimexpr\ht2-\ht4-.2pt\relax
  \kern\dimen@
  {#2}%
  \kern\dimen@
  \copy0 %                       second \perp
} 
\makeatother

\usetikzlibrary{shapes,decorations,arrows,calc,arrows.meta,fit,positioning}
\tikzset{
    -Latex,auto,node distance =1 cm and 1 cm,semithick,
    state/.style ={ellipse, draw, minimum width = 0.7 cm},
    point/.style = {circle, draw, inner sep=0.04cm,fill,node contents={}},
    bidirected/.style={Latex-Latex,dashed},
    el/.style = {inner sep=2pt, align=left, sloped}
}
\usetikzlibrary{calc}
\usetikzlibrary{arrows.meta}
\usetikzlibrary{positioning}
\usetikzlibrary{fit}
\usetikzlibrary{shapes}

    {\endtcolorbox}

\newlength\oversetwidth
\newlength\underwidth
\newcommand\alignedoverset[2]{
  \settowidth\oversetwidth{$\overset{#1}{#2}$}
  \settowidth\underwidth{$#2$}
  \setlength\oversetwidth{\oversetwidth-\underwidth}
  \hspace{.5\oversetwidth}
  &
  \settowidth\oversetwidth{$\overset{#1}{#2}$}
  \settowidth\underwidth{$#2$}
  \setlength\oversetwidth{\oversetwidth-\underwidth}
  \hspace{-.5\oversetwidth}
  \overset{#1}{#2}
}

\usepackage[backgroundcolor=white]{todonotes}

\allowdisplaybreaks
\begin{document}

\begin{frontmatter}
\title{On Efficient Inference of Causal Effects \\with Multiple Mediators}
%\title{A sample article title with some additional note\thanksref{t1}}
% \runtitle{SICG}
%\thankstext{T1}{A sample additional note to the title.}

\begin{aug}
%%%%%%%%%%%%%%%%%%%%%%%%%%%%%%%%%%%%%%%%%%%%%%%
%% Only one address is permitted per author. %%
%% Only division, organization and e-mail is %%
%% included in the address.                  %%
%% Additional information can be included in %%
\author[A]{Haoyu Wei*\ead[label=e1]{h8wei@ucsd.edu}},
\author[B]{Hengrui Cai*\ead[label=e2]{hengrc1@uci.edu}}\footnote{Equal contribution.}
\author[C]{Chengchun Shi\ead[label=e3]{c.shi7@lse.ac.uk}}
\and
\author[D]{Rui Song\dag\ead[label=e4]{songray@gmail.com}}\footnote{Corresponding author.}
%%%%%%%%%%%%%%%%%%%%%%%%%%%%%%%%%%%%%%%%%%%%%%
%% Addresses                                %%
%%%%%%%%%%%%%%%%%%%%%%%%%%%%%%%%%%%%%%%%%%%%%%
\address[A]{Department of Economics, University of California, San Diego\printead[presep={,\ }]{e1}}
\address[B]{Department of Statistics, University of California, Irvine\printead[presep={,\ }]{e2}}
\address[C]{Department of Statistics, London School of Economics and Political Science \printead[presep={,\ }]{e3}}
\address[D]{Department of Statistics, North Carolina State University \printead[presep={,\ }]{e4}}
\end{aug}

\begin{abstract}
    This paper provides robust estimators and efficient inference of causal effects involving multiple interacting mediators. Most existing works either %heavily rely on the assumption of a linear structure 
    impose a linear model assumption among the mediators or are restricted to handle %a single mediator or independent multiple 
    conditionally independent mediators given the exposure. To overcome these %fundamental 
    limitations, we %refine the current 
    define causal and individual mediation effects %with multiple interacted mediators, which do not depend on a specific structure for the causal graph. Based on the proposed definitions, we 
    in a general setting, and employ a semiparametric framework to develop quadruply robust estimators for these causal effects. %The proposed estimators remain consistent under four types of possible misspecification. 
    We further establish the asymptotic normality of the proposed estimators and prove their local semiparametric efficiencies. The %validity of our approach is substantiated through theoretical results, as well as applications to both 
    proposed method is empirically validated via simulated and real datasets concerning psychiatric disorders in trauma survivors.
\end{abstract}



\end{frontmatter}

%\linenumbers
\section{Introduction}

Causal inference plays a crucial role in various fields, such as epidemiology \citep{hernan2004definition}, medicine \citep{hernan2000marginal}, education \citep{card1999causal}, and economics \citep{panizza2014public}. Within this spectrum, Pearl's causal graphical models \citep{pearl2000causality,pearl2009causal} have recently emerged as a powerful tool for disentangling causal structures among variables (such as confounders, exposure, mediator(s), and outcome). Causal mediation analysis, a core method for examining causal graphical models, aims to reveal the causal mechanisms underlying observed effects from exposure to outcome through the mediator(s), to evaluate the effectiveness of the intervention, and to better understand the roles of mediators \citep[see, for example, ][]{pearl2012causal,pearl2014interpretation}.

Existing statistical inferential tools for multiple mediators \citep[see e.g.,][]{robins1992identifiability,petersen2006estimation,imai2010general,vanderweele2015explanation,chakrabortty2018inference,cai2020anoce,shi2021testing} %within an unknown graph structure usually 
comprise the following three principal steps. Initially, causal structure learning methodologies \citep[see e.g.,][]{spirtes2000constructing,chickering2002optimal,nandy2018high,li2019likelihood,yuan2019constrained,li2023inference} are applied to estimate the causal graph, often presented by a directed acyclic graph (DAG), using observational data. 
%Common methodologies for this estimation include the PC algorithm \citep{spirtes2000constructing}, greedy equivalence search (GES) \citep{chickering2002optimal}, and the adaptively restricted greedy equivalence search (ARGES) \citep{nandy2018high}. 
In the absence of additional assumptions \citep{shimizu2006linear, neal2020introduction}, the graph is only identified up to a Markov equivalence class (MEC), and a completed partially directed acyclic graph (CPDAG) in such a class is often used to represent the graph structure. 
%The preference for estimating a CPDAG rather than a fully directed acyclic graph (DAG) is due to the potential non-identifiability of the true causal structure, or Directed Acyclic Graph (DAG), from observations alone in the absence of additional assumptions \citep{shimizu2006linear, neal2020introduction}. 
The subsequent step is the estimation of the causal effects of mediators based on the DAG or CPDAG obtained from the initial phase. For this task, a variety of estimation techniques have been proposed, including the application of ordinary least squares (OLS) estimators \citep{vanderweele2014causal,lin2017interventional,chakrabortty2018inference}, parametric models \citep{vanderweele2014mediation,vanderweele2016causal,chen2023discovery}, and nonparametric methods \citep{an2022opening,brand2023recent}. 
The final step is to conduct inferences based on the estimated effects, which often requires finding the exact (asymptotic) distributions of the estimators. As pointed out in \cite{chen2023discovery}, such an inference is often regarded a separate task and has received less attention in recent causal graph literature.
%with various challenges still yet to be resolved.

%Previous research on the estimation and inference of mediation effects has largely focused on scenarios involving total mediation effects with conditionally independent mediators \citep{preacher2008asymptotic, boca2014testing}, as well as on extending these analyses to cases with increasing dimensions \citep{zhang2016estimating}, or on estimating individual effects for a set of transformed, conditionally independent variables \citep{huang2016hypothesis}. A significant advancement in inferencing mediation effects was made by \cite{chakrabortty2018inference}, where mediator interaction is allowed. Despite these developments, the current state of work remains incomplete, with several challenges still needing to be addressed.

%and properties of causal effects of mediators currently being explored. 
%This requires an in-depth understanding of the theoretical foundations of these definitions and an examination of their connections under varying scenarios. 
%Another category does not impose such a condition, but estimating individual effects for a set of transformed, conditionally independent variables \citep{sampson2018fwer,djordjilovic2022optimal,zhao2022multimodal,zhao2022pathway}.

Although most of the existing work on causal mediation inference is limited to scenarios with a single mediator \citep{tchetgen2012semiparametric,tchetgen2013inverse,kennedy2017non,wang2018bounded,xia2023identification}, there are some studies that employ the three main standard steps to conduct mediation analysis. However, all of them fall short of comprehensive.  Theoretical challenges in unknown causal structures have led to current methods for multiple mediators inference being categorized mainly into three types. One approach assumes that multivariate mediators are conditionally independent given the treatment, or a set of transformed, conditionally independent variables, significantly simplifying the analysis \citep{preacher2008asymptotic, boca2014testing, zhang2016estimating, huang2016hypothesis, guo2023statistical, yuan2023confounding}. Another category, which does not impose this condition, relies on linear structural equation models (LSEMs) \citep{maathuis2009estimating,nandy2017estimating,nandy2018high,chakrabortty2018inference,zhao2022multimodal,zhao2022pathway}. The last category allows for a general causal structure and correlated mediators, but uses approximations, such as assuming Gaussian conditional distributions under exposure \citep{daniel2015causal,kim2019bayesian,tai2022path}, or following a Probit/logistic model for odds ratios \citep{vanderweele2014mediation,steen2017flexible,park2018causal}. However, these approaches present limitations for complex applications where the causal structure may not be correctly specified.

To bridge this significant gap in addressing potential model misspecification, we consider developing a semiparametric framework to infer causal effects, adapting the general causal structure. 
Extensive research exists on deriving double robust and highly efficient estimates of the total causal effect of exposure when the model is misspecified \citep{scharfstein1999adjusting,bickel2001inference,bang2005doubly}. Complementary to this, multiple robust estimators have been developed to quantify direct and indirect effects \citep{goetgeluk2008estimation,tchetgen2012semiparametric,chan2016globally,bhattacharya2022semiparametric,xia2023identification}. A notable benefit of these multiple robust techniques is their integration of dimension reduction strategies with confounding adjustment, such that the estimators are consistent and asymptotically normal, provided that at least one of the strategies is correct \citep{van2006targeted}. These methods also achieve semiparametric efficiency when all included strategies are correct \citep{van1996weak,bickel2001inference,bang2005doubly}. 
Despite considerable progress in the field, current multiple robust estimators are limited to only a single mediator. Hence, a new semiparametric inference is on demand for inferring causal effects involving multiple interacting mediators under (potentially) unknown causal graphs.

% On the other hand, since the validity of estimating and performing statistical inference on mediation effects hinges critically on the assumed causal model, the concern of model misspecification must be carefully considered. Semiparametric methods have been developed for analyzing observational studies, producing double robust and highly efficient estimates of the total causal effect of exposure when the model is misspecified \citep{scharfstein1999adjusting,bickel2001inference,bang2005doubly}. Similar methods and related multiple robust estimators have also been created to estimate controlled direct and indirect effects \citep{goetgeluk2008estimation,tchetgen2012semiparametric,chan2016globally,bhattacharya2022semiparametric,xia2023identification}. A significant advantage of multiple robust methods is that they combine dimension-reduction strategies for confounding adjustment. This results in an estimator of the causal effect that remains consistent and asymptotically normal, as long as at least one of the strategies is correct, without needing to know which strategy is indeed correct \citep{van2006targeted}. Moreover, it achieves semi-parametric efficiency when all strategies are correct \citep{van1996weak,bickel2001inference,bang2005doubly}. Despite these advancements, all current multiple robust estimators are only designed for a single mediator. 
% There is a clear lack of methods for making semiparametric inferences about the interventional effect on multiple interacting mediators with (potentially) unknown causal graphs.

\subsection*{Our Contributions}

We conclude our contributions with the following three folds.
\begin{itemize}
    \setlength\itemsep{1em}
    \item[1.] Conceptually, we introduce the causal direct and indirect interventional effects for individual mediators %. For the $j$-th mediator, these are denoted as $DM_j$ and $IM_j$, respectively 
    (Definition \ref{def_med} and Equation \eqref{TM_avg}). Our definitions expand upon those existing in various literature, accommodating a more general model setting.
    Specifically, it is applicable to both linear and non-linear models, thereby extending beyond existing literature such as \citet{nandy2017estimating,chakrabortty2018inference,cai2020anoce}. Moreover, our approach allows mediators to take a general value space, making it more flexible than the discrete settings as in \cite{albert2011generalized,lin2017interventional}. 
    Importantly, our definitions are consistent with the aforementioned literature when applied to the same settings. We further establish the identifiability results of the proposed definitions based on the estimated CPDAG from the data. 
    
    %Additionally, similar to the methodology in \cite{chakrabortty2018inference}, our framework is identifiable based on the estimated CPDAG from the data. 

    \item[2.] Methodological-wise, based on the proposed definitions, we firstly introduce the semiparametric framework concerning potential model misspecification under unknown graph structure for multiple interacted mediators (Theorem \ref{thm_EIF} and Corollary \ref{cor_eff_DM_IM}).
    Our analytical approach stands out for its novel insights into efficiency and robustness in the context of statistical inference of mediators on causal graphs. Specifically, we integrate four different estimating strategies to introduce new quadruply robust estimators for the causal effects of mediators. Additionally, we propose two algorithms to calculate these estimators together with the confidence intervals provided (Algorithm \ref{alg_QR}, Algorithm \ref{alg_fast_QR}, and Proposition \ref{pro_fast}) to handle general noises and to increase computational speed, respectively. Under a semi-linear framework (Assumption \ref{ass_slsem}), we derive concise parametric expressions for all proposed causal effects, %natural direct effects ($DE$) and natural indirect effects ($IE$)-as described in \cite{pearl2000causality}-along with the causal effects of mediators we defined ($DM_j$ and $IM_j$). Under these parametric expressions, we 
    and propose OLS estimators that can be computed using standard regressions, allowing for the direct acquisition of asymptotically valid confidence intervals simultaneously.
    % With respect to methodology, based on our definitions, to circumvent potential model misspecification issues, we present a semiparametric framework for the causal effects defined in this work (Theorem \ref{thm_EIF} and Corollary \ref{cor_eff_DM_IM}). Our analytical approach stands out for its novel insights on efficiency and robustness in the context of statistical inference of mediators on causal graphs. Specifically, we integrate four different estimating strategies to introduce new quadruply robust estimators for $DM_j$ and $IM_j$, and propose a general algorithm and a fast algorithm under some linear assumptions for calculating these estimators (Algorithm \ref{alg_QR}, Algorithm \ref{alg_fast_QR}, and Proposition \ref{pro_fast}). Furthermore, in the context of a semi-linear structure, we derive concise parametric expressions for natural direct effects ($DE$) and natural indirect effects ($IE$), as defined in \cite{pearl2000causality}, as well as  the causal effects of mediators we defined ($DM_j$ and $IM_j$). Based on these parametric expressions, we propose ordinary least squares (OLS) estimators under direct strategy. These OLS estimators are straightforward to compute by just applying regressions.

    \item[3.] From a theoretical perspective, we prove the asymptotic properties of both our OLS estimators and quadruply robust estimators under mild conditions. Specifically: (i) Our OLS estimators are asymptotical normal with the analytical form of asymptotic variance provided, even under high-dimensional setting (Theorem \ref{thm_CI_DE_IE_DM} and Theorem \ref{thm_CI_IM}); (ii) The introduced quadruply robust estimator is consistent to the true as long as at least some of the conditional densities or conditional expectations are correctly specified, even under a potently increasing function class. Moreover, if all the conditional densities and conditional expectations are correctly specified, and if converge at rates that are conservatively permissible by various machine learning approaches, these estimators can assuredly achieve $n^{1 / 2}$-consistency, asymptotic normality, and semiparametric efficiency (Theorem \ref{thm_QR_nor}). 
    
    %\blue{Lastly, drawing inspiration from the semiparametric efficiency bound achieved by the estimator, we introduce a semiparametric score-based variance estimator.}
    %\todohw{我看 \cite{shi2022off,liu2021multiply} 里面都给出了 mutiple robust estimator 的非 bootstrap 的方差估计量。所以他们在模拟中也给出了渐进 CI 正确性的模拟。
    %但 \cite{tchetgen2012semiparametric,xia2023identification} 都没有给出，他们artifical data 的模拟也没有给出，但 real data 中是给出 bootstrap 的 CI。我们这里是采取的第二种方法。但可以按照第一种加入 variance 的估计量。
    %但加入这个后，我们的模拟可能得重新做下，去印证 CI 的正确性。
    %}
    % On theoretical sense, we prove our OLS estimators and quadruply estimators have nice asymptotic properties under very mild conditions. Specially, (i) our OLS estimators can facilitate asymptotic statistical inferences with simulations estimated confidence intervals, even under increasing dimensions (Theorem \ref{thm_CI_DE_IE_DM} and Theorem \ref{thm_CI_IM}); (ii) the proposed quadruply estimators can still achieve consistency as long as some of the conditional densities or conditional expectations is correctly specified under a potently increasing function class, and can guarantee $n^{1 / 2}$-consistency, asymptotic normality, and semiparametric efficiency if the conditional densities and conditional expectations are correctly specified and have conservative convergence rates allowed by various machine learning methods (Theorem \ref{thm_QR_nor}). Finally, inspired by the semiparametric efficiency bound the estimator attains, we provide the semiparametric score-based variance estimator.
\end{itemize}

The rest of this paper is organized as follows: Section \ref{sec_pre} presents preliminary concepts. Section \ref{sec_def} formally defines the direct and indirect causal effects of mediators. Section \ref{sec_semi} outlines the semiparametric efficient scores for these causal effects. Section \ref{sec_straight_strategy} explores the direct strategy for estimating the causal effects defined in Section \ref{sec_def}, along with an OLS estimation procedure for semi-linear structures. Section \ref{sec_mr} presents alternative estimation strategies, including the introduction of novel quadruply robust estimators. This section also provides both a general algorithm and, under specific conditions, a faster algorithm for computing these estimators. Section \ref{sec_asy} discusses the asymptotic properties of both the OLS and quadruply robust estimators. Section \ref{sec_sim} presents various simulation results that validate the theories proposed for the estimators. In Section \ref{sec_emp}, an application of the
proposed estimators is used to analyze real data collected from trauma survivors. The glossary of notations, all proofs, and additional technical materials are collected in the Appendix.
%Lastly, the Appendix contains the definitions and relationships of path-specific effects with the causal effects defined in the main body, along with all proofs and technical tools.

\section{Preliminaries}\label{sec_pre}

\subsection{Graph Terminology} 

Consider a graph $\mathcal{G} =({X}, E)$ with a set of nodes $X$ and a set of edges $E$. There is at most one edge between any pair of nodes. If there is an edge between $X_i$ and $X_j$, then $X_i$ and $X_j$ are adjacent. The node $X_i$ is said to be a parent of $X_j$ if there is a directed edge from $X_i$ to $X_j$. Let the set of all parents of node $X_j$ in $\mathcal{G}$ be $ \operatorname{Pa} (X_j) = \operatorname{Pa}_{\mathcal{G}} (X_j) = \Pa_{X_j} (\mathcal{G})$, and all adjacent nodes of $X_j$ in $\mathcal{G}$ by $\operatorname{adj}(X_j) =  \operatorname{adj}_{\mathcal{G}}(X_j)$. A path from $X_i$ to $X_j$ in $\mathcal{G}$ is a sequence of distinct vertices, $\pi := \{a_0, a_1,\cdots,a_L\}\subset V$ such that $a_0 =X_i$, and $a_L=X_j$. A directed path from $X_i$ to $X_j$ is a path between $X_i$ and $X_j$ where all edges are directed towards $X_j$. 
% The node $X_i$ is said to be an ancestor of $X_j$ if there is a directed path from $X_i$ to $X_j$. Denote the set of all ancestors of node $X_j$ in $\mathcal{G}$ as $\operatorname{Anc} (X_j) = \operatorname{Anc}_{\mathcal{G}} (X_j) = \operatorname{Anc}_{X_j} (\mathcal{G})$. 
% Similarly, if there is a directed path from $X_i$ to $X_j$, then $X_j$ is a descendant of $X_i$. The set of all descendants of $X_i$ in $\mathcal{G}$ is denoted as $\operatorname{Des} (X_i) = \operatorname{Des}_{\mathcal{G}} (X_i) = \operatorname{Des}_{X_i} (\mathcal{G})$. 
A directed cycle is formed by the directed path from $X_i$ to $X_j$ together with the directed edge $X_j$ to $X_i$. A directed graph that does not contain directed cycles is called a directed acyclic graph (DAG). A directed graph is acyclic if and only if it has a topological ordering. 
% Finally, without causing ambiguity, we also denote $\mathcal{G}$ as the adjacency matrix corresponding to the causal graph $\mathcal{G}$.

\subsection{Causal Graph Structural Assumption}

Let $A$ be a \textbf{binary} exposure/treatment in $\{ 0, 1\}$, $M := (M_1,M_2,\cdots,M_p)^{\top} \in \mathbb{R}^p$ be mediators with dimension $p$ in its support $\mathcal{M} = \mathcal{M}_1 \times \cdots \times \mathcal{M}_p \subseteq \mathbb{R}^p$, and $Y \in \mathbb{R}$ be the outcome of interest. Additionally, we also consider that there are $t - 1$ confounders $C: = (C_1, \ldots, C_{t - 1})^{\top} \in \mathbb{R}^{t - 1}$ in its support $\mathcal{C} \subseteq \mathbb{R}^{t - 1}$. We would just let $t = 1$ here to represent the absence of confounders, that is $C = \varnothing$. Suppose that there exists a DAG $\mathcal{G}=(X, E)$ that characterizes the causal relationship among $X=(C^{\top}, A, M^\top, Y)^\top $, where the dimension of $X$ is $d = t + p + 1$. 
% In this paper, we allow the dimension of mediator $p = p_n$ can increase with the sample size $n$. 
We suppose we observe i.i.d data on $X = (C^{\top}, A, M^{\top}, Y)^{\top}$ is collected for $n$ subjects. 
To characterize our model, we consider the following assumptions.

\begin{ass}\label{ass_structure}
    The causal graph $\mathcal{G}$ satisfies \textit{Causal Markov Condition}, \textit{Causal Faithfulness Condition}, and \textit{Causal Sufficiency} \citep{hasan2023a}. The random vector $X$ satisfies the structure assumption: (i) No potential mediator is a direct cause of confounders $C$; (ii) The outcome $Y$ has no descendant; (iii) The only parents of treatment $A$ are confounders.
\end{ass}

In many instances, the accessible data offers an incomplete view of the inherent causal structure.  To address this gap, \textit{Causal Markov Condition}, \textit{Causal Faithfulness Condition}, and \textit{Causal Sufficiency} in the above assumption provide a sufficient condition for causal discovery in i.i.d. data contexts  \citep{lee2020towards,assaad2022survey,hasan2023a}. The rigorous definitions for them and related details can be found in Section 2.4 in \cite{hasan2023a}. Furthermore, our structural assumptions aim at ensuring the identifiability of the causal model, which are similar to Consistency Assumption and Sequential Ignorability Assumption in \cite{tchetgen2012semiparametric}, and the structure assumptions in Section 2.4 of \cite{chakrabortty2018inference}.

\subsection{Markov Equivalence Class}

A general causal DAG, $\mathcal{G}$, may not be identifiable from the distribution of $X$. According to \cite{pearl2000causality}, a DAG only encodes conditional independence relationships through the concept of $d$-separation. In general, several DAGs can encode the same conditional independence relationships, and such DAGs form a Markov equivalence class. Two DAGs belong to the same Markov equivalence class if and only if they have the same skeleton and the same v-structures \citep{kalisch2007estimating}. A Markov equivalence class of DAGs can be uniquely represented by a completed partially directed acyclic graph (CPDAG) \citep{spirtes2000constructing}, which is a graph that can contain both directed and undirected edges. A CPDAG satisfies the following: $X_i \leftrightarrow X_j$ in the CPDAG if the Markov equivalence class contains a DAG including $X_i \rightarrow X_j$, as well as another DAG including $X_j \rightarrow X_i$. CPDAGs can be estimated from observational data using various algorithms, such as the algorithms in \cite{kalisch2007estimating}, \cite{harris2013pc}, and \cite{zhang2018non}. The Markov equivalence class for a fixed CPDAG $\mathcal{C}$ is denoted by $\operatorname{MEC}(\mathcal{C})$, which is a set containing all DAGs $\mathcal{G}$ that have the CPDAG structure $\mathcal{C}$. 
% In general, if the error distribution is joint Gaussian with different diagonal elements in the LSEM satisfying our structure assumption, the data distribution can only give a CPDAG and thus a MEC \citep{shimizu2006linear}. In the linear setting, however, if the error distribution is non-Gaussian, the distribution of $X$ will completely give a unique DAG; see Theorem 11.4 in \cite{neal2020introduction}. Since in practice we usually do not know whether the error is Gaussian, it is safer to assume that we can only obtain a CPDAG instead of a DAG from the data. 
If we can obtain the true DAG from the data, we can simply treat it as a special case of the "MEC" containing only this DAG, i.e., $\operatorname{MEC}(\mathcal{G})= \{ \mathcal{G} \}$. For simplicity, we denote the corresponding causal structure for the mediators $M$ as $\mathcal{G}_M$, which can be obtained by deleting nodes $C, A, Y$, and the corresponding edges from $\mathcal{G}$. 
% Specially, we denote the causal structure of mediators as $\mathcal{G}_{M}$. 
% Obviously, the adjacency matrix $\mathcal{G}_{M}$ satisfies $\mathcal{G}_{M} = \big[ \mathcal{{G}}_{kk} \big]_{k \in \{ t + 1, \ldots, t + p\}}$. 
The CPDAG of mediators is similarly denoted as $\mathcal{C}_M$. For simplicity and with a minor stretch of notation, we employ $\mathcal{G}_{M}$ and $\mathcal{C}_M$ to denote the causal DAG and CPDAG of $X$, respectively, such that their corresponding mediators' causal DAG and CPDAG are represented by $\mathcal{G}_{M}$ and $\mathcal{C}_M$ exactly.

\section{Definition of Causal Effects}\label{sec_def}

In this section, we will formally give our refined definition of the causal effects of mediators. 
To begin with, we give the total effect $TE$, the natural direct effect that is not mediated by mediators $DE$, and the natural indirect effect that is regulated by mediators $IE$ defined in \cite{pearl2009causal}.

\begin{definition}[\cite{pearl2009causal}]\label{def_na}
        Natural effects are defined as follows:
        \begin{eqnarray*}
            &&TE = \E \big[ Y \mid do(A=1) \big] - \E \big[ Y \mid do(A=0) \big] ,\\
            &&DE = \E \Big[ \E \big[ Y \mid do(A = 1, M = M^{(0)})\big] \Big] - \E \big[ Y \mid do(A=0) \big],\\
            &&IE = \E \Big[ \E \big[ Y \mid do(A=0, M = M^{(1)}) \big] \Big] - \E \big[ Y \mid do(A=0) \big].
        \end{eqnarray*}
        %     &TE = \E \big[ Y \mid do(A=1) \big] - \E \big[ Y \mid do(A=0) \big] ,
        % \]
        % \[
        %     &DE = \E \Big[ \E \big[ Y \mid do(A = 1, M = M^{(0)})\big] \Big] - \E \big[ Y \mid do(A=0) \big],
        % \]
        % \[
        %     &IE = \E \Big[ \E \big[ Y \mid do(A=0, M = M^{(1)}) \big] \Big] - \E \big[ Y \mid do(A=0) \big].
        % \]
\end{definition}

\noindent In the above definition,  $do(A=0) = do_{\mathcal{G}}(A = 0)$ is a mathematical operator to simulate physical interventions that hold $A$ constant as $0$ while keeping the rest of the model unchanged, which corresponds to remove edges into $A$ and replace $A$ by the constant $0$ in the original causal graph $\mathcal{G}$. Here, $M^{(0)}$ is the (random) value of $M$ if setting $do(A=0)$, and $M^{(1)}$ is the (random) value of $M$ if setting $do(A=1)$. 
One can refer to \cite{pearl2009causal} for more details of `do-operator'. 
The expectation $\E [\cdot]$ is an abbreviation of $\E_P [\cdot]$ with $P = P_X$ is the law of $X$ under $\mathcal{G}$.
Inspired by the above definition, we can give the definition of the causal effects for an individual mediator.

\begin{definition}\label{def_med}
    Let $TM_j(\mathcal{G}_M)$ represent total individual mediation effects via an individual mediator $M_j$ defined as 
\begin{eqnarray*}
        TM_j(\mathcal{G}_M) := \bigg\{ \E \big[ Y \mid do(A = 1) &&\big]  - \E \Big[ \E \big[ Y \mid do_{\mathcal{G}_M}(A = 1, M_j)\big] \Big] \bigg\} \\
            && - \bigg\{ \E \big[ Y \mid do(A = 0) \big] - \E \Big[ \E \big[ Y \mid do_{\mathcal{G}_M}(A = 0, M_j)\big] \Big] \bigg\},
        \end{eqnarray*}
    under any fixed mediators' causal structure $\mathcal{G}_M$. Let $DM_j$ denote the direct interventional effect via an individual mediator $M_j$ defined as 
    %in \cite{vansteelandt2017interventional},
    \begin{eqnarray*}
            DM_j := \E \Bigg[ \int_{\mathcal{M}} \E \big[ Y \mid C, A && = 1, M = m\big] f_{M_{-j} \mid C, A}(m_{-j} \mid C, A = 0 ) \\
            && \times \big\{ f_{M_j \mid C, A} (m_j \mid C, A = 1) - f_{M_j \mid C, A}(m_j \mid C, A = 0) \big\} \, \mathrm{d} m \Bigg],
        \end{eqnarray*}
    where $f_{\cdot \mid \cdot} (\cdot \mid \cdot)$ are the conditional density (or mass) functions.
    Then the indirect interventional effect for $M_j$ under $\mathcal{G}_M$ is defined as $IM_j(\mathcal{G}_M) := TM_j(\mathcal{G}_M) - DM_j$.
\end{definition}

\noindent The Definition \ref{def_med} serves important meanings when we are concerned with different impacts of mediators. Note that $TM_j(\mathcal{G}_M)$ in our definition is an extension for the individual mediation effect in \cite{chakrabortty2018inference} for LSEMs, denoted as 
\[
    \eta_j (\mathcal{G}_M) = \frac{\partial}{\partial a} \E \big[ Y \mid do(A = a) \big] -  \frac{\partial}{\partial a} \E \big[ Y \mid do_{\mathcal{G}_M}(A = a, M_j = m_j)\big],
\]
which can be interpreted as the change in the total causal effect of the exposure $A$ on the response $Y$ when the potential mediator $M_j$ is removed from the causal graph $\mathcal{G}$ through the intervention $d o\left(M_j=m_j\right)$. But under the non-linearity assumption with binary exposure, the above $\eta_j$ will be a function of $j$-th mediator $m_j$ (Remark 2.1 in \cite{chakrabortty2018inference}). Therefore, for solving this problem, we take integral with respect to the density $f_{M_j}(m_j)$ for $M_j$, i.e.
\[
    \begin{aligned}
        & \int_{\mathcal{M}_j} \bigg[ \frac{\partial}{\partial a} \E \big[ Y \mid do(A = a) \big] -  \frac{\partial}{\partial a} \E \big[ Y \mid do_{\mathcal{G}_M}(A = a, M_j = m_j) \big] \bigg] f_{M_j}(m_j) \, \mathrm{d} m_j \\
        = & \frac{\partial}{\partial a} \E \big[ Y \mid do(A = a) \big] - \E \bigg[ \frac{\partial}{\partial a} \E \big[ Y \mid do_{\mathcal{G}_M}(A = a, M_j) \big] \bigg].
    \end{aligned}
\]
Then by $ a \in \{ 0, 1\}$, we get the expression of $TM_j(\mathcal{G}_M)$ in Definition \ref{def_med}. Our definition of $DM_j$ is a straightforward extension of Equation (6) in \cite{vansteelandt2017interventional} by replacing summation with integral. The introduction of $DM_j$ and $IM_j(\mathcal{G}_M) = TM_j(\mathcal{G}_M) - DM_j$ in the above definition is driven by the need for an orthogonal decomposition of the total causal effect of mediators, a concept crucial for unraveling the intricate relationships among variables in mediation analysis, as argued in \cite{cai2020anoce}.
Here $DM_j$ can be interpreted as the causal effect through a particular mediator from the exposure to the outcome, i.e., $\E \big[ Y \mid C, A = 1, M = m\big]  \big\{ f_{M_j\mid C, A} (m_j \mid C, A = 1) - f_{M_j\mid C, A} (m_j \mid C, A = 0) \big\}$,
that is not regulated by any other mediators, i.e., $f_{M_{-j}\mid C, A}(m_{-j} \mid C, A = 0)$, and thus not regulated by its descendant mediators. Then $IM_j(\mathcal{G}_M) = TM_j(\mathcal{G}_M) - DM_j$ captures the indirect effect of the particular mediator $M_j$ on the outcome $Y$ regulated by descendant mediators. 

Definition \ref{def_med} is generally defined for any causal graph with binary exposure $A$. In the context of particular linear causal structures, these definitions transform into concise parametric expressions, providing more intuitive `do' representations and aligning consistently with existing literature. A comprehensive discussion on this can be found in Section \ref{sec_ols_est}.

Usually, we do not know the true structure of mediators $\mathcal{G}_{M}$ and we can only estimate its corresponding CPDAG $\mathcal{C}_{ M}$ \citep{maathuis2009estimating}. 
If the number of $\operatorname{MEC}(\mathcal{C}_M)$ is larger than one, $TM_j$ and $IM_j$ based on elements $\operatorname{MEC}(\mathcal{C}_{M})$ will be not unique. 
As a result, we define an identifiable version of $TM_j$ based on a CPDAG as the average over $\operatorname{MEC}(\mathcal{C}_{0, M})$. Specifically,
\begin{equation}\label{TM_avg}
    \begin{aligned}
        \overline{TM}_j & := \frac{1}{\# \operatorname{MEC}(\mathcal{C}_{M})} \sum_{\mathcal{G}_{M} \in \operatorname{MEC}(\mathcal{C}_{M})} TM_j(\mathcal{G}_{M}).
    \end{aligned}
\end{equation}
The corresponding identifiable indirect interventional effect of mediator $M_j$ is $\overline{IM}_j := \overline{TM}_j - DM_j$ for $j \in [p]$.

\section{Semiparametric Efficient Scores}\label{sec_semi}

We start with exploring the definition in Section \ref{sec_def}. Define propensity score $e_{a'}(x_S) := \pr (A = a' \mid X_S = x_S)$ for $a' \in \{ 0, 1\}$, the outcome mean $\mu_{}(x_S) := \E \big[ Y \mid X_S = x_S\big]$, and the conditional density $\pi_{x_S}(m_T) := f_{M_T \mid X_S}(m_{T} \mid X_{S} = x_{S})$ for any subset $T \subseteq [p]$ and $S \subseteq [t + p + 1]$. Suppose all these functions belong to $\ell^2$-class. Note that in our notation, \( x_S \) and \( m_T \) can represent vectors. At times, we may abbreviate \( \mu \big((x_{S_1}^{\top}, x_{S_2}^{\top}, \ldots, x_{S_N}^{\top})^{\top} \big) \) as \( \mu (x_{S_1}, x_{S_2}, \ldots, x_{S_N}) \) and \( \pi_{(x_{S_1}^{\top}, x_{S_2}^{\top}, \ldots, x_{S_N}^{\top})^{\top}}(m_T) \) as \( \pi_{x_{S_1}, x_{S_2}, \ldots, x_{S_N}}(m_T) \) when \( x_S \) is the concatenation of these vectors, i.e., \( x_S = (x_{S_1}^{\top}, x_{S_2}^{\top}, \ldots, x_{S_N}^{\top})^{\top} \).
Let $\Pa_j (\mathcal{G}_M) = \Pa_{\mathcal{G}_M}(M_j) \subseteq \{ M_1, \ldots, M_{j - 1}, \penalty 0 M_{j + 1}, \ldots, M_p\}$ denote the \textit{parent mediators} of $M_j$ and $\pa_j (\mathcal{G}_M)$ as its realization.
% \todohw{I ask $\ell^2$ class instead of $C^{\infty}$ to ensure the Cauchy inequality can hold in proof. Counterexample
% $$
% x \mapsto \begin{cases}\mathrm{e}^{-1 /\left(x^2(x-1)^2\right)} & \text { if } x \in(0,1) \\ 0 & \text { otherwise. }\end{cases}
% $$}
% \todohw{The $\ell^2$ assumption is a simple replacement for bounded assumption for $\mu$, $\pi$.}
Denote 
\begin{equation}\label{def_kappa}
    \kappa (a', C) := \E [Y \mid C, A = a'],
    %= \mu(C, a'),
\end{equation}
\begin{equation}\label{def_zeta}
        \begin{aligned}
            %\zeta_j (M_j = M_j^{(a')}, M_{-j} = M_{-j}^{(a)}, C) & = 
            % & \zeta_j (a', a, C) \\
            % := & \int_{\mathcal{M}} \E [Y \mid C, A = a^*, M = m] f_{M_j \mid C, A} (m_j \mid C, A = a') f_{M_{-j} \mid C, A} (m_{-j} \mid C, A = a) \, \mathrm{d} m \\
            % = & \int_{\mathcal{M}} \mu(C, a^*, m) \pi_{C, a'}(m_j) \pi_{C, a}(m_{-j}) \, \mathrm{d} m,
            %& = \E_M  \E \big[Y \mid do(A = a^*, M_j = M_j^{(a')}, M_{-j} = M_{-j}^{(a)}), C \big],
            \zeta_j (a', a, C) := \int_{\mathcal{M}} \mu(C, a^*, m) \pi_{C, a'}(m_j) \pi_{C, a}(m_{-j}) \, \mathrm{d} m,
        \end{aligned}
\end{equation}
and 
\begin{equation}\label{def_varrho}
    \begin{aligned}
        % \varrho_j (A = a', M_j, C \, ; \, \mathcal{G}_M) & =
        % & \varrho_j (a', M_j, C \, ; \, \mathcal{G}_M) \\
        % := & \int_{\mathcal{M}_{\pa_j} (\mathcal{G}_M)} \E [Y \mid C, A = a', \Pa_j(\mathcal{G}_M) = \pa_j, M_j] f_{\Pa_j (\mathcal{G}_M) \mid C, A} (\pa_j \mid C, A = a') \, \mathrm{d} \pa_j  \\
        % = & \int_{\mathcal{M}_{\pa_j} (\mathcal{G}_M)} \mu(C, a', \pa_j, M_j) \, \pi_{C, a'}(\pa_j)  \, \mathrm{d} \pa_j,
        \varrho_j (a', M_j, C \, ; \, \mathcal{G}_M) := \int_{\mathcal{M}_{\pa_j} (\mathcal{G}_M)} \mu(C, a', \pa_j, M_j) \, \pi_{C, a'}(\pa_j)  \, \mathrm{d} \pa_j,
    \end{aligned}
\end{equation}
for a fixed causal graph $\mathcal{G}_M$, where $\mathcal{M}_{\pa_j} (\mathcal{G}_M)$ is the support of $\Pa_j (\mathcal{G}_M)$ and $a^* = 0$ is the reference level of the exposure. It is worth noting that all three quantities above are random due to the randomness in $C$ (and $M_j$). 
Given that the exposure features two levels, $0$ and $1$, for simplicity, we use the notation $\langle \cdot \rangle$ to signify the difference evaluated at these two levels. Specifically, let us define
\begin{equation}\label{def_diff_notation}
    \big\langle g_{\boldsymbol{\cdot}, u} (\boldsymbol{\cdot}, v) \big\rangle := g_{1, u} (1, v) - g_{0, u} (0, v),
\end{equation}
for any function $g_{\boldsymbol{\cdot}, u} (\boldsymbol{\cdot}, v)$, where $u$ and $v$ are arbitrary parameters.
Then we can have the following theorem to characterize the relationship between the interventional effects of mediators as specified in Definition \ref{def_med} and the quantities defined above.

\begin{theorem}\label{thm_direct_exp}
    Suppose Assumption \ref{ass_structure} holds, then for any $j \in [p]$,
    \[
        %DM_j = \E \big[ \zeta_j (1, 0, C) - \zeta_j(0, 0, C) \big],
        DM_j = \E \big[ \big\langle \zeta_j (\boldsymbol\cdot, 0, C) \big\rangle \big],
    \]
    and for any fixed $\mathcal{G}_M$ 
    \[
        \begin{aligned}
            %TM_j(\mathcal{G}_M) = \E \Big[ \big\{\kappa(1, C) - \varrho_j (1, M_j, C \, ; \, \mathcal{G}_M) \big\} - \big\{ \kappa(0, C) - \varrho_j (0, M_j, C\, ; \, \mathcal{G}_M) \big\} \Big].
            TM_j(\mathcal{G}_M) = \E \big[ \big\langle \kappa(\bcdot, C) - \varrho_j (\bcdot, M_j, C \, ; \, \mathcal{G}_M) \big\rangle \big].
        \end{aligned}
    \]
\end{theorem}
Consider $\mathscr{M}_{\text {nonpar }}$ as the full model where the observed data likelihood is not constrained, encompassing all conventional laws $P_X$ or, equivalently, distribution $F_X$ of the observed data $X$. The aforementioned theorem establishes that the causal effects detailed in Section \ref{sec_def} can be represented as regular expectations. Consequently, they function as mappings from $\mathscr{M}_{\text{nonpar}}$ to the real line $\mathbb{R}$. We assume $P_X$
satisfies the positivity assumption given below. 
\begin{ass}\label{ass_pos}
     There exists a $\varepsilon > 0$ such that for any $c \in \mathcal{C}$, $a' \in \{ 0, 1\}$, and $m \in \mathcal{M}$,
     \[
        \varepsilon < e_{a'}(c) < 1 - \varepsilon \qquad \text{and} \qquad \varepsilon < \pi_{c, a', m_j}(m_{-j}) < \infty,
     \]
     with probability one. 
\end{ass}
% \todohw{I change $0$ to $\varepsilon$ as common literature (for proof) but different to \cite{tchetgen2012semiparametric}}
The efficient scores for the functionals $DE$ and $IE$ have been studied in various literature \citep{tchetgen2012semiparametric,tchetgen2013inverse,shi2020multiply}. The explicit expressions and detailed analysis can be seen in Theorem 1 of \cite{tchetgen2012semiparametric}. For finding the efficient scores for the functionals $DM_j$ and $IM_j$, we denote
\begin{equation}\label{def_tau}
    \tau_{\boldsymbol{\cdot} \, ; \, S}(C, a', M_T) := \int_{\mathcal{M}_{S}} \mu (C, a', m_{S}, M_{T})  \pi_{\boldsymbol{\cdot}} (m_{S}) \, \mathrm{d} m_{S}
\end{equation}
for any $a' \in \{ 0, 1\}$ and $T, S\subseteq [p]$. Then we can derive the efficient score for $\kappa(a', C)$, $\zeta_j (a', a, C)$, and $\varrho_j (a', M_j, C \, ; \, \mathcal{G}_M)$ on $\mathscr{M}_{\text{nonpar}}$ in the following theorem.

\begin{theorem}\label{thm_EIF}
    Suppose the Assumption \ref{ass_structure} and \ref{ass_pos} hold, we have the efficient scores for $\E \kappa(a', C)$, $\E \zeta_j(a', 0, C)$, and $\E \varrho_j (a', M_j, C\, ; \, \mathcal{G}_M)$ as
    \[
         S^{\text{eff}, \text{nonpar}} \big(\E \kappa(a', C) \big) = \frac{\mathds{1}(A = a')}{e_{a'} (C)} \Big\{ Y - \kappa(a', C) \Big\} + \kappa(a', C) - \E \kappa(a', C),
    \]
    \[
        \begin{aligned}
            % & S^{\text{eff}, \text{nonpar}} \big( \E \zeta_j(a', 0, C) \big) = \frac{\mathds{1}(A = 1)}{e_1(C)} \frac{  \pi_{C, a'}(M_{-j}) }{\pi_{C, 1, M_{j}}(M_{-j})} \Big\{ Y - \mu (C, 1, M) \Big\} \\
            % & + \frac{\mathds{1}(A = 0)}{e_0(C)} \bigg\{ \int_{\mathcal{M}_j} \mu (C, 1, m_j, M_{-j}) \pi_{C, a'}(m_j) \, \mathrm{d} m_j - \zeta_j (a', 0, C) \bigg\} \\
            % & + \frac{\mathds{1}(A = a')}{e_{a'} (C)} \bigg\{ \int_{\mathcal{M}_{-j}} \mu(C, 1, M_j, m_{-j}) \pi_{C, 0} (m_{-j}) \, \mathrm{d} m_{-j} - \zeta_j (a', 0, C) \bigg\} +  \zeta_j (a', 0, C)  - \E  \zeta_j (a', 0, C),
            & S^{\text{eff}, \text{nonpar}} \big( \E \zeta_j(a', 0, C) \big) = \frac{\mathds{1}(A = 1)}{e_1(C)} \frac{  \pi_{C, a'}(M_{-j}) }{\pi_{C, 1, M_{j}}(M_{-j})} \Big[ Y - \mu (C, 1, M) \Big] \\
            & + \frac{\mathds{1}(A = 0)}{e_0(C)} \Big[ \tau_{C, a'; j}(C, 1,M_{-j})  - \zeta_j (a', 0, C) \Big] + \frac{\mathds{1}(A = a')}{e_{a'} (C)} \Big[ \tau_{C, 0 \, ; \, -j}(C, 1, M_j) - \zeta_j (a', 0, C) \Big]  \\
            & +  \zeta_j (a', 0, C)  - \E  \zeta_j (a', 0, C),
        \end{aligned}
    \]
    and
    \[
        \begin{aligned}
            % &  S^{\text{eff}, \text{nonpar}} \big( \E \varrho_j(a', M_j, C \, ; \, \mathcal{G}_M) \big) = \frac{\mathds{1}(A = a')}{e_{a'}(C)} \pi_{C, a'} \big(\Pa_j (\mathcal{G}_M) \big) \Big\{  Y - \mu(C, a', M) \Big\} \\
            % & + \frac{\mathds{1}(A = a')}{e_{a'}(C)} \left\{ \int_{\mathcal{M}_j} \mu(C, a', \Pa_j(\mathcal{G}_M), m_j) \pi_{C}(m_j) \, \mathrm{d} m_j - \E \big[ \varrho_j (a', M_j, C \, ; \, \mathcal{G}_M) \mid C \big]\right\} \\
            % & + \varrho_j(a', M_j, C \, ; \, \mathcal{G}_M) - \E \big[ \varrho_j (a', M_j, C \, ; \, \mathcal{G}_M) \mid C \big] + \E \big[ \varrho_j (a', M_j, C \, ; \, \mathcal{G}_M) \mid C \big] - \E \varrho_j(a', M_j, C \, ; \, \mathcal{G}_M) 
            &  S^{\text{eff}, \text{nonpar}} \big( \E \varrho_j(a', M_j, C \, ; \, \mathcal{G}_M) \big) = \frac{\mathds{1}(A = a')}{e_{a'}(C)} \pi_{C, a'} \big(\Pa_j (\mathcal{G}_M) \big) \Big[ Y - \mu(C, a', M) \Big] \\
            & + \frac{\mathds{1}(A = a')}{e_{a'}(C)} \Big[ \tau_{C\, ; \, j}(C, a', \Pa_j(\mathcal{G}_M))  - \E \big[ \varrho_j (a', M_j, C \, ; \, \mathcal{G}_M) \mid C \big] \Big] \\
            & + \varrho_j(a', M_j, C \, ; \, \mathcal{G}_M) - \E \big[ \varrho_j (a', M_j, C \, ; \, \mathcal{G}_M) \mid C \big] + \E \big[ \varrho_j (a', M_j, C \, ; \, \mathcal{G}_M) \mid C \big] - \E \varrho_j(a', M_j, C \, ; \, \mathcal{G}_M) 
        \end{aligned}
    \]
    under model $\mathscr{M}_{\text{nonpar}}$ for any $j \in [p]$, $a' \in \{ 0, 1\}$, and fixed $\mathcal{G}_M$, where 
    \[
        \begin{aligned}
            & \E \big[ \varrho_j (a', M_j, C \, ; \, \mathcal{G}_M) \mid C \big] \\
            % = & \int_{\mathcal{M}_{\pa_j} (\mathcal{G}_M) \, \cup \, \mathcal{M}_j} \E \big[ Y \mid C, A = a', \Pa_j (\mathcal{G}_M)  = \pa_j, M_j = m_j\big] \\
            % & ~~~~~~~~~~~~~~~~~~~~~~~~~~~~~~~ \times f_{\Pa_j (\mathcal{G}_M) \mid C, A}(\pa_j \mid C, A = a') f_{M_j \mid C}(m_j \mid C)  \, \mathrm{d} \pa_j \, \mathrm{d} m_j \\
            = & \int_{\mathcal{M}_{\pa_j} (\mathcal{G}_M) \, \cup \, \mathcal{M}_j} \mu(C, a', \pa_j(\mathcal{G}_M), m_j) \pi_{C, a'} \big(\pa_j (\mathcal{G}_M) \big) \pi_{C}(m_j) \, \mathrm{d} \pa_j(\mathcal{G}_M) \, \mathrm{d} m_j.
        \end{aligned}
    \]
    Here, we let $\pi_{C, a'} \big(\pa_j (\mathcal{G}_M) \big) \equiv 1$ if $\Pa_j (\mathcal{G}_M) = \varnothing$.
\end{theorem}

\noindent In Theorem \ref{thm_EIF}, we retain the final two terms in $S^{\text{eff}, \text{nonpar}} \big( \E \varrho_j(a', M_j, C , ; , \mathcal{G}_M) \big)$, because we want to express $S^{\text{eff}, \text{nonpar}} \big( \E \varrho_j(a', M_j, C \, ; \, \mathcal{G}_M) \big)$ is composed of four parts. This result can be then combined with Theorem \ref{thm_direct_exp} to obtain the efficient scores for $DM_j$ and $IM_j(\mathcal{G}_M)$, thus $TM_j(\mathcal{G}_M) = DM_j + IM_j(\mathcal{G}_M)$. %Therefore, we can obtain the efficient scores for causal effects we interested as follows.

\begin{corollary}\label{cor_eff_DM_IM}
    Suppose the conditions in Theorem \ref{thm_EIF} holds, then we have 
    \[
        \begin{aligned}
            S^{\text{eff}, \text{nonpar}} (TE) = \Big\langle S^{\text{eff}, \text{nonpar}} \big(\E \kappa(\boldsymbol{\cdot}, C) \big) \Big\rangle,
            % S^{\text{eff}, \text{nonpar}} (TE) & = \frac{\mathds{1}(A = 1)}{e_1(C)} \Big[ Y - \kappa(1, C) \Big] - \frac{\mathds{1}(A = 0)}{e_0(C)} \Big[ Y - \kappa(0, C) \Big]  \\
            % & ~~~~~~ + \big[ \kappa(1, C) - \kappa(0, C)  \big] - \E \big[ \kappa(1, C) - \kappa(0, C)  \big],
        \end{aligned}
    \]
    \[
        \begin{aligned}
            S^{\text{eff}, \text{nonpar}} (DM_j) =  \Big\langle S^{\text{eff}, \text{nonpar}} \big( \E \zeta_j(\boldsymbol\cdot, 0, C) \big) \Big\rangle,
            % & S^{\text{eff}, \text{nonpar}} (DM_j) = \frac{\mathds{1}(A = 1)}{e_1(C)} \bigg[ \frac{\pi_{C, 1}(M_{-j}) - \pi_{C, 0}(M_{-j})}{\pi_{C, 1, M_j}(M_{-j})} \bigg]  \times  \Big\{ Y - \mu(C, 1, M)\Big\} \\
            % % & + \frac{\mathds{1}(A = 0)}{e_0(C)} \bigg\{ \int_{\mathcal{M}_j} \mu(C, 1, m_j, M_{-j}) \big( \pi_{C, 1}(m_j) - \pi_{C, 0}(m_j)\big) \, \mathrm{d} m_j  - \Big( \zeta_j (1, 0, C) - \zeta_j (0, 0, C) \Big) \bigg\} \\
            % & + \frac{\mathds{1}(A = 0)}{e_0(C)} \bigg[  \Big( \tau_{C, 1 \, ; \, j} (C, 1, M_{-j}) - \tau_{C, 0\, ; \, j} (C, 1, M_{-j}) \Big)  - \Big( \zeta_j (1, 0, C) - \zeta_j (0, 0, C) \Big) \bigg] \\
            % % & + \bigg[ \frac{\mathds{1}(A = 1)}{e_1(C)} -  \frac{\mathds{1}(A = 0)}{e_0(C)} \bigg]  \int_{\mathcal{M}_{-j}}\mu(C, 1, M_j, m_{-j})\pi_{C, 0} (m_{-j})  \, \mathrm{d} m_{-j} \\
            % & + \bigg[ \frac{\mathds{1}(A = 1)}{e_1(C)} -  \frac{\mathds{1}(A = 0)}{e_0(C)} \bigg]  \tau_{C, 0 \, ; \, -j} (C, 1, M_j) +  \bigg[ \frac{\mathds{1}(A = 1)}{e_1(C)} \zeta_j (1, 0, C) - \frac{\mathds{1}(A = 0)}{e_0(C)} \zeta_j (0, 0, C) \bigg] \\
            % & + \zeta_j (1, 0, C) - \zeta_j (0, 0, C) - \E \Big[ \zeta_j (1, 0, C) - \zeta_j (0, 0, C)  \Big],
        \end{aligned}
    \]
    for any $j \in [p]$. Furthermore, for any fixed $\mathcal{G}_M$, we have
    \[
         \begin{aligned}
             & S^{\text{eff}, \text{nonpar}}  \big(TM_j(\mathcal{G}_M) \big) = S^{\text{eff}, \text{nonpar}} (TE)  - \Big\langle S^{\text{eff}, \text{nonpar}} \big( \E \varrho_j(\boldsymbol{\cdot}, M_j, C \, ; \, \mathcal{G}_M) \big)  \Big\rangle
             % - \Bigg[ \frac{\mathds{1} (A = 1)}{e_1(C)} \pi_{C, 1} \big(\Pa_j (\mathcal{G}_M) \big) \Big[ Y - \mu(C, 1, \Pa_j (\mathcal{G}_M), M_j)\Big] \\
             % & ~~~~~~~~~~~~~~~~~~~~~~~~~~~~~~~ - \frac{\mathds{1} (A = 0)}{e_0(C)} \pi_{C, 0 \, ; \, j} \big(\Pa_j (\mathcal{G}_M) \big) \Big[ Y - \mu(C, 0, \Pa_j (\mathcal{G}_M), M_j)\Big] \\
             % % & + \frac{\mathds{1} (A = 1)}{e_1(C)} \bigg( \int_{\mathcal{M}_j} \mu(C, 1, \Pa_j (\mathcal{G}_M), m_j)  \pi_{C}(m_j) \, \mathrm{d} m_j  - \E \big[ \varrho_j (1, M_j, C\, ; \, \mathcal{G}_M) \mid C \big] \bigg) \\
             % % & - \frac{\mathds{1} (A = 0)}{e_0(C)} \bigg( \int_{\mathcal{M}_j}  \mu(C, 0, \Pa_j (\mathcal{G}_M), m_j) \pi_{C}(m_j) \, \mathrm{d} m_j - \E \big[ \varrho_j (0, M_j, C \, ; \, \mathcal{G}_M) \mid C \big] \bigg) \\
             % & + \frac{\mathds{1} (A = 1)}{e_1(C)} \bigg[ \tau_{C \, ; \, j}\big( C, 1, \Pa_j (\mathcal{G}_M)\big)  - \E \big[ \varrho_j (1, M_j, C\, ; \, \mathcal{G}_M) \mid C \big] \bigg] \\
             % & - \frac{\mathds{1} (A = 0)}{e_0(C)} \bigg[ \tau_{C \, ; \, j} \big( C, 0, \Pa_j (\mathcal{G}_M) \big) - \E \big[ \varrho_j (0, M_j, C \, ; \, \mathcal{G}_M) \mid C \big] \bigg] \\
             % & + \Big\{ \varrho_j (1, M_j, C\, ; \, \mathcal{G}_M) - \varrho_j (0, M_j, C\, ; \, \mathcal{G}_M)  \Big\} - \E \Big[ \varrho_j (1, M_j, C\, ; \, \mathcal{G}_M) - \varrho_j (0, M_j, C\, ; \, \mathcal{G}_M) \mid C \Big] \\
             % & +  \E \Big[ \varrho_j (1, M_j, C\, ; \, \mathcal{G}_M) - \varrho_j (0, M_j, C\, ; \, \mathcal{G}_M) \mid C \Big] -  \E \Big[ \varrho_j (1, M_j, C\, ; \, \mathcal{G}_M) - \varrho_j (0, M_j, C\, ; \, \mathcal{G}_M)\Big] \Bigg],
         \end{aligned}
    \]
    and $S^{\text{eff}, \, \text{nonpar}} \big(IM_j (\mathcal{G}_M)\big) = S^{\text{eff}, \, \text{nonpar}} \big(TM_j (\mathcal{G}_M) \big) - S^{\text{eff}, \text{nonpar}} (DM_j)$ for any $j \in [p]$.
\end{corollary}

\noindent In Corollary \ref{cor_eff_DM_IM}, the explicit formulas for \( S^{\text{eff}, \text{nonpar}} (DM_j) \), \( S^{\text{eff}, \text{nonpar}} \big( TM_j(\mathcal{G}_M)\big) \), and \( S^{\text{eff}, \text{nonpar}} \big( IM_j(\mathcal{G}_M)\big) \) can be directly derived by substituting the relevant expressions from Theorem \ref{thm_EIF}. Consequently, the semiparametric efficiency bounds for estimating \( DM_j \), \( IM_j(\mathcal{G}_M) \), and \( TM_j(\mathcal{G}_M) \) within the full nonparametric model \( \mathscr{M}_{\text {nonpar}} \) are respectively \( \E \big[ S^{\text{eff, nonpar}} (DM_j)\big]^2 \), \( \E \big[ S^{\text{eff, nonpar}} (IM_j(\mathcal{G}_M))\big]^2 \), and \( \E \big[ S^{\text{eff, nonpar}} (TM_j(\mathcal{G}_M))\big]^2 \) for any specified \( \mathcal{G}_M \), all with clearly delineated forms.
The asymptotic variances of any regular asymptotic linear estimators in $\mathscr{M}_{\text{nonpar}}$ must be greater than or equal to these bounds. Given that $TM_j(\mathcal{G}_M) = DM_j + IM_j(\mathcal{G}_M)$, we will only focus on $DM_j$ and $IM_j(\mathcal{G}_M)$ in the subsequent sections. 

\section{Direct Strategy and Ordinary Least Squares (OLS) Estimations}\label{sec_straight_strategy}

An important implication of Corollary \ref{cor_eff_DM_IM} is that all regular and asymptotically linear (RAL) estimators of $DM_j$ and $IM_j(\mathcal{G}_M)$ in the model $\mathscr{M}_{\text{nonpar}}$ share the common score $S^{\text{eff, nonpar}} (DM_j)$ and $S^{\text{eff,  nonpar}} (IM_j(\mathcal{G}_M))$, respectively. For illustrating this and as a motivation for multiply robust estimation when nonparametric methods are not appropriate, we provide a detailed study of different estimating strategies in this section and the next section.

Theorem \ref{thm_direct_exp} gives an explicit expression for $DM_j$ and ${IM}_j(\mathcal{G}_M)$, we can correspondingly give their estimators by (i) replacing the unknown quantities $\kappa(\cdot, \cdot)$, $\zeta_j(\cdot, \cdot, \cdot)$, $\varrho_j(\cdot, \cdot, \cdot \, ; \, \mathcal{G}_M)$ with their estimators and then (ii) replacing $\E [\cdot]$ by $\pn [\cdot] = n^{-1} \sum_{i = 1}^n [\cdot]_i$ directly. To be specific, in the step (i), we construct the following estimators:
\[
    \widehat\kappa^{\mathscr{M}_0} (a', c) := \widehat{\mu} (c, a'),
\]
\[
    \widehat\zeta_j^{\mathscr{M}_0}(a', 0, c) := \int_{\mathcal{M}} \widehat{\mu}(C, 1, m) \widehat\pi_{C, a'}(m_j) \widehat\pi_{C, 0}(m_{-j}) \, \mathrm{d} m,
\]
and
\[
    \widehat\varrho_j^{\mathscr{M}_0} (a', m_j, c \, ; \, \mathcal{G}_M) := \int_{ \mathcal{M}_{\pa_j} (\mathcal{G}_M)} \widehat\mu \big(C, a', \pa_j (\mathcal{G}_M), M_j \big) \widehat\pi_{C, a'} \big(\pa_j(\mathcal{G}_M) \big) \, \mathrm{d} \pa_j (\mathcal{G}_M),
\]
which are consistent for $\kappa(a', C)$, $\zeta_j(a', 0, c)$ and $\varrho_j(a', m_j, c)$ for any $c \in \mathcal{C}$, $a' \in \{ 0,1 \}$, and $m \in \mathcal{M}$.
%as well as $\widehat{\mu} (C, a')$ is consistent for $\E[Y \mid A = a', C]$. 
%Here, the upper index "$\text{str}$" means "straightforward", which means we just replace the unknown quantities by the corresponding estimator in the definition of $\zeta_j (\cdot, \cdot, \cdot)$ and $\varrho_j (\cdot, \cdot, \cdot)$.
Note that $\kappa(a', C)$ can be written as 
\[
    \E [Y \mid C, A = a'] = \int_{\mathcal{M}} \E [Y \mid C, A = a', M = m] f_{M \mid C, A}(m \mid C, A = a') \, \mathrm{d} m.
\]
Therefore, the consistency of $\widehat\kappa^{\mathscr{M}_0} (a', c)$, $\widehat\zeta_j^{\mathscr{M}_0}(a', 0, c)$, and $\widehat\varrho_j^{\mathscr{M}_0} (a', m_j, c \, ; \, \mathcal{G}_M)$ will use the correctly specified information as follows:
\begin{itemize}
    \item $\mathscr{M}_0 $: the conditional expectation $\E [Y \mid C = \bcdot, A = \bcdot, M = \bcdot]$
    and the conditional density of the mediator $f_{M \mid C, A}(m \mid C = \bcdot, A = \bcdot)$ are correctly specified.
\end{itemize}
Then in the step (ii), we can construct the estimators
\[
    %\widehat{DM}_j^{\mathscr{M}_0} = \pn \Big[ \widehat\zeta_j^{\mathscr{M}_0} (1, 0, C) - \widehat\zeta_j^{\mathscr{M}_0}(0, 0, C) \Big],
    \widehat{DM}_j^{\mathscr{M}_0} = \pn \Big[ \big\langle \widehat\zeta_j^{\mathscr{M}_0} (\boldsymbol{\cdot}, 0, C) \big\rangle \Big],
\]
\[
    \widehat{TM}_j^{\mathscr{M}_0} (\mathcal{G}_M) =\pn \Big[ \big\langle \widehat\kappa^{\mathscr{M}_0} (\bcdot, C) -  \widehat\varrho_j^{\mathscr{M}_0} (\bcdot, M_j, C \, ; \, \mathcal{G}_M) \big\rangle \Big]
    %\widehat{TM}_j^{\mathscr{M}_0} (\mathcal{G}_M) =\pn \Big[ \big\{ \widehat\kappa^{\mathscr{M}_0} (1, C) -  \widehat\varrho_j^{\mathscr{M}_0} (1, M_j, C \, ; \, \mathcal{G}_M) \big\} - \big\{ \widehat\kappa^{\mathscr{M}_0} (0, C) -  \widehat\varrho_j^{\mathscr{M}_0} (0, M_j, C \, ; \, \mathcal{G}_M) \big\} \Big],
\]
and $\widehat{IM}_j^{\mathscr{M}_0} (\mathcal{G}_M) = \widehat{TM}_j^{\mathscr{M}_0} (\mathcal{G}_M) - \widehat{DM}_j^{\mathscr{M}_0}$. Next, the estimators for the identifiable $\overline{TM}_j$ and $\overline{IM}_j$ are 
\[
    \widehat{TM}_j^{\text{avg}, \mathscr{M}_0 } = \frac{1}{\# \operatorname{MEC}(\widehat{\mathcal{C}}_{ M})} \sum_{\mathcal{G}_{M} \in \operatorname{MEC}(\widehat{\mathcal{C}}_{ M})} \widehat{TM}_j^{\mathscr{M}_0} (\mathcal{G}_M),
\]
and $\widehat{IM}_j^{\text{avg}, \mathscr{M}_0} = \widehat{TM}_j^{\text{avg}, \mathscr{M}_0} - \widehat{DM}_j^{\mathscr{M}_0}$, where $\widehat{\mathcal{C}}_{M} \in \{ 0, 1\}^{p \times p}$ is the adjacency matrix of the estimated CPDAG for the mediators $M = (M_1, \ldots, M_p)^{\top} \in \mathbb{R}^p$. The consistent causal structure $\widehat{\mathcal{C}}_M$ can be obtained by after we obtain the estimated adjacency matrix $\widehat{\mathcal{C}}$ for the whole causal graph, and we then extract a subset $\widehat{\mathcal{C}}_M = \big[ \widehat{\mathcal{C}}_{kk} \big]_{k \in { t + 1, \ldots, t + p}}$ to arrive at the causal structure for the mediators. The estimated adjacency matrix  $\widehat{\mathcal{C}}$ can be achieved through methods such as the PC algorithm \citep{spirtes2000constructing}, greedy equivalence search (GES) \citep{chickering2002optimal}, and adaptively restricted greedy equivalence search (ARGES) \citep{nandy2018high}, among others.

\subsection{OLS estimator under semi-linear model}\label{sec_ols_est}

The direct strategy in model $\mathscr{M}_0$ involves two unknown quantities: $\E [Y \mid C = \bcdot, A = \bcdot, M = \bcdot]$ and $f_{M \mid C, A}(m \mid C = \bcdot, A = \bcdot)$. Specially, when we have known that the structure follows $Y \, \leftarrow \, h_Y(C, A, M) + \epsilon_{Y}$ and $M \, \leftarrow \, h_M(C, A, M) + \epsilon_{M}$\footnote{The symbol $\leftarrow$ emphasizes that the expressions should be understood
as a generating mechanism rather than as a mere equation.} with given functions $h_Y(\bcdot)$ and $h_M(\bcdot)$, and given that the error terms $\epsilon_{Y}$ and $\epsilon_M$ belong to some classes of distributions, $\mathscr{M}_0$ will be correctly recovered. 
% Here $\epsilon_{Y}$ and $\epsilon_{M}$ are error terms, and we assume we know the distribution of $\epsilon_{M}$. 
A commonly used approach for this is assuming Linear Structural Equation Models (LSEMs), i.e., $X \leftarrow B^{\top} X + \epsilon$ with mean-zero and jointly independent error vector $\epsilon$, where $B=(b_{ij})_{1\leq i\leq d,1\leq j\leq d}$ be a $d\times d$ matrix, where $b_{ij}$ is the weight of the edge $X_i\rightarrow X_j \in E$, and $b_{ij}=0$ otherwise. 
There are numerous rigorous theoretical findings for LSEMs, as discussed in \citep{chakrabortty2018inference,cai2020anoce,shi2022testing}.
However, LSEMs are not applicable when dealing with binary exposure, given that the element $A$ in $X$ is constrained to either $0$ or $1$. In lieu of LSEMs, we propose the following \textit{semi-linear} structure assumption.

\begin{ass}\label{ass_slsem}
    We assume $X$ is \textbf{semi-linear} when it is generated as follows
    
    \begin{equation}\label{lsem_x}
        \begin{aligned}
            A \, & \leftarrow \, h (C, \epsilon_A), \\
            M \, & \leftarrow \, B_{M C}^{\top} C + \beta_{MA} A + B_{MM}^{\top} M + \epsilon_{M}, \\
            Y \, & \leftarrow \, \beta_{YC}^{\top} C + \alpha_{YA} A + \beta_{YM}^{\top} M + \epsilon_{Y},
        \end{aligned}
    \end{equation}
    
    where $h : \mathbb{R}^{t - 1} \times \mathbb{R} \rightarrow \{ 0, 1\}$ is a known link function and $\epsilon_A, \epsilon_M, \epsilon_Y$ are mean-zero error terms independent with each other as well as $C$.
\end{ass}

Through this paper, $\alpha$, $\beta$, and $B$ will always represent a scalar, vector, and matrix, respectively. Define $\theta_{MA} := \big[ (I_p - B_{MM}^{\top})^{-1} \beta_{MA} \big]$ where $I_p \in \mathbb{R}^{p \times p}$ is the identity matrix, then under the above semi-linear structural assumptions, we have the following propositions for the uniqueness of $\theta_{MA}$ under MEC, interpretation displays, and neat {parametric} expressions for the causal effects defined in Section \ref{sec_def}.

\begin{pro}[Identification]\label{thm_unique}
    Under Assumption \ref{ass_slsem}, $\theta_{MA}$ is unique in any fixed $\operatorname{MEC}(\mathcal{C}_M)$. Hence, $DE$, $IE$, and $DM_j$ are also unique in $\operatorname{MEC}(\mathcal{C}_M)$.
\end{pro}

\begin{pro}[Interpretation]\label{pro_inter}
    Under Assumption \ref{ass_structure} and Assumption \ref{ass_slsem}, for any fixed $\mathcal{G}_M$, we have
    \[
        DM_j = \Big\{ \E \big[ Y \mid do (A = 0, M_j = m_j^{(0)} + 1, M_{-j} = m_{-j}^{(0)}) \big]  - \E \big[ Y \mid do (A = 0) \big] \Big\} \times \Delta_j^*,
    \]
    and 
    \[
        \begin{aligned}
            IM_j (\mathcal{G}_M) & = \Big\{ \E \big[ Y \mid do (A = 0, M_j = m_j^{(0)} + 1) \big] \\
            & ~~~~~~~~~~~~~~~~~~~~~~~~~~~~~~- \E \big[ Y \mid do (A = 0, M_j = m_j^{(0)} + 1, M_{-j} = m_{-j}^{(0)}) \big]  \Big\}  \times \Delta_j^*,
        \end{aligned}
    \]
    where $\Delta_j^* = \big\langle \E \big[ M_j \mid do (A = \bcdot) \big] \big\rangle$.
    %$\Delta_j^* = \E \big[ M_j \mid do (A = 1) \big] - \E \big[ M_j \mid do (A = 0) \big]$.
\end{pro}

\begin{pro}\label{thm_par_exp}
    Under Assumption \ref{ass_structure} and Assumption \ref{ass_slsem}, we have 
    (i):
    \[
        DE = \alpha_{YA}, \qquad IE = \theta_{MA}^{\top} \beta_{MA}
    \]
    and hence $ TE  = \alpha_{YA} + \theta_{MA}^{\top} \beta_{MA}$ for natural effects, and \[
        DM_j = \beta_{YM, j} \theta_{MA, j}.
    \]
    (ii) For any fixed $\mathcal{G}_M$,
    \[
        TM_j = \beta_{YM}^{\top} \theta_{MA} - \beta_{YM_{-j}}^{\top} \theta_{M_{-j}A},
    \]
    where $\theta_{M_{-j}A} := (I - B_{M_{-j}M_{-j}}^{\top})^{-1} \beta_{M_{-j}A}$, and hence
    \[
        IM_j = \beta_{YM, j}^{\top} (\theta_{MA, -j} - \theta_{M_{-j}A}).
    \]
    (ii') For any fixed $\mathcal{G}_M$, we have the alternative expression for $TM_j$ as 
    \[
        TM_j = \theta_{MA, j} \times \E^{\text{reg}} \big[ Y \mid M_j \cup \Pa_j (\mathcal{G}_{M}) \cup A \cup C \big]_1,
    \]
    hence
    \[
        IM_j = \theta_{MA, j} \times \Big( \E^{\text{reg}} \big[ Y \mid M_j \cup \Pa_j (\mathcal{G}_{M}) \cup A \cup C \big]_1 - \beta_{YM, j} \Big),
    \]
    where \(\E^{\text{reg}} [X_k \mid X_l \, \cup \, X_S]_1\) denotes the true coefficient of \(X_l\) in the linear regression of \(X_k\) on the combined set \(X_l \, \cup \, X_S\).
\end{pro}

\noindent Proposition \ref{thm_unique} gives the fact that only $IM_j$ and $TM_j$ require specific DAG structure, while other quantities do not require any knowledge of the causal structure under semi-linear assumption. 
%This confirms the reasonality of primary definitions of causal effects in Section \ref{sec_def}. 
Meanwhile, Proposition \ref{pro_inter} implies that, under the semi-linear assumption, our definitions for direct/indirect individual mediation effects in Definition \ref{def_med} exactly coincides with the definitions in \cite{cai2020anoce}: the first multiplier is in Proposition \ref{pro_inter} with the classical meaning of `natural' in the causal inference literature
\cite{pearl2000causality}. Thus, $DM_j$ can be interpreted as the causal effect through a particular mediator from the treatment on the outcome that is not regulated by its descendant mediators. Similarly, by the first multiplier in the $IM_j$, we know that $IM_j$ captures the indirect effect of a particular mediator on the outcome regulated by its descendant mediators.

More importantly, Proposition \ref{thm_unique} can imply a simply OLS estimator for the direct strategy together with Proposition \ref{thm_par_exp} as long as the sample size $n$ is larger than the dimension $d$. Indeed, we can rewrite the part of semi-linear structure \eqref{ass_slsem} as follows:
\begin{equation}\label{reg_exp}
    \left\{ \begin{array}{l}
            M = \underbrace{(I - B_{MM}^{\top})^{-1}B_{MC}^{\top}}_{= :\Theta_{MC}} C + \underbrace{(I - B_{MM}^{\top})^{-1} \beta_{MA}}_{\theta_{MA}} A + \underbrace{(I - B_{MM}^{\top})^{-1}\epsilon_M}_{=: e_M}, \\
            Y =  (\beta_{YC}^{\top}, \alpha_{YA}, \beta_{YM}^{\top}) (C^{\top}, A, M^{\top})^{\top}+ \epsilon_Y
    \end{array} \right..
\end{equation}
Write $\widehat{\theta}_{MA}$ as the OLS estimator of unknown parameter ${\theta}_{MA}$, similarly define the other corresponding estimated quantities as follows:
\[
    \left [ \begin{array}{c}
                 \widehat{\Theta}_{MC}^{\top} \\
                 \widehat{\theta}_{MA}^{\top} 
            \end{array}\right] \qquad \text{ and } \qquad \left[ \begin{array}{c}
                \widehat{\beta}_{YC} \\
                \widehat{\alpha}_{YA} \\
                \widehat{\beta}_{YM}
            \end{array}\right].
\]
Then we will have OLS estimators for the direct strategy estimators: For $DE$ and $IE$, $\widehat{DE}^{\text{OLS}} = \widehat{\alpha}_{YA}$ and $\widehat{IE}^{\text{OLS}} = \widehat{\beta}_{YM}^{\top} \widehat\theta_{MA}$; For $DM_j$ and $IM_j$, $\widehat{DM}_j^{\text{OLS}} = \widehat{\beta}_{YM, j}\widehat\theta_{MA, j}$, $\widehat{{IM}}_j^{\text{OLS}} (\mathcal{G}_{M}) = \widehat{\theta}_{MA, j} \big\{ \widehat{\E}^{\text{reg}} \big[ Y \mid M_j \cup \Pa_j (\mathcal{G}_{M}) \cup A \cup C \big]_1 - \widehat{\beta}_{YM, j} \big\}$, and
\[
    \widehat{{IM}}_j^{\text{avg}, \, \text{OLS}} = \frac{1}{\# \operatorname{MEC}(\widehat{\mathcal{C}}_{ M})} \sum_{\mathcal{G}_{M} \in \operatorname{MEC}(\widehat{\mathcal{C}}_{ M})} \widehat{{IM}}_j^{\text{OLS}} (\mathcal{G}_{M}),
\]
where $\widehat\E^{\text{reg}} [X_k \mid X_l \, \cup \, X_S]_1$ is the estimated coefficient for $X_l$ obtained from the linear regression of $X_k$ on $X_l \, \cup \, X_S$, as determined from the data.
Thus, when the semi-linear structure is determined, we can simplify direct strategy estimators to OLS estimators. 
All the these OLS estimators can be easily obtained by just applying simple regressions with nice properties, we will discuss their asymptotic properties in Section \ref{sec_asy}. 

\section{Multiple Robust Estimators}\label{sec_mr}

\subsection{Several Alternative Strategies}
   
% \todohw{这节有个重要的问题，对于每个 j，我们这里的 model
% specifications $\mathscr{M}_0$, $\mathscr{M}_{j, 1}$, $\mathscr{M}_{j, 2}$, $\mathscr{M}_{j, 3}$ 是不相交的。但是如果加入了不同的 j，就存在相交的问题，比如 $\mathscr{M}_{j, 1}$ 是正确识别的，那么 $\mathscr{M}_{k, 2}$ 就是正确识别的，对于 $k \neq j$。我目前的方法是单独对每个 $j$ 单独建立 $\mathscr{M}_{j, \ell}$。$\mathscr{M}_0$ 是对所有 $j \in [p]$ 都成立的。}

For a fixed $j \in [p]$, beyond the direct strategy above, there are alternative identification formulas for $\E \kappa(a', C)$, $\E \zeta_j (a', 0, C)$, and $\E \varrho_j (a', M_j, C \, ; \, \mathcal{G}_M)$. Based on these formulations, we can derive the corresponding estimators. We will discuss them one by one in the subsequent sections. 

\subsubsection{Alternative Strategy 1}\label{sub_sec_alternative_strategy_1}

The first one is using propensity score to construct the inverse probability weighting estimator. Note that we have\footnote{The calculation details are shown in  \ref{proof_alter_strategies}.}
\[
    \E \bigg[ \frac{\mathds{1}(A = a')}{e_{a'}(C)} Y\bigg] = \E \kappa(a', C),
\]
\begin{equation}\label{strategy_1_DM}
    \begin{aligned}
        \E \bigg[ \frac{\mathds{1}(A = 1)}{e_1(C)} \frac{ \pi_{C, a'}(M_{-j})}{\pi_{C, 1, M_j}(M_{-j}) } Y \bigg] = \E \zeta_j (a', 0, C),
    \end{aligned}
\end{equation}
and
\begin{equation}\label{strategy_1_IM}
    \begin{aligned}
        \E \bigg[ \frac{\mathds{1}(A = a')}{e_{a'}(C)} \pi_{C, a'} \big(\Pa_j (\mathcal{G}_M) \big) Y \bigg] = \E \varrho_j (a', M_j, C \, ; \, \mathcal{G}_M).
    \end{aligned}
\end{equation}
Thus, corresponding estimators take the form
\[
    \widehat\kappa^{\mathscr{M}_1} (a', C) = \frac{\mathds{1}(A = a')}{\widehat{e}_{a'} (C)} Y,
\]
\[
    \begin{aligned}
        \widehat\zeta_j^{\mathscr{M}_1} (a', 0, C) = \frac{\mathds{1}(A = 1)}{\widehat{e}_1(C)} \frac{\widehat{\pi}_{C, a'}(M_{-j})}{\widehat\pi_{C, 1, M_j}(M_{-j})} Y,
    \end{aligned}
\]
and
\[
    \widehat\varrho_j^{\mathscr{M}_1} (a', M_j, C \, ; \, \mathcal{G}_M) = \frac{\mathds{1}(A = a')}{\widehat{e}_{a'}(C)} \widehat\pi_{C, a'} \big(\Pa_j (\mathcal{G}_M) \big) Y,
\]
respectively. Here, the propensity scores $e_{a'}(\bcdot)$ and conditional densities $\pi_{\bcdot}(\bcdot)$ appearing in $\widehat\kappa^{\mathscr{M}_1} (a', C)$, $\widehat\zeta_j^{\mathscr{M}_1} (a', 0, C)$, and $\widehat{\varrho}_j^{\mathscr{M}_1} (a', M_j, C \, ; \, \mathcal{G}_M)$ should be correctly estimated in the collection of quantities $\mathscr{M}_{j, \, 1}$ such that 
\begin{itemize}
    \item $\mathscr{M}_{j, \, 1}$: The propensity scores $\pr (A = \bcdot \mid C = \bcdot)$, and the conditional density of the mediator $f_{M_{-j} \, \mid \, C, A}(\bcdot \mid C = \bcdot, A = \bcdot)$ and $f_{M_{-j} \, \mid \, C, A, M_j} (\bcdot \mid C = \bcdot, A = 1, M_j = \bcdot)$ are correctly specified for any $\mathcal{G}_M \in \operatorname{MEC}(\mathcal{C}_{M})$.
\end{itemize}
Define
% \begin{equation}\label{def_kappa_0_1}
%     \pn \widehat{\kappa}_{1-0}^{\mathscr{M}_1}(C) := \pn \big[ \widehat\kappa^{\mathscr{M}_1} (1, C) - \widehat\kappa^{\mathscr{M}_1} (0, C)  \big] 
% \end{equation}
% and
\begin{equation}\label{def_varrho_0_1}
    \begin{aligned}
        & \overline{\pn} \widehat\varrho_{j \, ; \, 1- 0}^{\mathscr{M}_1} (M_j, C \, ; \, \mathcal{C}_M) = \frac{1}{\# \operatorname{MEC}({\mathcal{C}}_{ M})} \sum_{\mathcal{G}_{M} \in \operatorname{MEC}({\mathcal{C}}_{ M})} \pn \Big[\big\langle  \widehat\varrho_j^{\mathscr{M}_1} (\bcdot, M_j, C \, ; \, \mathcal{G}_M) \big\rangle \Big],
    \end{aligned}
\end{equation}
then, we can construct the estimators under $\mathscr{M}_{j, \, 1}$ is $\widehat{DM}_j^{\mathscr{M}_1} = \pn \Big[ \big\langle \widehat\zeta_j^{\mathscr{M}_1} (\bcdot, 0, C) \big\rangle  \Big]$,
\[
    \begin{aligned}
        \widehat{TM}_j^{\text{avg}, \, \mathscr{M}_1} & = \frac{1}{\# \operatorname{MEC}(\widehat{\mathcal{C}}_{ M})} \sum_{\mathcal{G}_{M} \in \operatorname{MEC}(\widehat{\mathcal{C}}_{ M})} \widehat{TM}_j^{\mathscr{M}_1} (\mathcal{G}_{M} ) \\
        & = \pn \big[\big\langle \widehat\kappa^{\mathscr{M}_1} (\bcdot, C) \big\rangle \big] - \overline{\pn} \widehat\varrho_{j \, ; \, 1- 0}^{\mathscr{M}_1} (M_j, C \, ; \, \widehat{\mathcal{C}}_M),
        % & = \pn \big[ \widehat\kappa^{\mathscr{M}_1} (1, C) - \widehat\kappa^{\mathscr{M}_1} (0, C)  \big] \\
        % & ~~ - \frac{1}{\# \operatorname{MEC}(\widehat{\mathcal{C}}_{ M})} \sum_{\mathcal{G}_{M} \in \operatorname{MEC}(\widehat{\mathcal{C}}_{ M})} \pn \Big[  \widehat\varrho_j^{\mathscr{M}_1} (1, M_j, C \, ; \, \mathcal{G}_M) -  \widehat\varrho_j^{\mathscr{M}_1} (0, M_j, C \, ; \, \mathcal{G}_M)  \Big],
    \end{aligned}
\]
and $\widehat{IM}_j^{\text{avg}, \mathscr{M}_1} = \widehat{TM}_j^{\text{avg}, \mathscr{M}_1} - \widehat{DM}_j^{\mathscr{M}_1}$, where $\widehat{\mathcal{C}}_M$ is the estimated adjacency matrix
consistent with the true $\mathcal{C}_M$.
Here, the superscript $\mathscr{M}_1$ associated with these estimators signifies that their consistency relies on the correct specification of information in $\mathscr{M}_{j, \, 1}$. For clarity and where there is no risk of confusion, we will also employ $\mathscr{M}_1$ to represent the estimation methodology behind these estimators. In the subsequent two subsections, the notations $\mathscr{M}_{2}$ and $\mathscr{M}_{3}$ bear analogous meanings.

\subsubsection{Alternative Strategy 2}\label{sub_sec_alternative_strategy_2}
Similarly, we can verify that 
\begin{equation}\label{strategy_2_DM}
    \begin{aligned}
        \E \bigg[ \frac{\mathds{1}(A = 0)}{e_0(C)} \int_{\mathcal{M}_j} \mu(C, 1, m_j, M_{-j}) \pi_{a', C}(m_j) \, \mathrm{d} m_j \bigg] = \E \zeta_j (a', 0, C)
    \end{aligned}
\end{equation}
and 
\begin{equation}\label{strategy_2_IM}
    \begin{aligned}
         \E \bigg[ \frac{\mathds{1}(A = a')}{e_{a'} (C)}  \int_{\mathcal{M}_j} \mu(C, a',  \Pa_j(\mathcal{G}_M), m_j) \pi_{C}(m_j) \, \mathrm{d} m_j \bigg] = \E \varrho_j (a', M_j, C \, ; \, \mathcal{G}_M).
    \end{aligned}
\end{equation}
Thus, corresponding estimators take the forms
\[
    \begin{aligned}
        \widehat\zeta_j^{\mathscr{M}_2} (a', 0, C)  = \frac{\mathds{1}(A = 0)}{\widehat{e}_0(C)} \int_{\mathcal{M}_j} \widehat{\mu}(C, 1, m_j, M_{-j}) \widehat{\pi}_{C, a'}(m_j)\, \mathrm{d} m_j,
    \end{aligned}
\]
and
\[
    \widehat{\varrho}_j^{\mathscr{M}_2} (a', M_j, C \, ; \, \mathcal{G}_M) = \frac{\mathds{1}(A = a')}{\widehat{e}_{a'} (C)} \int_{\mathcal{M}_j}\widehat{\mu}(C, a', \Pa_j(\mathcal{G}_M), m_j) \widehat{\pi}_C(m_j) \, \mathrm{d} m_j
\]
with estimators $\widehat{e}_{\bcdot} (\bcdot)$, $\widehat\mu(\bcdot)$, and $\widehat\pi_{\bcdot} (\bcdot)$ appear in $\widehat\kappa^{\mathscr{M}_1} (a', C)$, $\widehat\zeta_j^{\mathscr{M}_2} (a', 0, C)$, and $\widehat{\varrho}_j^{\mathscr{M}_2} (a', M_j, C \penalty 0 \, ; \, \mathcal{G}_M)$. They use the information in $\mathscr{M}_{j, 2}$ such that
\begin{itemize}
    \item $\mathscr{M}_{j, \, 2}$: The propensity scores $\pr (A = \bcdot \mid C = \bcdot)$, conditional density of $j$-th mediator $f_{M_j \mid C, A}(\bcdot \mid C = \bcdot, A = \bcdot)$, and the conditional expectations $\E [Y \mid C = \bcdot, A = \bcdot, M = \bcdot]$ and $\E [Y \mid C = \bcdot, A = \bcdot, \Pa_{j} (\mathcal{G}_M) = \bcdot, M_j = \bcdot]$ are correctly specified for any $\mathcal{G}_M \in \operatorname{MEC}(\mathcal{C}_M)$.
\end{itemize}
\noindent Here we notice the fact that $\pi_{c}(m_j) = \pi_{c, 0}(m_j) + \pi_{c, 1}(m_j)$, and thus, $\widehat{\kappa}^{\mathscr{M}_1}(a', C)$ will also be consistent in $\mathscr{M}_{j, \, 2}$.
Then $\widehat{DM}_j^{\mathscr{M}_2} = \pn \big[ \big\langle \widehat\zeta_j^{\mathscr{M}_2} (\bcdot, 0, C) \big\rangle \big]$, 
\[
    \begin{aligned}
        \widehat{TM}_j^{\text{avg}, \, \mathscr{M}_2} = \pn \big[\big\langle \widehat\kappa^{\mathscr{M}_1} (\bcdot, C) \big\rangle \big] - \overline{\pn} \widehat\varrho_{j \, ; \, 1- 0}^{\mathscr{M}_2} (M_j, C \, ; \, \widehat{\mathcal{C}}_M),
        % & = \frac{1}{\# \operatorname{MEC}(\widehat{\mathcal{C}}_{ M})} \sum_{\mathcal{G}_{M} \in \operatorname{MEC}(\widehat{\mathcal{C}}_{ M})} \widehat{TM}_j^{\mathscr{M}_2} (\mathcal{G}_{M} ) \\
        % & = \pn \big[ \widehat\kappa^{\mathscr{M}_1} (1, C) - \widehat\kappa^{\mathscr{M}_1} (0, C)  \big] \\
        % & ~~ - \frac{1}{\# \operatorname{MEC}(\widehat{\mathcal{C}}_{ M})} \sum_{\mathcal{G}_{M} \in \operatorname{MEC}(\widehat{\mathcal{C}}_{ M})} \pn \Big[  \widehat\varrho_j^{\mathscr{M}_2} (1, M_j, C \, ; \, \mathcal{G}_M) -  \widehat\varrho_j^{\mathscr{M}_2} (0, M_j, C \, ; \, \mathcal{G}_M)  \Big],
    \end{aligned}
\]
and  \( \widehat{IM}_j^{\text{avg}, \, \mathscr{M}_2} = \widehat{TM}_j^{\text{avg}, \, \mathscr{M}_2} - \widehat{DM}_j^{\mathscr{M}_2} \) are consistent provided that the estimated adjacency matrix \( \widehat{\mathcal{C}}_M \) is consistent to \( \mathcal{C}_{M} \). \( \overline{\pn} \widehat\varrho_{j \, ; \, 1- 0}^{\mathscr{M}_2} (M_j, C \, ; \, \widehat{\mathcal{C}}_M) \) is similarly defined by substituting \( \mathscr{M}_1 \) with \( \mathscr{M}_2 \) as detailed in \eqref{def_varrho_0_1}.

\subsubsection{Alternative Strategy 3}\label{sub_sec_alternative_strategy_3}

The last strategy is based on the third representation of the functional as follows:
\begin{equation}\label{strategy_3_DM}
     \begin{aligned}
         \E \bigg[ \frac{\mathds{1}(A = a')}{e_{a'} (C)}  \int_{\mathcal{M}_{-j}} \mu(C, 1, M_j, m_{-j}) \pi_{C, 0}(m_{-j}) \, \mathrm{d} m_{-j} \bigg] = \E \zeta_j(a', 0, C),
     \end{aligned}
\end{equation}
and
\begin{equation}\label{strategy_3_IM}
     \begin{aligned}
         & \E \Big[ \E \big[ \varrho_j (a', M_j, C \, ; \, \mathcal{G}_M ) \mid C \big] \Big] = \E \varrho_j (a', M_j, C \, ; \, \mathcal{G}_M ). 
         % = & \E \Bigg[ \int_{\mathcal{M}_{\pa_j} (\mathcal{G}_M)\, \cup \, \mathcal{M}_j} \E \big[ Y \mid A = a', M_j = m_j, {\Pa}_j (\mathcal{G}_M) = \pa_j, C \big]  f(\pa_j \mid a', C) f(m_j \mid C)  \, \mathrm{d} \pa_j \, \mathrm{d} m_j \Bigg] \\
     \end{aligned}
\end{equation}
Similarly, we can consider the estimators
\[
    \widehat{\zeta}_j^{\mathscr{M}_3}(a', 0, C) = \frac{\mathds{1}(A = a')}{\widehat{e}_{a'}(C)} \int_{\mathcal{M}_{-j}} \widehat{\mu} (C, 1, M_{j}, m_{-j}) \widehat\pi_{C, 0} (m_{-j})\, \mathrm{d} m_{-j},
\]
and
\[
    \begin{aligned}
        & \widehat{\varrho}_j^{\mathscr{M}_3} (a', M_j, C \, ; \, \mathcal{G}_M) = \widehat\E \big[ \varrho_j (a', M_j, C \, ; \, \mathcal{G}_M ) \mid C \big] \\
        = & \int_{\mathcal{M}_{\pa_j} (\mathcal{G}_M)\, \cup \, \mathcal{M}_j} \widehat{\mu} \big(C, a', \pa_j(\mathcal{G}_M), m_j \big) \widehat\pi_{C, a'} \big(\pa_j(\mathcal{G}_M) \big) \widehat\pi_{C} (m_j)\, \mathrm{d} \pa_j (\mathcal{G}_M) \, \mathrm{d} m_j.
    \end{aligned}
\]
% \blue{
% \begin{equation}\label{strategy_3_eq_1}
%     \begin{aligned}
%         \E \bigg[ \frac{\mathds{1}(A = 1)}{e_1(C)} \frac{ \pi_{C, a'}(M_{-j})}{\pi_{C, 1, M_j}(M_{-j}) } Y \bigg] = \E \big[ \zeta_j (a', 0, C)\big],
%     \end{aligned}
% \end{equation}
% and
% \[
%     \begin{aligned}
%         \E \bigg[ \frac{\mathds{1}(A = a')}{e_{a'}(C)} \frac{\pi_{C}(M_j)}{\pi_{C, a',  \Pa_j (\mathcal{G}_M)}(M_j)} Y \bigg] = \E \big[ \varrho_j (a', M_j, C \, ; \, \mathcal{G}_M)\big].
%     \end{aligned}
% \]
% Thus, our last estimators take the form:
% \[
%     \begin{aligned}
%         \widehat\zeta_j^{\mathscr{M}_3} (a', 0, C) = \frac{\mathds{1}(A = 1)}{\widehat{e}_1(C)} \frac{\widehat{\pi}_{C, a'}(M_{-j})}{\widehat\pi_{C, 1, M_j}(M_{-j})} Y 
%     \end{aligned}
% \]
% and
% \[
%     \widehat\varrho_j^{\mathscr{M}_3} (a', M_j, C \, ; \, \mathcal{G}_M) = \frac{\mathds{1}(A = a')}{\widehat{e}_{a'}(C)} \frac{\widehat\pi_{C}(M_j)}{\widehat\pi_{C, a',  \Pa_j (\mathcal{G}_M)}(M_j)} Y.
% \]
% }
Thus, our estimators under the third identification formulas can be written as $\widehat{DM}_j^{\mathscr{M}_3} = \pn \big[ \widehat\zeta_j^{\mathscr{M}_3} (1, 0, C) - \widehat\zeta_j^{\mathscr{M}_3} (0, 0, C) \big]$, 
\[
    \begin{aligned}
        \widehat{TM}_j^{\text{avg}, \, \mathscr{M}_3} =\pn \big[\big\langle \widehat\kappa^{\mathscr{M}_1} (\bcdot, C) \big\rangle \big] - \overline{\pn} \widehat\varrho_{j \, ; \, 1- 0}^{\mathscr{M}_3} (M_j, C \, ; \, \widehat{\mathcal{C}}_M),
        % = \frac{1}{\# \operatorname{MEC}(\widehat{\mathcal{C}}_{ M})} \sum_{\mathcal{G}_{M} \in \operatorname{MEC}(\widehat{\mathcal{C}}_{ M})} \widehat{TM}_j^{\mathscr{M}_3} (\mathcal{G}_{M} ) \\
        % & = \pn \big[ \widehat\kappa^{\mathscr{M}_1} (1, C) - \widehat\kappa^{\mathscr{M}_1} (0, C)  \big] \\
        % & ~~ - \frac{1}{\# \operatorname{MEC}(\widehat{\mathcal{C}}_{ M})} \sum_{\mathcal{G}_{M} \in \operatorname{MEC}(\widehat{\mathcal{C}}_{ M})} \pn \Big[  \widehat\varrho_j^{\mathscr{M}_3} (1, M_j, C \, ; \, \mathcal{G}_M) -  \widehat\varrho_j^{\mathscr{M}_3} (0, M_j, C \, ; \, \mathcal{G}_M)  \Big],
    \end{aligned}
\]
and $\widehat{IM}_j^{\text{avg}, \, \mathscr{M}_3} = \widehat{TM}_j^{\text{avg}, \, \mathscr{M}_3} - \widehat{DM}_j^{\mathscr{M}_3}$, where the estimated adjacency matrix $\widehat{\mathcal{C}}_M$ is consistent to $\mathcal{C}_M$, and \( \overline{\pn} \widehat\varrho_{j \, ; \, 1- 0}^{\mathscr{M}_3} (M_j, C \, ; \, \widehat{\mathcal{C}}_M) \) is  by replacing \( \mathscr{M}_1 \) with \( \mathscr{M}_3 \) in \eqref{def_varrho_0_1}. The estimators $\widehat{e}_{\bcdot} (\bcdot)$, $\widehat\mu(\bcdot)$, and $\widehat\pi_{\bcdot} (\bcdot)$ use the following information:
\begin{itemize}
    \item $\mathscr{M}_{j, \, 3}$: The propensity scores $\pr(A = \bcdot \mid C = \bcdot)$, the conditional densities $f_{M_{-j} \mid C, A}(\bcdot \mid C = \bcdot, A = \bcdot)$ and $f_{M_j \mid C}(\bcdot \mid C = \bcdot)$, and the conditional expectation $\E [Y \mid C = \bcdot, A = \bcdot, M = \bcdot]$ are correctly specified for any $\mathcal{G}_M \in \operatorname{MEC}(\mathcal{C}_M)$.
\end{itemize}
Here we note the fact that $\widehat{\kappa}^{\mathscr{M}_1}(a', C)$ will be consistent in $\mathscr{M}_{j, \, 3}$ again.

\subsection{Quadruply Robust Estimator}

Denote $\mathscr{M}_{j, \, \text{union}} := \mathscr{M}_0 \, \cup \, \mathscr{M}_{j, \, 1} \, \cup \, \mathscr{M}_{j, \, 2} \, \cup \, \mathscr{M}_{j, \, 3}$, then $\cup_{j = 1}^p \, \mathscr{M}_{j, \, \text{union}} \subsetneq \mathscr{M}_{\text{nonpar}}$, and $\widehat{DM}_j^{\mathscr{M}_{\ell}}$ and $\widehat{IM}_j^{\mathscr{M}_{ \ell}} (\mathcal{G}_M)$ are all mapping the estimated distribution $\widehat{F}_X$ to the true $DM_j$ and ${IM}_j (\mathcal{G}_M)$ defined in Definition \ref{def_med} for $\ell = 0, 1, 2, 3$ and any fixed $\mathcal{G}_M$,
since all these representations agree on the nonparametric model $\mathscr{M}_{\text{nonpar}}$. Therefore, we may conclude that both direct strategy and alternative strategies are in fact asymptotically efficient in $\mathscr{M}_{\text{nonpar}}$ with common scores $S^{\text {eff, nonpar }}(DM_j)$ and $S^{\text {eff, nonpar }} \big(IM_j (\mathcal{G}_M)\big)$. Furthermore, from this observation, one further concludes that (asymptotic) inferences obtained using one of the four representations are identical to inferences using either of the other three representations for a fixed $\mathcal{G}_M$. 
However, to achieve this, each strategy need exactly correctly specified for the conditional expectation and conditional density, i.e., correctly specified for the corresponding collections $\mathscr{M}_{j, \ell}$ in Section \ref{sec_straight_strategy}, \ref{sub_sec_alternative_strategy_1}, \ref{sub_sec_alternative_strategy_2}, and \ref{sub_sec_alternative_strategy_3}, where we denote $\mathscr{M}_{j, \, 0} \equiv \mathscr{M}_0$ for each $j \in [p]$. In general, $\widehat{DM}_j^{\mathscr{M}_{\ell}}$, $\widehat{IM}_j^{\mathscr{M}_{\ell}} (\mathcal{G}_M)$ fail to be consistent outside of the corresponding submodel $\mathscr{M}_{j, \, \ell}$ for each $\ell \in \{ 0, 1, 2, 3\}$.

Note that the alternative strategy 1 in Section \ref{sub_sec_alternative_strategy_1} in $\mathscr{M}_{j, \, 1}$ induces Inverse Probability Weighted (IPW) estimator. A commonly-used method is combining the direct strategy estimator in Section \ref{sec_straight_strategy} correctly specified with the model $\mathscr{M}_0$ and IPW estimator in \ref{sub_sec_alternative_strategy_1} with the model $\mathscr{M}_0$, and getting the double robust estimator.
But the double robust estimators only combine two estimation strategies, $\mathscr{M}_0$ and $\mathscr{M}_{j, \, 1}$, and ignore use other two alternative strategies. Hence, the double robust estimator may be inconsistent outside of $\mathscr{M}_0 \, \cup \, \mathscr{M}_{j, 1}$.   
To overcome this problem, we propose an approach that produces a quadruply robust estimator by \textit{combining the above all four strategies} as follows: 
\begin{itemize}
    \item $\widehat{DM}_j^{\text{QR}}$ solves
    \[
        \pn \widehat{S}^{\text{eff, nonpar}} (\widehat{DM}_j^{\text{QR}}) = 0;
    \]
    % \item For a fixed DAG $\mathcal{G}_M$, $\widehat{TM}_j^{\text{QR}}(\mathcal{G}_M)$ solves
    % \[
    %     \pn \widehat{S}^{\text{eff, nonpar}} (\widehat{TM}_j^{\text{QR}}(\mathcal{G}_M)) = 0;
    % \]

    \item For a fixed DAG $\mathcal{G}_M$, $\widehat{IM}_j^{\text{QR}}(\mathcal{G}_M)$ solves
    \[
        \pn \widehat{S}^{\text{eff, nonpar}} (\widehat{IM}_j^{\text{QR}}(\mathcal{G}_M)) = 0,
    \]
\end{itemize}
where $\widehat{S}^{\text{eff, nonpar}} (\bcdot)$ is equal to ${S}^{\text{eff, nonpar}} (\bcdot)$ evaluated at the given consistent estimators $\widehat{e}_{a'}(\bcdot)$, $\widehat{\pi}_{\bcdot}(\bcdot)$, and $\widehat{\mu}(\bcdot)$ for all propensity scores, the conditional densities, and the conditional expectations appearing in ${S}^{\text{eff, nonpar}} (\bcdot)$. Denote
\[
    \widehat{\tau}_{\boldsymbol{\cdot} \, ; \, S} (C, a', M_{T}) := \int_{\mathcal{M}_{-T}} \widehat\mu (C, a', m_{-T}, M_{T}) \widehat\pi_{\boldsymbol{\cdot}} (m_{-T}) \, \mathrm{d} m_{-T},
\]
as the corresponding estimator for ${\tau}_{\boldsymbol{\cdot} \, ; \, S} (C, a', M_{T})$ defined in \eqref{def_tau},
then we have the following explicit expressions for the quadruply estimators as
\begin{equation}\label{DM_MR}
    \begin{aligned}
        & \widehat{DM}_j^{\text{QR}} := \widehat{DM}_j^{\mathscr{M}_0} + \mathbb{P}_n \Bigg[ \frac{\mathds{1}(A = 1)}{\widehat{e}_1(C)}  \frac{\langle \widehat\pi_{C, \boldsymbol{\cdot}}(M_{-j}) \rangle}{\widehat\pi_{C, 1, M_j}(M_{-j})} \Big[ Y - \widehat\mu (C, 1, M) \Big] \\
        & ~~~~~ + \frac{\mathds{1}(A = 0)}{\widehat{e}_0(C)} \left\langle \widehat\tau_{C, \boldsymbol\cdot \, ; \, j}(C, 1, M_{-j})  - \widehat{\zeta}_j^{\mathscr{M}_0} (\boldsymbol\cdot,  0, C) \right\rangle \\
        & ~~~~~ + \left\langle \frac{\mathds{1}(A = \boldsymbol\cdot)}{\widehat{e}_{\boldsymbol\cdot}(C)} \right\rangle \widehat{\tau}_{C, 0 \, ; \, -j} (C, 1, M_{j}) - \left\langle \frac{\mathds{1}(A = \boldsymbol{\cdot})}{\widehat{e}_{\boldsymbol{\cdot}}(C)} \widehat\zeta_j^{\mathscr{M}_0} (\boldsymbol{\cdot}, 0, C) \right\rangle\Bigg],
    \end{aligned}
\end{equation}
\begin{equation}\label{IM_MR}
    \begin{aligned}
        & \widehat{TM}_j^{\text{QR}} (\mathcal{G}_M) := \widehat{TM}_j^{\mathscr{M}_0} + 
        \mathbb{P}_n \Bigg[ \left\langle \frac{\mathds{1}(A = \boldsymbol{\cdot})}{\widehat{e}_{\boldsymbol{\cdot}}(C)} \Big( Y - \widehat{\mu} (C, \boldsymbol{\cdot}) \Big) \right\rangle \Bigg] \\
        & - \mathbb{P}_n \Bigg[ \left\langle \frac{\mathds{1} (A = \boldsymbol{\cdot})}{\widehat{e}_{\boldsymbol{\cdot}}(C)} \widehat\pi_{C, \boldsymbol{\cdot}} \big(\Pa_j (\mathcal{G}_M) \big) \Big[ Y - \widehat{\mu} (C, \boldsymbol{\cdot}, \Pa_j(\mathcal{G}_M), M_j)\Big] \right\rangle  \\
        & ~~~~~~~~~~~~~ + \left\langle \frac{\mathds{1} (A = \boldsymbol{\cdot})}{\widehat{e}_{\boldsymbol{\cdot}}(C)} \bigg( \widehat\tau_{C \, ; \, j} \big( C, \boldsymbol{\cdot}, \Pa_j (\mathcal{G}_M) \big) - \widehat{\E} \big[ \varrho_j (\boldsymbol{\cdot}, M_j, C\, ; \, \mathcal{G}_M) \mid C \big] \bigg) \right\rangle \Bigg],
    \end{aligned}
\end{equation}
and $\widehat{IM}_j^{\text{QR}} (\mathcal{G}_M) = \widehat{TM}_j^{\text{QR}} (\mathcal{G}_M) - \widehat{DM}_j^{\text{QR}}$. Then the quadruply estimator for indirect interventional mediation effect with a consistent estimated $\widehat{\mathcal{C}}_M$ is defined as
%are $\widehat{TM}_j^{\text{avg}, \, \text{QR}} = \frac{1}{\# \operatorname{MEC}(\widehat{\mathcal{C}}_{ M})} \sum_{\mathcal{G}_{M} \in \operatorname{MEC}(\widehat{\mathcal{C}}_{ M})}  \widehat{TM}_j^{\text{QR}} (\mathcal{G}_M)$ and 
$\widehat{IM}_j^{\text{avg}, \, \text{QR}} = \frac{1}{\# \operatorname{MEC}(\widehat{\mathcal{C}}_{M})} \sum_{\mathcal{G}_{M} \in \operatorname{MEC}(\widehat{\mathcal{C}}_{ M})}  \penalty 0 \widehat{IM}_j^{\text{QR}} (\mathcal{G}_M)$. Compared to double robust estimators, our novel quadruply robust estimators can tolerate a higher degree of misspecification outside of $\mathscr{M}_0 \, \cup \, \mathscr{M}_{j, \, 1}$ and still achieve consistency. We will see this in Section \ref{sec_asy_mr}.

Subject to some mild regularity conditions, delineated in  Section \ref{sec_asy_mr}, our quadruply estimators are asymptotic normal and efficient. Thus, based on the semiparametric efficient scores, we get the score-based variance estimators for $\widehat{DM}_j^{\text{QR}}$ and $\widehat{IM}_j^{\text{avg}, \, \text{QR}}$ as
\[
    \widehat\var ( \widehat{DM}_j^{\text{QR}} ) := \frac{1}{n^2} \sum_{i = 1}^n \left[  \widehat{S}^{\text{eff, nonpar}} (\widehat{DM}_j^{\text{QR}}) - \widehat{DM}_j^{\text{QR}}\right]^2
\]
and
\[
    \begin{aligned}
        & \widehat\var ( \widehat{IM}_j^{\text{avg}, \, \text{QR}} ) \\
        & := \frac{1}{n^2} \sum_{i = 1}^n \left[ \frac{1}{\# \operatorname{MEC}({\mathcal{C}}_{M})} \sum_{\mathcal{G}_{M} \in \operatorname{MEC}({\mathcal{C}}_{ M})} \left( \widehat{S}^{\text{eff, nonpar}} \big(\widehat{IM}_j^{\text{QR}}(\mathcal{G}_M) \big) - \widehat{IM}_j^{\text{QR}}(\mathcal{G}_M) \right) \right]^2
    \end{aligned}
\]
correspondingly. However, in practical scenarios, confidence intervals (CIs) derived using the Wald-type method, especially when grounded on score-based variance estimators, tend to be more narrow \citep{boos2013essential}. This can potentially result in anti-conservatism. To achieve more concise statistical inference for our quadruply estimators, we consider utilizing the variances derived from the symmetric $t$-bootstrap approach \citep{hall1988symmetric} here. A pseudocode summarizing
the proposed algorithm for these quadruply estimators and their bootstrap CIs is given in Algorithm \ref{alg_QR}. 
The $\log n$ truncations in Algorithm \ref{alg_QR} aims to achieve the numerical stability, which is a technique widely recognized in statistical literature \citep{heckman1976common,sun2020adaptive,chinot2020robust}.
\begin{algorithm}[!htp]
    \caption{General algorithm for quadruply robust estimations}\label{alg_QR}
    \begin{flushleft}
        \textbf{INPUT:} The data $\{ X_i \}_{i = 1}^n$, the treatment index $t$, and the Monte Carlo sample size $N$.
    \end{flushleft}
\begin{algorithmic}[1] 
    \State Apply any suitable structure learning algorithm to learn the CPDAG of $X$ and obtain the corresponding estimated adjacency matrix $\widehat{\mathcal{C}}$, then set the estimated adjacency matrix of mediators as $\widehat{\mathcal{C}}_M = \left[\widehat{\mathcal{C}}_{kk}\right]_{k \in (t + 1):(t + p)}$.

    \State Using any proper estimating procedure to estimate the propensity score $e_{a'}(c) = \pr (A = a' \mid C = c)$ with $a' \in \{ 0, 1\}$.

    \State Using any proper estimation method to estimate the conditional densities $\pi_{\bcdot}(\bcdot)$ and conditional expectations $\mu(\bcdot)$.

%    \For{$j \in [p]$}
    \State For estimating the causal effect of $j$-th mediator, set $\widehat{DM}_j^{\text{QR}} = \widehat{TM}_j^{\text{avg}, \, \text{QR}} = 0$.

    \For{$i \in [n]$}
        \State Sample $M_{\widehat{\pi}_{C_i, 1}}^{(1)}, \ldots, M_{\widehat{\pi}_{C_i, 1}}^{(N)}\, \overset{\text{i.i.d.}}{\sim} \, \widehat{\pi}_{C, 1} (m) := \widehat{\pi}_{C, 1}(m_j) \widehat{\pi}_{C, 1, m_j}(m_{-j})$, and similarly sample $M_{\bcdot,  \widehat{\pi}_{\bcdot_i}}^{(1)}, \ldots, M_{\bcdot,  \widehat{\pi}_{\bcdot_i}}^{(N)}$ from any other conditional densities displaying in \eqref{DM_MR}.

        \State $\widehat{DM}_j^{\text{QR}} \, \longleftarrow \, \widehat{DM}_j^{\text{QR}} + n^{-1} \widehat{QR}_{DM_j}^{\text{MC}}(X_i) \mathds{1} \big( |\widehat{QR}_{DM_j}^{\text{MC}}(X_i)|\leq \log n\big)$, where $\widehat{QR}_{DM_j}^{\text{MC}}(X_i)$ is defined in \eqref{DM_QR_MC}.

        \For{$\mathcal{G}_M \in \operatorname{MEC}(\widehat{\mathcal{C}}_M)$}
            \State Obtain $M_j$'s parent mediators $\Pa_j (\mathcal{G}_M) = \left[ \mathcal{G}_M^{\top}  \circ 1_p 1_p^{\top} \right]_{j:}$.

            \State Sample $M_{\bcdot,  \widehat{\pi}_{\bcdot_i}}^{(1)}, \ldots, M_{\bcdot,  \widehat{\pi}_{\bcdot_i}}^{(N)}$ from all conditional densities displaying in \eqref{IM_MR} and in $\widehat{TM}_j^{\mathscr{M}_0}$.

            \State $\widehat{TM}_j^{\text{QR}}(\mathcal{G}_M) \, \longleftarrow \, \widehat{TM}_j^{\text{QR}} (\mathcal{G}_M) + n^{-1} \widehat{QR}_{TM_j}^{\text{MC}}(X_i \, ; \, \mathcal{G}_M) \mathds{1} \big( |\widehat{QR}_{TM_j}^{\text{MC}}(X_i)|\leq \log n\big)$, where $\widehat{QR}_{TM_j}^{\text{MC}}(X_i \, ; \, \mathcal{G}_M)$ is defined in \eqref{TM_QR_MC}.
        \EndFor
    \EndFor
    
    \State \Return {$\widehat{DM}_j^{\text{QR}}$ and $\widehat{IM}_j^{\text{avg}, \, \text{QR}} = \frac{1}{\# \operatorname{MEC}(\widehat{\mathcal{C}}_M)} \sum_{\mathcal{G}_M \in \operatorname{MEC}(\widehat{\mathcal{C}}_M)} \widehat{TM}_j^{\text{QR}} (\mathcal{G}_M) - \widehat{DM}_j^{\text{QR}}$.}

    \State Symmetric \textit{t}-bootstrap \citep{hall1988symmetric} is applied to construct confidence intervals for $\sqrt{n} (\widehat{DM}_j^{\text{QR}} - DM_j)$ and $\sqrt{n} (\widehat{IM}_j^{\text{avg}, \, \text{QR}} - \overline{IM}_j)$.
%    \EndFor
\end{algorithmic}
\end{algorithm}

\subsection{Practical fast implement}

The formulas for the quadruply robust estimators, as shown in equations \eqref{DM_MR} and \eqref{IM_MR}, require several numerical integrals for each $i \in [n]$, which may be computationally demanding. To address this challenge, we purpose Algorithm \ref{alg_QR} in the above section, in which we employ the Monte Carlo method to evaluate these integrals. However, when the data partly satisfy the semi-linear structure and both $\epsilon_M$ and $\epsilon_Y$ adhere to a mean-zero Gaussian distribution, explicit expressions for these numerical integrals can be derived, facilitating faster computation. Indeed, if we assume the linear structure in $M \, \leftarrow \, C \oplus A \oplus M$ and denote the density (or mass) function of $\epsilon_M = (\epsilon_{M, 1}, \ldots, \epsilon_{M, p})^{\top}$ as $f(x) = f(x_1, \ldots, x_p)$, then the conditional density of $M$ given $C$ and $A = 1$ is $f\big(x - {\Theta}_{MC} C - \theta_{MA} \big)$ from \eqref{reg_exp}.
This allows us to compute
\begin{equation}\label{fast_imple_eq1}
	\begin{aligned}
		& \int_{\mathcal{M}_j} \mu (C, 1, m_j, M_{-j}) \pi_{C, 1} (m_j) \, \mathrm{d} m_j \\
		= & \int_{\mathcal{M}_j} \Big[ \beta_{YC} C + \alpha_{YA}  + \beta_{YM, j} m_j + \beta_{YM, -j}^{\top} M_{-j} \Big]  f\big(m_j - \big[{\Theta}_{MC} C \big]_j - \theta_{MA, j} \big)\, \mathrm{d} m_j \\
		= & \beta_{YC} C + \alpha_{YA}+ \beta_{YM, j} \Big\{ \big[\widehat{\Theta}_{MC} C \big]_j + \theta_{MA, j}\Big\} + \beta_{YM, -j}^{\top} M_{-j}.
	\end{aligned}
\end{equation}
Similarly, we can derive explicit expressions for some other integrals in equations \eqref{DM_MR} and \eqref{IM_MR} as long as the linear structure in $M \, \leftarrow \, C \oplus A \oplus M$ holds. One step more, when \( \epsilon_M \) is a mean-zero Gaussian distribution, any integral in \eqref{DM_MR} and \eqref{IM_MR} will have an explicit expression. This leads to a more efficient implementation of \eqref{DM_MR} and \eqref{IM_MR}. The following Algorithm \ref{alg_fast_QR} and Proposition \ref{pro_fast} elaborates on this. 

\begin{algorithm}[H]
    \caption{Fast implement algorithm for quadruply robust estimations}\label{alg_fast_QR}
    \begin{flushleft}
        \textbf{INPUT:} The data $\{ X_i \}_{i = 1}^n$ and the treatment index $t$.
    \end{flushleft}
\begin{algorithmic}[1] 
    \State Apply any suitable structure learning algorithm (such as GES or PC) to learn the CPDAG of $X$ and obtain $\widehat{\mathcal{C}}$, then obtain the CPDAG of mediators $\left[\widehat{\mathcal{C}}_{kk}\right]_{k \in (t + 1):(t + p)}$.

    \State Using any proper estimating procedure to estimate the propensity score $e_1(c) = \pr (A = 1 \mid C = c)$ with $\widehat{e}_1(c)$ and $\widehat{e}_0(c) = 1 - \widehat{e}_1(c)$.
    
    \State Regress $M$ on $(C^{\top}, A)^{\top}$ obtain OLS estimator \( (\widehat{\Theta}_{MC}, \widehat\theta_{MA})^{\top} \) and the estimated covariance of error term $\widehat{\var}(e_M)$.
    
    \State Regress $Y$ on $(C^{\top}, A, M^{\top})^{\top}$ and $(C^{\top}, A, M^{\top})^{\top}$, obtain OLS estimators \( (\widehat\beta_{YC}^{\top}, \widehat\alpha_{YC}, \widehat\beta_{YM}^{\top})^{\top} \) and $(\widehat\eta_{YA}^{\dag}, \widehat\gamma_{YC}^{\dag \top})^{\top}$ correspondingly.

    \For{$j \in [p]$}
        \State Set $\widehat{DM}_j^{\text{QR}} = \widehat\beta_{YM, j} \widehat\theta_{MA, j}$ and $\widehat{TM}_j^{\text{avg}, \, \text{QR}} = 0$.

         \For{$i \in [n]$}
             \State $\widehat{DM}_j^{\text{QR}} \, \longleftarrow \, \widehat{DM}_j^{\text{QR}} + n^{-1} \widehat{QR}_{DM_j}^{\text{fast}}(X_i)$, where $\widehat{QR}_{DM_j}^{\text{fast}}(X_i)$ is defined in \eqref{DM_QR_fast}.

             \For{$\mathcal{G}_M \in \operatorname{MEC}(\widehat{\mathcal{C}}_M)$}
                 \State Obtain $M_j$'s parent mediators $\Pa_j (\mathcal{G}_M) = \left[ \mathcal{G}_M^{\top}  \circ 1_p 1_p^{\top} \right]_{j:}$.

                 \State Regress $Y$ on $(M_j,  \Pa_j(\mathcal{G}_M)^{\top},  A, C)^{\top}$ and get OLS estimator $(\widehat\eta_{Y M_j}, \widehat\gamma_{Y \Pa_j(\mathcal{G}_M)}^{\top}, \widehat\eta_{YA}, \widehat\gamma_{YC}^{\top})^{\top}$.
                 
                 \State $\widehat{TM}_j^{\text{QR}} (\mathcal{G}_M) \, \longleftarrow \, n^{-1} \left( \widehat\eta_{Y M_j} + \widehat{QR}_{TM_j}^{\text{fast}}(X_i; \mathcal{G}_M) \right)$, where $\widehat{QR}_{TM_j}^{\text{fast}}(X_i, \mathcal{G}_M)$ is defined in \eqref{IM_QR_fast}.
             \EndFor
        \EndFor
        \State \Return {$\widehat{DM}_j^{\text{QR}}$ and $\widehat{IM}_j^{ \text{avg}, \, \text{QR}} = \frac{1}{\# \operatorname{MEC}(\widehat{\mathcal{C}}_M)} \sum_{\mathcal{G}_M \in \operatorname{MEC}(\widehat{\mathcal{C}}_M)} \widehat{TM}_j^{\text{QR}} (\mathcal{G}_M) - \widehat{DM}_j^{\text{QR}}$ with their symmetric \textit{t}-bootstrap CIs.}
    \EndFor
\end{algorithmic}
\end{algorithm}

\begin{pro}\label{pro_fast}
    Assume that at least one linear structure in Assumption \ref{ass_slsem} holds, and that \( \epsilon_M \) and \( \epsilon_Y \) are both mean-zero Gaussian distributed, Algorithm \ref{alg_fast_QR} produces valid quadruply robust estimators $\big\{ \widehat{DM}_j^{\text{QR}}, \widehat{IM}_j^{\text{avg}, \, \text{QR}} \big\}_{j = 1}^p$ as defined in Section \ref{sec_mr}.
\end{pro}

In practical scenarios where the sample size \( n \) is sufficiently large, it becomes reasonable to treat the sample means $\overline{\epsilon}_M := \frac{1}{n} \sum_{i = 1}^n \epsilon_{M, i}$ and $\overline{\epsilon}_Y := \frac{1}{n} \sum_{i = 1}^n \epsilon_{Y, i}$ as if they follow mean-zero Gaussian distributions. This permits the utilization of Proposition \ref{pro_fast}, particularly when empirical evidence can support the linear structural relationships for \( M \, \leftarrow \, C \oplus A \oplus M \) or \( Y \, \leftarrow \, C \oplus A \oplus M \).

\section{Asymptotic Behavior}\label{sec_asy}

In this section, we first give the asymptotic properties of the OLS estimators when the model satisfies Assumption \ref{ass_slsem}. Then we will establish the asymptotic normality of quadruply robust estimators, allowing the model misspecification. 

\subsection{Asymptotic Properties of OLS estimators}\label{sec_asy_ols}
In section \ref{sec_ols_est}, we highlighted that given the causal structure is appropriately specified as semi-linear according to Assumption \ref{ass_slsem}, one can employ OLS estimators by just applying two simple regressions. As we allow the number of mediators $p$ can grow with sample size $n$, some assumptions are required. The following assumptions come from \cite{portnoy1984asymptotic} and \cite{portnoy1985asymptotic}. 
They control the behavior of minimum eigenvalue will hold in probability if the observations are a sample from an appropriate distribution in $\mathbb{R}^p$. Denote the error vector as $\epsilon = (\epsilon_A^{\top}, \epsilon_M^{\top}, \epsilon_Y^{\top})^{\top}$.

\begin{ass}\label{ass_error_dis}
    (\textit{Assumptions for Error Distributions}) $\epsilon$ is marginal sub-Gaussian with finite Orlicz norm (Definition (6.18) in \cite{wainwright2019high}).
\end{ass}

\begin{ass}\label{ass_nor_eigen}
    (\textit{Restricted Eigenvalue Condition}) $\lim_{n \rightarrow \infty} \var \big( Y \mid C \, \cup \, A \, \cup \, M\big) > 0$ and $\lim_{n \rightarrow \infty } \E \var (M \mid C \, \cup \, A) \succ 0$.
\end{ass}

We use the bold symbol $\mathbf{X}$ to represent the data matrix of any i.i.d. random observations $\{ X_i \}_{i = 1}^n $. i.e. $\mathbf{X} = [X_1, \ldots, X_n]^{\top}$. Denote the transformation $\widehat\Gamma_{\bcdot, \bcdot}: \mathbb{R}^{n \times p_1} \times \mathbb{R}^{n \times p_2} \rightarrow \mathbb{R}^{p_1 \times n}$ of two data matrix with sample size $n$ as 
\begin{equation*}
    \widehat\Gamma_{X, Z} := \big[ \mathbf{X}^{\top} (I_n - P_{\mathbf{Z}}) \mathbf{X} \big]^{-1} \mathbf{X}^{\top} (I_n - P_{\mathbf{Z}}),
\end{equation*}
where $P_{\mathbf{Z}} = \mathbf{Z} (\mathbf{Z}^{\top} \mathbf{Z})^{-1} \mathbf{Z}^{\top} \in \mathbb{R}^{n \times n}$ is the projection matrix of $\mathbf{Z}$. This transformation streamlines our representation of the asymptotics for our OLS estimators.

\begin{theorem}\label{thm_CI_DE_IE_DM}
    Suppose the model satisfies Assumptions \ref{ass_structure}, \ref{ass_slsem}, \ref{ass_error_dis}, and \ref{ass_nor_eigen}. Let $\{ \widehat{e}_{M, i} \}_{i = 1}^n $ and $\{ \widehat{\epsilon}_{Y, i} \}_{i = 1}^n$ be the residuals from the OLS estimator in regression \eqref{reg_exp}.
    Then for any $\alpha \in (0, 1)$, we have 
    \[
        \lim_{n \rightarrow \infty}\pr \Bigg( \big| \sqrt{n}(\widehat{DE}^{\text{OLS}} - DE)  \big| \leq \Phi^{-1} (1 - \alpha / 2) \sqrt{\widehat{\Gamma}_{A, (M, C)} \widehat{\Gamma}_{A, (M, C)}^{\top} \sum_{i = 1}^n \widehat\epsilon_{Y, i}^2} \Bigg) = 1 - \alpha
    \]
    Furthermore, denote $\widehat{\Sigma}_{\beta_{YM}} = \sum_{i = 1}^n \widehat\epsilon_{Y, i}^2 \widehat{\Gamma}_{M, (C, A)} \widehat{\Gamma}_{M, (C, A)}^{\top}$ and $\widehat{\Sigma}_{\theta_{MA}} = \sum_{i = 1}^n \widehat{e}_{M, i} \widehat{e}_{M, i}^{\top} \widehat\Gamma_{A, C} \widehat\Gamma_{A, C}^{\top}$, then 
    \[
        \lim_{n \rightarrow \infty}\pr \Bigg( \big| \sqrt{n}(\widehat{IE}^{\text{OLS}} - IE)  \big| \leq \Phi^{-1} (1 - \alpha / 2)\sqrt{\widehat\beta_{YM}^{\top} \widehat{\Sigma}_{\theta_{MA}} \widehat\beta_{YM} + \widehat\theta_{MA}^{\top} \widehat{\Sigma}_{\beta_{YM}} \widehat\theta_{MA}} \Bigg) 
        \geq 1 - \alpha,
    \]
    and
    \[
        \lim_{n \rightarrow \infty}\pr \Bigg( \big| \sqrt{n}(\widehat{DM}_j^{\text{OLS}} - DM_j)  \big| \leq \Phi^{-1} (1 - \alpha / 2)\sqrt{\widehat\beta_{YM, j}^2\widehat{\Sigma}_{\theta_{MA}, jj} + \widehat\theta_{MA, j}^2 \widehat{\Sigma}_{\beta_{YM}, jj}} \Bigg) 
        \geq 1 - \alpha
    \]
    for any $j \in [p]$.
\end{theorem}

\noindent The above theorem ensures that under mild conditions we can construct valid confidence intervals for $\widehat{DE}^{\text{OLS}}$, $\widehat{IE}^{\text{OLS}}$, and $\widehat{DM_j}^{\text{OLS}}$ when $n$ is large enough. It is worthy to note that the probabilities for $\widehat{IE}^{\text{OLS}}$ and $\widehat{DM_j}^{\text{OLS}}$ is $\geq$ instead of $=$.
This distinction arises from the dual nature of the limiting distributions for these two OLS estimators: one is the standard normal, the other is not.
However, as argued in \cite{chakrabortty2018inference}, the non-standard asymptotic distributions here are more conservative than $\mathcal{N} (0, 1)$. Thus, we obtain $\geq$ instead of $=$. 
%Informally, we can view $\widehat\beta_{YM}^{\top} \widehat{\Sigma}_{\theta_{MA}} \widehat\beta_{YM} + \widehat\theta_{MA}^{\top} \widehat{\Sigma}_{\beta_{YM}} \widehat\theta_{MA}$ and $\widehat\beta_{YM, j}^2\widehat{\Sigma}_{\theta_{MA}, jj} + \widehat\theta_{MA, j}^2 \widehat{\Sigma}_{\beta_{YM}, jj}$ as the upper bounds of the variances for $\sqrt{n} (\widehat{IE}^{\text{OLS}} - IE)$ and $\sqrt{n} (\widehat{DM}_j^{\text{OLS}} - DM_j)$ correspondingly.
The details can be found in the proof. Notably, these asymptotic confidence intervals can be derived concurrently with the regression estimators and residuals. When applying the regression to procure these estimators, no additional steps are needed to obtain these confidence intervals.

For the estimators $\widehat{IM}_j^{\text{OLS}}$, additional assumptions are needed due to their reliance on the unknown DAG structure. This necessitates consistent CPDAG estimation, as well as more strong sparsity assumptions and restricted eigenvalue conditions, which are common in high-dimensional settings \citep{portnoy1985asymptotic, van2014asymptotically, zhang2014confidence, chakrabortty2018inference}.

\begin{ass}\label{ass_nor_cpdag}
    (\textit{Structure learning consistency}) Consistency of learning structure: $\pr (\widehat{\mathcal{C}}_{M} \neq \mathcal{C}_{M}) \longrightarrow 0$.
\end{ass}

\begin{ass}\label{ass_nor_adj}
    The sparsity of maximum degree in $\mathcal{C}_{ M}$, $\max_{j \in [p]} q_j = \max_{j \in [p]} \penalty 0 |\operatorname{adj}  (M_j)| = O(n^{1 - b_1})$ for some $0 < b_1 \leq 1$.
\end{ass}

\begin{ass}\label{ass_nor_ele}
    $\lim_{n \rightarrow \infty} \max_{j \in [p]} n^{-1 / 2}\left\{q_{j}+\log \left(L_{\mathrm{distinct}, j}\right)\right\} = 0$, where $L_{\text {distinct}, j}$ is the number of distinct elements of the set $\{{X}_{S_{j 1}}, \penalty 0 \ldots, {X}_{S_{j L_{j}}}\} = \big\{ (M_{j}, \Pa_j(\mathcal{G}_{M}), A, C)^{\top}: \mathcal{G}_{M}\in \operatorname{MEC}(\mathcal{C}_{M})\big\}$.
\end{ass}

\begin{ass}\label{ass_nor_egv_pa}
    $\lim_{n \rightarrow \infty} \min_{j \in [p]} \var (Y \mid \operatorname{adj}(M_j) \,  \cup \, C \, \cup \, M_j ) > 0$ and $\lim_{n \rightarrow \infty} \penalty 0 \min_{j \in [p]} \E \var (M_j \mid \operatorname{adj}(M_j) \, \cup \, C) > 0$.
\end{ass}

\begin{theorem}\label{thm_CI_IM}
    Suppose 
    Assumption \ref{ass_structure}, \ref{ass_slsem}, \ref{ass_error_dis}, \ref{ass_nor_eigen} and Assumption \ref{ass_nor_cpdag}, \ref{ass_nor_adj}, \ref{ass_nor_ele}, \ref{ass_nor_egv_pa} 
    hold, then
    \[
        \lim_{n \rightarrow \infty} \pr \Big( \sqrt{n} \big| \widehat{IM}_j^{OLS} - IM_j \big| \geq \widehat\sigma_{\overline{IM}_j} \Phi^{-1} (1 - \alpha / 2)  \Big) \geq 1 - \alpha
    \]
    for any $\alpha \in (0, 1)$. The explicit formula for $\widehat\sigma_{\overline{IM}_j}^2$ can be found in \eqref{est_sigma_IM} in Appendix \ref{proof_A}.
\end{theorem}

\noindent We now therefore obtain a valid asymptotic confidence interval for $\overline{IM}_j$ for any $j \in [p]$ alongside the regression from Theorem \ref{thm_CI_IM}.

% \blue{
% \begin{theorem}\label{thm_nor_IM}
%     Suppose 
%     Assumption \ref{ass_slsem},  \ref{ass_dimension}, \ref{ass_error_dis}, \ref{ass_nor_eigen} and Assumption \ref{ass_nor_cpdag}, \ref{ass_nor_adj}, \ref{ass_nor_ele}, \ref{ass_nor_egv_pa} 
%     hold, then
%     \[
%         \begin{aligned}
%             & \frac{\sqrt{n} (\widehat{IM}_j - \overline{IM}_j)}{\widehat\sigma_{\overline{IM}_j}} \, \rightsquigarrow \\
%             &  \left\{ \begin{array}{ll}
%                 \displaystyle\frac{Z_1 Z_2}{\sqrt{Z_1^2  + Z_2^2 + 2 \rho Z_1 Z_2}}, \quad & \text{if } \theta_{MA, j} = \frac{\sum_{\mathcal{G}_{0, M} \in \operatorname{MEC}(\mathcal{C}_{0, M})} \left[ \E \left[ Y \mid M_j \cup \pa_j (\mathcal{G}_{0, M}) \cup A \cup C \right]_1 - \beta_{YM, j} \right]}{\# \operatorname{MEC}(\mathcal{C}_{0, M})} = 0\\
%                 \mathcal{N} (0, 1), & \text{otherwise.}
%             \end{array}\right.
%         \end{aligned}
%     \]
%     where $\left( \begin{array}{c}
%          Z_1 \\
%          Z_2
%     \end{array}\right) \, \sim \, \mathcal{N}\left(0, \left(\begin{array}{cc}
%     1 & \rho_j \\
%     \rho_j & 1
%     \end{array}\right)\right)$ and $\rho_j = \plim \widehat\rho_j$. The explicit formula for $\widehat\sigma_{\overline{IM}_j}$ and $\widehat\rho_j$ can be found in Appendix.
% \end{theorem}
% }

\subsection{Asymptotic Properties of Quadruply Robust Estimators}\label{sec_asy_mr}

The quadruply robust estimators aim to obtain the robust estimators even when the model is misspeficied. The double robust estimators, which combines the direct and IPW strategies, possess commendable properties and have been the subject of extensive research as evidenced in literature such as \citep{laan2003unified,tsiatis2006semiparametric,kang2007demystifying}. In this section, we will show that the proposed novel quadruply robust estimators exhibit more favorable asymptotic properties.

To present the results, we assume that the propensity score $e_{a'} (\bcdot) \in \mathcal{E}$ with some function classes $\mathcal{E}$. Similarly, for each $j \in [p]$, we assume any conditional density employed in \eqref{DM_MR} and \eqref{IM_MR} is 
\[
    f_{M_T \mid X_S}(m_T \mid x_{S}) \, \in \, \mathcal{F}_{j, \, T \mid S},
\]
and any conditional mean used in \eqref{DM_MR} and \eqref{IM_MR} adheres to
\[
    {\E} [Y \mid x_{S}] \, \in \, \mathcal{U}_{j, \, S} 
\]
with some specific function classes $\mathcal{F}_{j, \, T \mid S}$ and $\mathcal{U}_{j, \, S}$. We propose the following assumptions concerning these function classes and the convergence rates of the estimators within these classes.

\begin{ass}\label{ass_fun_class}
    For any fixed $j \in [p]$, any subset $S \subseteq [t + p + 1]$ and $T \subseteq [p]$ used in \eqref{DM_MR} and \eqref{IM_MR}, the function classes $\mathcal{E}$, $\mathcal{F}_{j, \, T \mid S}$, and $\mathcal{U}_{j, \, S}$ are bounded and belong to VC type classes (Definition 2.1 in \cite{chernozhukov2014gaussian}) with VC indices upper bounded by $v_j = O(n^{\vartheta_j})$ for some $\vartheta_j$ such that $\vartheta_j \in [0, 1/2)$.
\end{ass}

% \todohw{Revised for the consistency of $\widehat{IM}_j^{\text{MR}, \, \text{avg}}$.}

% \begin{ass}\label{ass_conv_rate_DR}
%     For any subset $(S_1, S_2) \subseteq [t + p + 1]$ used in \eqref{DM_DR} and \eqref{IM_DR}, the estimators $\widehat{\E} [x_{S_2} \mid x_{S_1}]$ and $\widehat{f} (x_{S_2} \mid x_{S_1})$ converge with $\ell_2$-norm to their true values at a rate of $n^{-\vartheta^*}$ for some $\vartheta^* \geq 1 / 2$ in the function classes $\mathcal{U}_{S_2 \mid S_1}$ and $\mathcal{F}_{S_2 \mid S_1}$, respectively.
% \end{ass}

% \begin{ass}\label{ass_conv_rate_QR}
%     For any subset $(S_1, S_2) \subseteq [t + p + 1]$ used in \eqref{DM_MR} and \eqref{IM_MR}, the estimators $\widehat{\E} [x_{S_2} \mid x_{S_1}]$ and $\widehat{f} (x_{S_2} \mid x_{S_1})$ converge with $\ell_2$-norm to their true values not less than the parametric rate $n^{-\vartheta^*}$ for some $\vartheta^* > 1 / 4$ in the function classes $\mathcal{U}_{S_2 \mid S_1}$ and $\mathcal{F}_{S_2 \mid S_1}$, respectively.
% \end{ass}

\begin{ass}\label{ass_conv_rate_QR}
    For any fixed $j \in [p]$, any subset $S \subseteq [t + p + 1]$ and $T \subseteq [p]$ used in \eqref{DM_MR} and \eqref{IM_MR}, the estimators $\widehat\pi_{x_S} (m_T)$ and $\widehat{\mu} (x_s)$ converge with $\ell_2$-norm to their true values at the rates of $n^{-\vartheta_{j, \pi}^*}$ and $n^{-\vartheta_{j, \mu}^*}$, and the propensity score estimator $\widehat{e}_{a'} (\bcdot)$ with $a' \in \{ 0, 1\}$ converge with $\ell_2$-norm to their true values at the rates of $n^{-\vartheta_{j, e}^*}$. Here the positive numbers $\vartheta_{j, e}^*, \vartheta_{j, \pi}^*, \vartheta_{j, \mu}^*$ satisfies: (i) $\min \{ \vartheta_{j, e}^*, \vartheta_{j, \pi}^*, \vartheta_{j, \mu}^* \} > \vartheta_j / 2$; (ii) $v_1^* + v_2^* > 1 / 2$ for any $\{ v_1^*, v_2^*\} \subsetneq \{ \vartheta_{j, e}^*, \vartheta_{j, \pi}^*, \vartheta_{j, \mu}^* \}$.
\end{ass}

% \todohw{这里改成了 Theorem 7 的\cite{kallus2020statistically} 中的收敛。另外我加入了 propensity 的 function class。}
% \todohw{这里改成了单独的 $j \in [p]$，后面也去掉了联合正态性。}
% \todohw{上面这些 assumption 的符号可能需要加上 $j$。}

\noindent Assumption \ref{ass_fun_class} is reasonably moderate, as the function classes are user-defined. VC-type classes encompass a broad spectrum of functional categories, including but not limited to classic parametric model, neural networks and regression trees. The VC index governs the complexity of the model, typically escalating with an increase in the number of parameters within the model. We permit the VC index to diverge alongside the sample size, which serves to minimize the estimator's bias arising from model misspecification. On the other hand, an important feature of Assumption \ref{ass_conv_rate_QR} is that the required estimators' convergence rates can only be nonparametric (slower than $n^{-1 / 2}$ ) and no metric entropy condition (Donsker class for instance) is needed. In particular, $\vartheta_{j, e}^*, \vartheta_{j, \pi}^*, \vartheta_{j, \mu}^* > 1 / 4$ will perfectly admit Assumption \ref{ass_conv_rate_QR}. Therefore, the estimators can be computed via standard nonparametric estimation \citep{fan2003nonlinear} and supervised learning algorithms \citep[including random forests and deep learning,][]{wager2018estimation,schmidt2020nonparametric}. 
The reason both assumptions regarding the sizes of the function classes and the rate of convergence for $\widehat{DM}_j^{\text{QR}}$ and $\widehat{IM}_j^{\text{avg}, \, \text{QR}}$ are identical, which are different from conditions in Theorem \ref{thm_CI_DE_IE_DM} and Theorem \ref{thm_CI_IM}, stems from the uniform convergence characteristics of our estimations for any $\mathcal{G}_M$ in $\operatorname{MEC}(\mathcal{C}_M)$.

\begin{theorem}\label{thm_QR_nor}
    Let the conditions in Theorem \ref{thm_EIF} and Assumption \ref{ass_fun_class} hold. Suppose the estimators $\widehat{e}_{a'}(x_s)$, $\widehat{\mu} (x_s)$, and $\widehat{\pi}_{x_s} (m_T)$ in either $\mathscr{M}_0$, $\mathscr{M}_{j, 1}$, $\mathscr{M}_{j, 2}$, or $\mathscr{M}_{j, 3}$ converges in $\ell_2$-norm to their true values for each $j \in [p]$. Then
    \begin{itemize}
        \item $\widehat{DM}_j^{\text{QR}}$ is the consistent estimator of $DM_j$ under the model $\mathscr{M}_{j, \, \text{union}}$ for any $j \in [p]$. Furthermore, if Assumption \ref{ass_conv_rate_QR} holds, then $\sqrt{n} \big( \widehat{DM}_j^{\text{QR}}  - {DM}_j \big)$ is asymptotic normally distributed under model $\mathscr{M}_{\text{nonpar}}$ with asymptotic variance $\E \Big\{ \big[ S^{\text{eff, nonpar}}(DM_j)\big]^2 \Big\}$.
        \item If Assumption \ref{ass_nor_cpdag} also hold, $\widehat{IM}_j^{\text{avg}, \, \text{QR}}$ is the consistent estimator of $\overline{IM}_j$ under the model $\mathscr{M}_{j, \, \text{union}}$ for any $j \in [p]$. Furthermore, if Assumption \ref{ass_conv_rate_QR} holds, then $\sqrt{n} \big( \widehat{IM}_j^{\text{avg}, \, \text{QR}} - \overline{IM}_j\big)$ is asymptotically normally distributed under model $\mathscr{M}_{\text{nonpar}}$ with asymptotic variance 
        \[
            \begin{aligned}
               \E \Bigg\{ \Bigg[  \frac{1}{\# \operatorname{MEC}(\mathcal{C}_{M})}  \sum_{\mathcal{G}_{M} \in \operatorname{MEC}(\mathcal{C}_{M})} S^{\text{eff, nonpar}}(IM_j (\mathcal{G}_M)) \Bigg]^2 \Bigg\}.
            \end{aligned}
        \]
    \end{itemize}
\end{theorem}

\noindent An important result of Theorem \ref{thm_QR_nor} is that: for any $j \in [p]$, the quadruply robust estimators $\widehat{DM}_j^{\text{QR}}$ and $\widehat{IM}_j^{\text{avg}, \, \text{QR}}$
are semiparametric locally efficient in the sense that they are regular and asymptotically linear under model $\mathscr{M}_{j, \, \text{union}}$, and achieve the
semiparametric efficiency bound for ${DM}_j$ and $\overline{IM}_j$ under model at the intersection submodel $\mathscr{M}_0 \, \cap \, \mathscr{M}_{j, 1} \, \cap \, \mathscr{M}_{j, 2} \, \cap \, \mathscr{M}_{j, 3}$. Hence, when all models are correct, $\widehat{DM}_j^{\text{QR}}$ and $\widehat{IM}_j^{\text{avg}, \, \text{QR}}$ are semiparametric efficient in the model $\mathscr{M}_{\text{nonpar}}$ at the
intersection submodel $\mathscr{M}_0 \, \cap \, \mathscr{M}_{j, 1} \, \cap \, \mathscr{M}_{j, 2} \, \cap \, \mathscr{M}_{j, 3}$ by part iv in \cite{bickel2001inference} for any $j \in [p]$.

\section{Simulation Studies}\label{sec_sim}

In this section, we assess the finite-sample performance of our proposed quadruply robust estimators across two simulation scenarios.
The first scenario seeks to illustrate the robustness characteristics of our estimator in comparison to other estimation strategies, particularly when certain model specifications to a specific mediator are not met.
In the second simulation study, we demonstrate that our method can also be superior to any other estimation strategies in estimating the both direct and indirect interventional effects across all mediators, on average, within commonly adopted model configurations.

\subsection{Simulation for a single mediator}

We consider the finite-sample performance of the proposed quadruply robust estimators in comparison to the estimators under direct strategy, and the alternative strategies in Section \ref{sec_straight_strategy}, \ref{sub_sec_alternative_strategy_1}, \ref{sub_sec_alternative_strategy_2}, and \ref{sub_sec_alternative_strategy_3} for a single mediator.
We describe the detailed setting as follows: we set $t = p = 3$, and fix the pre-specified $j$ randomly sampled from $U\{ [p] \}$. Then we design the following four data generating processes (DGPs), here $\Phi(\bcdot)$ and $\operatorname{logit}(\bcdot)$ are the standard normal distribution function and the inverse of the standard logistic function, and $B_{j :}$ and $B_{- j :}$ represent the $j$-th row of the matrix $B$ and the matrix $B$ with
$j$-th row removed.
\begin{itemize}
    \setlength\itemsep{1.5ex}
    \item All correct: $C \, \leftarrow \, \mathcal{N} (0, I_{t - 1})$, $A \, \leftarrow \, \mathds{1} \big\{ U[0, 1] \leq \Phi(\beta_{AC}^{\top} C) \big\}$, $M \, \leftarrow \, B_{M C}^{\top} C + \beta_{MA} A + B_{MM}^{\top} M + \mathcal{N} (0, I_p)$, and $Y \, \leftarrow \, \beta_{YC}^{\top} C + \alpha_{YA} A + \beta_{YM}^{\top} M + \mathcal{N} (0, 1)$;
    \item $\mathscr{M}_0$ is correct: the exposure $A$ comes from $A \, \leftarrow \, \mathds{1} (\operatorname{logit}(U[0, 1]) \leq \beta_{AC}^{\top} C)$ instead;
    % \[
    %     \begin{aligned}
    %         & C \, \leftarrow \, \mathcal{N} (0, I_{t - 1}), \quad A \, \leftarrow \, \underbrace{\mathds{1} (\operatorname{logit}(U[0, 1]) \leq \beta_{AC}^{\top} C)}_{\text{wrongly specificed}}, \\
    %         & M \, \leftarrow \, B_{M C}^{\top} C + \beta_{MA} A + B_{MM}^{\top} M + \mathcal{N} (0, I_p), \quad Y \, \leftarrow \, \beta_{YC}^{\top} C + \alpha_{YA} A + \beta_{YM}^{\top} M + \mathcal{N} (0, 1).
    %     \end{aligned}
    % \]
    \item $\mathscr{M}_{j, \, 1}$ is correct: the outcome $Y$ comes from $Y \, \leftarrow \, (\beta_{YC}^{\top} C + \alpha_{YA} A + \beta_{YM}^{\top} M)^{2 / 3} + \mathcal{N} (0, 1)$ instead;
    % \[
    %     \begin{aligned}
    %         & C \, \leftarrow \, \mathcal{N} (0, I_{t - 1}), \quad A \, \leftarrow \, \mathds{1} \big\{ U[0, 1] \leq \Phi(\beta_{AC}^{\top} C) \big\},\\
    %         & M \, \leftarrow \, B_{M C}^{\top} C + \beta_{MA} A + B_{MM}^{\top} M + \mathcal{N} (0, I_p), \quad Y \, \leftarrow \, \underbrace{(\beta_{YC}^{\top} C + \alpha_{YA} A + \beta_{YM}^{\top} M)^{2 / 3}}_{\text{wrongly specificed}} + \mathcal{N} (0, 1).
    %     \end{aligned}
    % \]
    \item $\mathscr{M}_{j, \, 2}$ is correct: the mediators have the alternative structure $ M_j \, \leftarrow \, \Theta_{MC, j:} C + \theta_{MA, j} A + \Big[ (I - B_{MM}^{\top})^{-1}\mathcal{N} (0, I_p) \Big]_{j}$ and $M_k \, \leftarrow \, {(\Theta_{MA, k:} C + \theta_{MA, k} A)^{2 / 3}} + \Big[ (I - B_{MM}^{\top})^{-1}\mathcal{N} (0, I_p) \Big]_{k}$ for $ k \neq j$;
    % \[
    %     \begin{aligned}
    %         &  C \, \leftarrow \, \mathcal{N} (0, I_{t - 1}), \quad A \, \leftarrow \, \mathds{1} \big\{ U[0, 1] \leq \Phi(\beta_{AC}^{\top} C) \big\}, \\
    %         & M_j \, \leftarrow \, \Theta_{MC, j:} C + \theta_{MA, j} A + \Big[ (I - B_{MM}^{\top})^{-1}\mathcal{N} (0, I_p) \Big]_{j}, \\
    %         & M_k \, \leftarrow \, \underbrace{(\Theta_{MA, k:} C + \theta_{MA, k} A)^{2 / 3}}_{\text{wrongly specificed}} + \Big[ (I - B_{MM}^{\top})^{-1}\mathcal{N} (0, I_p) \Big]_{k}, \quad \text{ for } k \neq j, \\
    %         & Y \, \leftarrow \, \beta_{YC}^{\top} C + \alpha_{YA} A + \beta_{YM}^{\top} M + \mathcal{N} (0, 1).
    %     \end{aligned}
    % \]
    \item $\mathscr{M}_{j, \, 3}$ is correct: the mediators have the alternative structure $M_j \, \leftarrow \, \penalty 0 {\Theta_{MC, j:} C + \frac{1}{2} \theta_{MA, j}} + \Big[ (I - B_{MM}^{\top})^{-1}\mathcal{N} (0, I_p) \Big]_{j}$ and $M_{-j} \, \leftarrow \, \Theta_{MC, -j :} C + \theta_{MA, -j} A + \Big[ (I - B_{MM}^{\top})^{-1} \penalty 0 \mathcal{N} (0, I_p) \Big]_{- j}$.
    % \[
    %     \begin{aligned}
    %         & C \, \leftarrow \, \mathcal{N} (0, I_{t - 1}), \quad A \, \leftarrow \, U \{ 0, 1\},\\
    %         & M_{-j} \, \leftarrow \, \Theta_{MC, -j :} C + \theta_{MA, -j} A + \Big[ (I - B_{MM}^{\top})^{-1}\mathcal{N} (0, I_p) \Big]_{- j}, \\
    %         & M_j \, \leftarrow \, \penalty 0 \underbrace{\Theta_{MC, j:} C + \frac{1}{2} \theta_{MA, j}}_{\text{wrongly specificed for $\pi_{C, a'} (\cdot)$ but correctly specificed for $\pi_{C} (\cdot)$}} + \Big[ (I - B_{MM}^{\top})^{-1}\mathcal{N} (0, I_p) \Big]_{j}, \\
    %         & Y \, \leftarrow \, \beta_{YC}^{\top} C + \alpha_{YA} A + \beta_{YM}^{\top} M + \mathcal{N} (0, 1).
    %     \end{aligned}
    % \]
\end{itemize}
Here the true adjacency matrix of mediators is generated from the Erd\H{o}s-Rényi (ER) model with an expected degree as $\lfloor p / 2 \rfloor$, and the non-zero entries in $B_{MM} \in \mathbb{R}^{p \times p}$ and all the elements in $\alpha_{YA}, \beta_{MA}, \beta_{YC}, \beta_{YM}, B_{MC}$ are independently sampled from $U(-1, 1)$. In each estimation method, we consistently treat $\mathscr{M}_0$ as the underlying true model by default. We generate $n = 1000$ simulation samples, each comprising $N = 100$ independent observations, and the result for estimating the direct and indirect interventional effect of the pre-specified mediator $M_j$ is shown in Table \ref{tab_art_each}. Here, we use the PC algorithm \citep{harris2013pc} to estimate the adjacency matrix of CPDAGs.

\begin{table}[htbp]
  \centering
  \caption{The average Bias (Standard Error) under simulation with sample size $n = 1000$ under $N = 100$ replications. }
    \begin{tabular}{ccccccc}
    \toprule
    \toprule
          &       & all correct & $\mathscr{M}_0$ is correct & $\mathscr{M}_{j, 1}$ is correct & $\mathscr{M}_{j, 2}$ is correct & $\mathscr{M}_{j, 3}$ is correct \\
    \midrule
    \textit{direct} & \multirow{2}[1]{*}{$\mathscr{M}_0$} & 0.006 (0.009) & \multicolumn{1}{c}{0.004 (0.004)} & 0.187 (0.069) & 0.032 (0.038) & 0.014 (0.001) \\
\cmidrule{1-1}\cmidrule{3-7}    \textit{indirect} &       & 0.033 (0.028) & \multicolumn{1}{c}{0.041 (0.033)} & 1.930 (0.027) & 0.028 (0.035) & 0.001 (0.002) \\
    \midrule
    \textit{direct} & \multirow{2}[1]{*}{$\mathscr{M}_1$} & 0.227 (0.074) & \multicolumn{1}{c}{0.064 (0.069)} & 0.234 (0.485) & 0.899 (0.087) & 0.383 (0.087) \\
\cmidrule{1-1}\cmidrule{3-7}    \textit{indirect} &       & 0.683 (0.344) & \multicolumn{1}{c}{0.756 (0.395)} & 0.418 (0.511) & 2.214 (0.131) & 1.675 (0.140) \\
    \midrule
    \textit{direct} & \multirow{2}[1]{*}{$\mathscr{M}_2$} & 0.007 (0.009) & \multicolumn{1}{c}{0.004 (0.004)} & 0.187 (0.068) & 0.032 (0.038) & 0.014 (0.001) \\
\cmidrule{1-1}\cmidrule{3-7}    \textit{indirect} &       & 0.853 (0.081) & \multicolumn{1}{c}{0.914 (0.104)} & 2.040 (0.134) & 0.237 (0.102) & 1.292 (0.095) \\
    \midrule
    \textit{direct} & \multirow{2}[1]{*}{$\mathscr{M}_3$} & 0.038 (0.039) & \multicolumn{1}{c}{0.045 (0.057)} & 0.229 (0.179) & 0.033 (0.042) & 0.050 (0.062) \\
\cmidrule{1-1}\cmidrule{3-7}    \textit{indirect} &       & 0.176 (0.128) & \multicolumn{1}{c}{0.068 (0.082)} & 0.127 (0.156) & 1.030 (0.079) & 0.069 (0.086) \\
    \midrule
    \textit{direct} & \multirow{2}[1]{*}{QR} & 0.031 (0.043) & 0.008 (0.012) & 0.145 (0.212) & 0.054 (0.065) & 0.025 (0.031) \\
\cmidrule{1-1}\cmidrule{3-7}    \textit{indirect} &       & 0.124 (0.118) & 0.087 (0.109) & 0.281 (0.571) & 0.053 (0.064) & 0.024 (0.030) \\
    \bottomrule
    \bottomrule
    \end{tabular}%
  \label{tab_art_each}%
\end{table}%

\noindent As illustrated in Table \ref{tab_art_each}, the simulation results align with the theoretical predictions made in previous sections. Specifically, when the entire distribution $F_X(\bcdot)$ is correctly specified, all estimators display consistency. However, in the presence of at least one misspecified component, only the quadruply robust estimator retains consistency. In contrast, one among the other estimators, $\mathscr{M}_{\ell}$ for $\ell = 0, 1, 2, 3$, becomes inconsistent. Although we present only the continuous scenario in this part, our simulations under discrete $C$ or $M$ settings yielded similar outcomes. Importantly, under this simulation scenario, the estimator $\widehat{TM}_j^{\mathscr{M}_0} = \widehat{DM}_j^{\mathscr{M}_0} + \widehat{IM}_j^{\mathscr{M}_0}$ corresponds precisely to the estimator utilized for the individual mediation effect $\eta_j$ proposed in \cite{chakrabortty2018inference}. Thus, our quadruply robust estimators outperform the estimator defined in \cite{chakrabortty2018inference}.

\subsection{Simulation for all mediators}

Next, we consider the average performance of our quadruply robust estimators compared with other estimations under a fair model misspecification scenario in both continuous case (Section 1.4 in \cite{kang2007demystifying}) and discrete case (Section 4.1 in \cite{xia2023identification}). The DGPs are defined as follows:
\begin{itemize}
    \setlength\itemsep{1.5ex}
    \item \textit{Continuous $M$:} $Z \, \leftarrow \, \mathcal{N} (0, I_{t - 1})$, $A \, \leftarrow \, \mathds{1} \{ U[0, 1] \leq \Phi(\beta_{AC}^{\top} Z) \}$, 
    \[
        M \, \leftarrow \, B_{M C}^{\top} Z + \beta_{MA} A + B_{MM}^{\top} M + \mathcal{N} \left(0, \left[ \begin{matrix}
            \sigma_1^2 & &  \\
            & \ddots & \\
            & & \sigma_p^2
        \end{matrix}\right] \right),
    \]
    and $Y \, \leftarrow \, \beta_{YC}^{\top} Z + (\alpha_{YA} A + \beta_{YM}^{\top} M)^{2 / 3} + \mathcal{N} (0, 1)$;
    % \[
    % \begin{aligned}
    %     & Z \, \leftarrow \, \mathcal{N} (0, I_{t - 1}), \quad C_k \, \leftarrow \, (Z_k^2 + \sin Z_k) / \sqrt{2} ~\text{ for }~ k \in [t - 2], \quad C_{t - 1} \, \leftarrow \, Z_{t - 1}, \\
    %     & A \, \leftarrow \, \mathds{1} \{ U[0, 1] \leq \Phi(\beta_{AC}^{\top} Z) \}, \quad M \, \leftarrow \, B_{M C}^{\top} Z + \beta_{MA} A + B_{MM}^{\top} M + \mathcal{N} \left(0, \left[ \begin{matrix}
    %         \sigma_1^2 & &  \\
    %         & \ddots & \\
    %         & & \sigma_p^2
    %     \end{matrix}\right] \right), \\
    %     & Y \, \leftarrow \, \beta_{YC}^{\top} Z + (\alpha_{YA} A + \beta_{YM}^{\top} M)^{2 / 3} + \mathcal{N} (0, 1).
    % \end{aligned}
    % \]

    \item \textit{Discrete $M$:} $C \,\leftarrow \, \mathcal{N} (0, I_{t - 1})$, $A \, \leftarrow \, \mathds{1} \{ U[0, 1] \leq \Phi(\beta_{AC}^{\top} C) \}$, 
    \begin{small}
    \[
        M_j \, \leftarrow \, \mathds{1}\left\{ \operatorname{logit}(U[0, 1]) \leq \Theta_{MC, j:} C + \theta_{MA, j} A + \left[ (I - B_{MM}^{\top})^{-1}\mathcal{N} \left(0, \left[ \begin{matrix}
            \sigma_1^2 & &  \\
            & \ddots & \\
            & & \sigma_p^2
        \end{matrix}\right] \right) \right]_{j} \right\},
    \]
    \end{small}
    and $Y \, \leftarrow \, \beta_{YC}^{\top} C + \alpha_{YA} A + \beta_{YM}^{\top} M + \beta_{YC}^{\top} AC + \beta_{YM}^{\top} AM + \mathcal{N} (0, 1)$.
    % \[
    % \begin{aligned}
    %     & C \,\leftarrow \, \mathcal{N} (0, I_{t - 1}), \quad A \, \leftarrow \, \mathds{1} \big\{ U[0, 1] \leq \Phi(\beta_{AC}^{\top} C) \big\}, \\
    %     & M_j \, \leftarrow \, \mathds{1}\left\{ \operatorname{logit}(U[0, 1]) \leq \Theta_{MC, j:} C + \theta_{MA, j} A + \left[ (I - B_{MM}^{\top})^{-1}\mathcal{N} \left(0, \left[ \begin{matrix}
    %         \sigma_1^2 & &  \\
    %         & \ddots & \\
    %         & & \sigma_p^2
    %     \end{matrix}\right] \right) \right]_{j} \right\}, \\
    %     & Y \, \leftarrow \, \beta_{YC}^{\top} C + \alpha_{YA} A + \beta_{YM}^{\top} M + \beta_{YC}^{\top} AC + \beta_{YM}^{\top} AM + \mathcal{N} (0, 1),
    % \end{aligned}
    % \]
\end{itemize}
Here $\sigma_1^2, \ldots, \sigma_p^2$ are independently drawn from the uniform distribution in $[0.5,1]$, whereas the other setting is the same as previous. In the continuous setting, instead of observing the $Z_i$ 's, we observe $C_i$ as the transformations of $Z_i$. We will always leave out the  interaction and $x^{2 / 3}$ when fitting
each model, and we also assume the link functions are all Probit.
For computations in the continuous $M$ setting, we implement Algorithm \ref{alg_fast_QR}. While in the discrete $M$ context, we employ Algorithm \ref{alg_QR}, setting the Monte Carlo sample size to $L = 100$.

As illustrated in Figure \ref{fig_art_all}, aside from the quadruply estimators (QR), other methods fail to yield consistent results. Furthermore, in most cases, our quadruply estimators exhibit a lower standard error compared to other methods. Thus, this also shows the robustness of our estimators.

%The different estimations for average $DM_j$ and $IM_j$ with $j \in [p]$ for sample size $n$ from $250$ to $1000$ under Algorithm \ref{alg_fast_QR} are presented in Figure \ref{fig_art_all}.

\begin{figure}[!htp]
    \minipage{0.5\textwidth}
         \includegraphics[width=\linewidth]{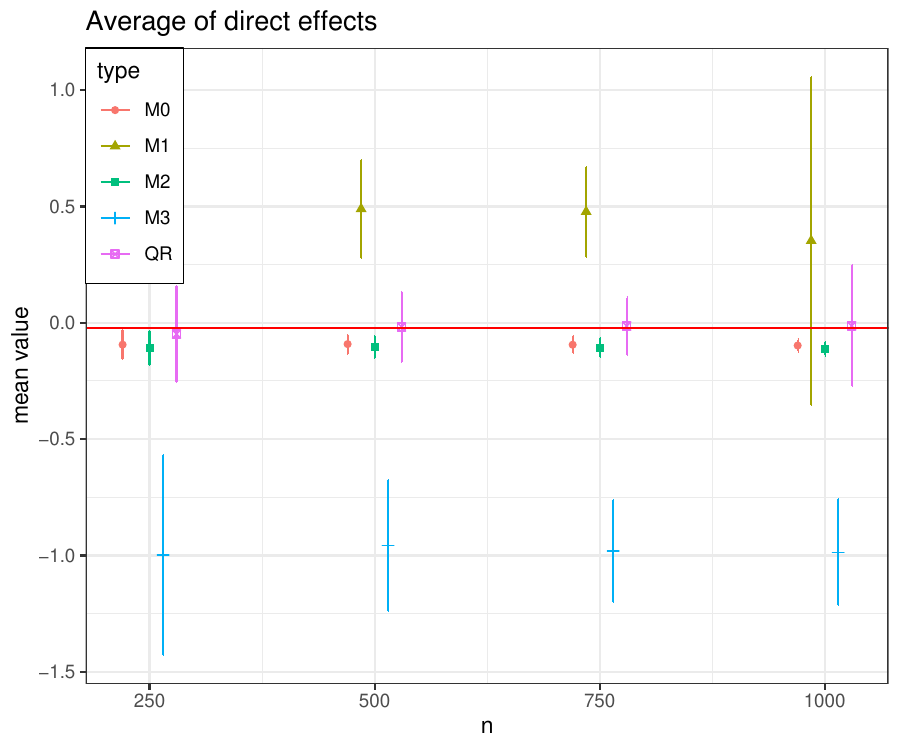}
    \endminipage\hfill
    \minipage{0.5\textwidth}
        \includegraphics[width=\linewidth]{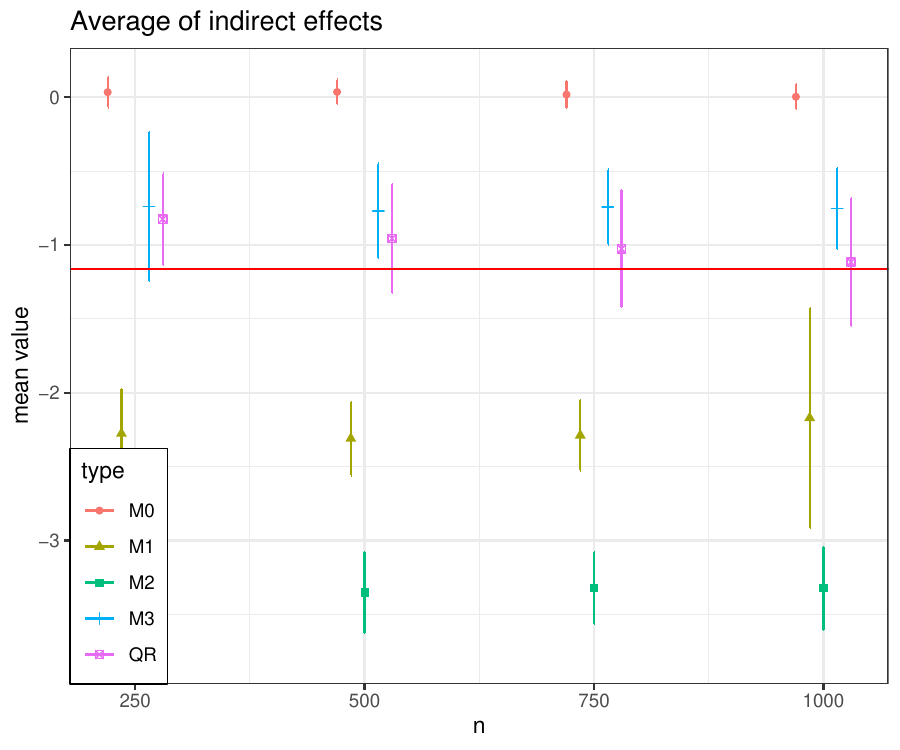}
    \endminipage \\
    \minipage{0.5\textwidth}
         \includegraphics[width=\linewidth]{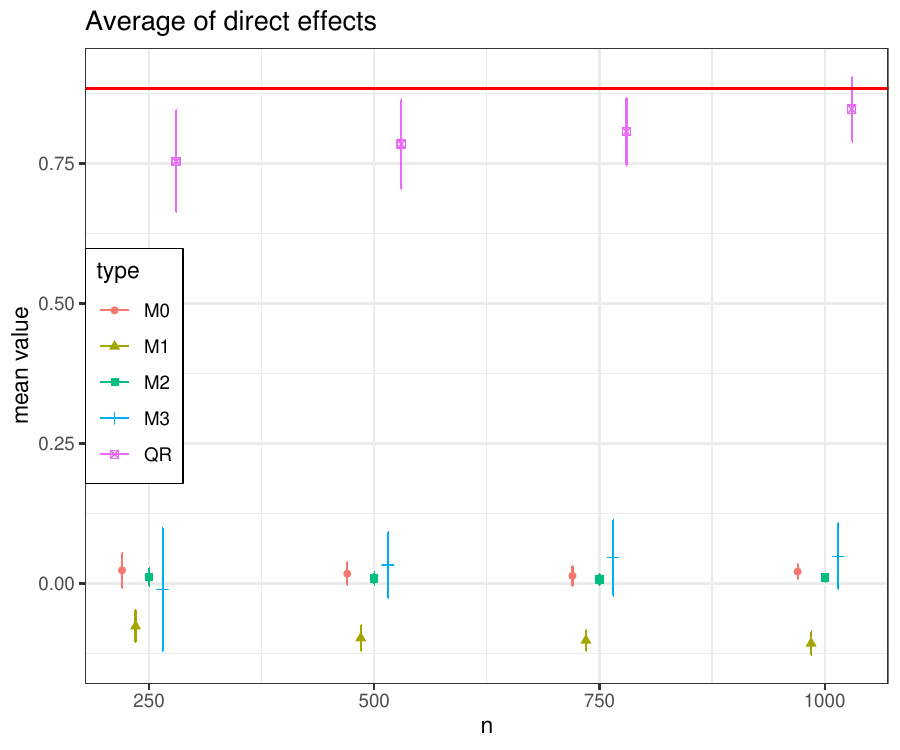}
    \endminipage\hfill
    \minipage{0.5\textwidth}
        \includegraphics[width=\linewidth]{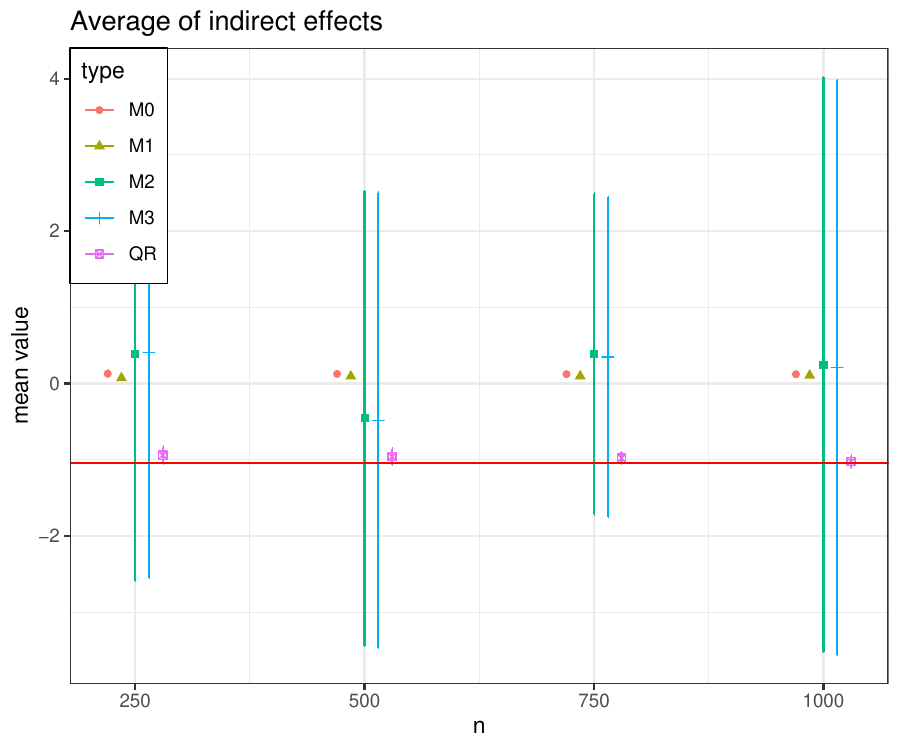}
    \endminipage
    \caption{Simulation results for both continuous (top row) and discrete (bottom row) scenarios, showcasing the estimated average causal direct mediation effect (left column) and indirect mediation effect (right column) over $N = 100$ replications. The dots represent estimated averages, error bars detail the standard error derived from the replications, and the red line represents the true average value, calculated via numeric integrals according to Definition \ref{def_med}.}
    \label{fig_art_all}
\end{figure}

% Finally, similar to Section 6 in \cite{tchetgen2012semiparametric} and Section 4.1 in \cite{xia2023identification}, we consider the case that the structure that totally validates the linear structure assumption as follows:

% with $\sigma_1^2, \ldots, \sigma_p^2$ independently drawn from the uniform distribution in $[0.5,1]$. We leave out all the interaction terms as the potential wrong specification terms for all estimators. 

\section{Empirical Study}\label{sec_emp}

In this section, we illustrate our estimator in a real world application from AURORA study to explore the causal association of psychiatric disorders among trauma survivors, which is also studied in \cite{watson2023heterogeneous}. In the study, our primary response of interest is the post-traumatic stress disorder (PTSD), which was assessed three months post-trauma $Y$. The focal event, in this case, is the pre-trauma insomnia $A$ that trauma survivors often experience: $A = 1$ represents survivor does have insomnia and $A = 0$ represents does not. The 4-dimensional potential mediator $M$ including Peri-traumatic PT (PTSD), stress, acute distress (ASD), and depression, gauged two weeks subsequent to the traumatic incident, are included in our analysis. This study also accounts for various confounders $C$ is a 9-dimensional vector such as age, gender, race, education level, pre-trauma physical and mental health, perceived stress level, neuroticism, and childhood trauma. The same as \cite{watson2023heterogeneous}, before employing our methodology, categorical variables underwent one-hot encoding, numerical variables were centralized, and any missing data was excluded. 
The total number of observations is $n = 1494$ with $t = 10$ and $p = 4$.
The estimated DAG of the mediators by PC algorithm \citep{harris2013pc} is shown in Figure \ref{fig_emp_dag}. Results from the quadruply robust estimators, as obtained using Algorithm \ref{alg_QR} with Monte Carlo sample size $L = 100$ and a bootstrap number of $B = 500$, along with other estimation methods employing a \(\log n\) truncation and the same Monte Carlo sample size and bootstrap number, are presented in Table \ref{tab_emp}.

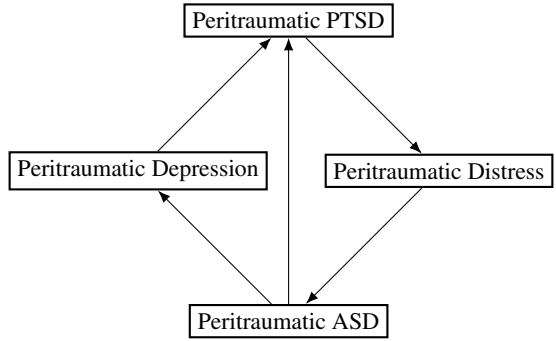
\begin{figure}
    \centering
        \begin{tikzpicture}
            \begin{scope}[every node/.style={thick,draw}]
                 \node (dp) at (0, 0) {Peritraumatic Depression};
                 \node (ptsd) at (2, 2) {Peritraumatic PTSD};
                 \node (asd) at (2, -2) {Peritraumatic ASD};
                 \node (dt) at (4, 0) {Peritraumatic Distress};
             \end{scope}
             \begin{scope}
                 \path (dp) edge node {} (ptsd);
                 \path (ptsd) edge node {} (dt);
                 \path (asd) edge node {} (dp);
                 \path (asd) edge node {} (ptsd);
                 \path (dt) edge node {} (asd);
             \end{scope}
         \end{tikzpicture}
    \caption{Estimated DAG of mediators}\label{fig_emp_dag}
\end{figure}

\begin{table}[!htp]
  \centering
  \caption{Estimated direct and indirect interventional effects of mediators using Algorithm \ref{alg_QR}. Values in bold denote statistical significance at the 95\% confidence level.}
    \resizebox{0.99\textwidth}{!}{%
    \begin{tabular}{ccccccc}
    \toprule
    \toprule
          &       & $\mathscr{M}_0$    & $\mathscr{M}_1$    & $\mathscr{M}_2$    & $\mathscr{M}_3$    & QR \\
    \midrule
    \multirow{2}[1]{*}{distress} & \textit{direct}  & 0.012 (0.022) & 0.564 (17.827) & 0.004 (0.677) & 0.555 (144.040) & 0.080 (0.055) \\
\cmidrule{2-7}          & \textit{indirect} & \multicolumn{1}{l}{0.004 (0.029)} & -12.970 (146.383) & -13.672 (1150.637) & -14.021 (246.881) & 0.101 (0.089) \\
    \midrule
    \multirow{2}[1]{*}{ASD} & \textit{direct}  & 0.134 (0.112) & -3.284 (33.808) & 0.043 (6.680) & 0.557 (162.481) & 0.166 (0.140) \\
\cmidrule{2-7}          & \textit{indirect} & 0.810 (0.633) & {-8.846 (149.805)} & -14.406 (2454.382) & -13.265 (264.681) & \textbf{0.781 (0.278)} \\
    \midrule
    \multirow{2}[1]{*}{PTSD} & \textit{direct}  & 0.604 (0.358) & 0.120 (22.549) & 0.194 (96.854) & {0.711 (164.921)} & \textbf{0.737 (0.278)} \\
\cmidrule{2-7}          & \textit{indirect} & 0.004 (0.164) & {-12.742 (153.783)} & -14.883 (2367.793) & -13.630 (313.450) & 0.012 (0.093) \\
    \midrule
    \multirow{2}[1]{*}{depression} & \textit{direct}  & 0.104 (0.087) & -3.410 (56.095) & 0.034 (5.875) & 1.108 (142.886) & 0.134 (0.105) \\
\cmidrule{2-7}          & \textit{indirect} & 0.120 (0.280) & -8.307 (148.382) & -13.731 (1522.495) & -14.517 (247.827) & 0.198 (0.108) \\
    \bottomrule
    \bottomrule
    \end{tabular}%
  }
  \label{tab_emp}%
\end{table}%

% \noindent As illustrated in Table \ref{tab_emp}, the OLS estimation reveals that both direct and indirect effects of the mediators are non-significant.
% %This aligns with the findings presented in \cite{watson2023heterogeneous}.
% In contrast, when using the both double robust and quadruply robust estimation, the indirect causal effect of acute distress (ASD) emerges as significant, whereas the other effects retain their non-significant status.

\noindent As demonstrated in Table \ref{tab_emp}, for each mediator under consideration, a substantial discrepancy is observed between the estimates of $\mathscr{M}_{\ell_1}$ and those of $\mathscr{M}_{\ell_2}$ for $\ell_2 \neq \ell_1$ when employing any of the four estimation methods $\mathscr{M}_{\ell}$ for $\ell = 0, 1, 2, 3$. 
Moreover, none of these estimation methods manage to identify significant direct or indirect interventional effects for any of the mediators. 
This highlights the pressing need for robust estimation approaches in this dataset. Notably, with the quadruply robust estimation, we discern that both the indirect interventional effect of acute distress and the direct interventional effect of peritraumatic PT are significant at the 95\% confidence level, while other effects remain non-significant. These results also indicate that preventive intervention of 3-month PTSD after trauma exposure that focuses on reducing acute distress and peritraumatic PT is more likely to be effective for trauma survivors.  

% \todohw{Add the meaning of this finding?}

\section{Discussion}\label{sec_dis}

The main contribution of this article is the introduction of direct and indirect interventional effects of mediators, alongside their semiparametric bounds and quadruply robust estimators. Our method accommodates continuous, categorical, and multivariate pre-treatments, mediators, and outcomes. Moreover, extending our methodology and theory to polytomous exposures is straightforward.
However, extending to continuous exposure, even under LSEMs, is non-trivial in theoretical sense. A potential method is suggested in \cite{cai2021deep} to replace the indicator function $\mathds{1}(A = a)$ with some kernel function $K ((A - a) / h)$ under bandwidth $h$, but as discussed in \cite{diaz2013targeted, kennedy2017non, kennedy2023semiparametric}, pathwise differentiability will fail in this case, necessitating alternative estimation procedures and techniques.
On the other hand, note that our framework is dimensional-free, as long as the conditional densities and expectations meet the mild convergence rate, our quadruply robust estimations will always achieve semiparametric efficiency. However, in high-dimensional cases, non-parametric estimation mentioned in this article may not achieve the rate, necessitating additional assumptions like symmetry and shape constraints, as discussed in \cite{deng2021confidence, xu2021high, rodriguez2022data}. Introducing these assumptions still validates the semiparametric framework in our article under the full nonparametric model $\mathscr{M}_{\text{nonpar}}$, but our quadruply estimators may not be the most efficient under these added conditions. 

\bibliographystyle{chicago}
\bibliography{ref}

\appendix

\section{Glossary of Terms and Notations}\label{app_glossary}

\renewcommand*{\arraystretch}{1.3}
\begin{longtable}{|c|m{3.5cm}|c|m{3.5cm}|}
\hline
\hline
\textbf{Symbol} & \textbf{Definition} & \textbf{Symbol} & \textbf{Definition} \\
\hline
$X$ & Observed variables  & $\mathcal{G}$ & Causal DAG (or corresponding adjacency matrix) of $X$ \\
$C$ & Confounders & $\mathcal{C}$ & Causal CPDAG (or corresponding adjacency matrix) of $X$ \\
$A$ & $0$-$1$ exposure & $\operatorname{MEC}(\mathcal{C})$ & Adjacency matrix in the Markov equivalence class of $\mathcal{C}$ \\
$M$ & Mediators & $\mathcal{G}_M$ & Causal DAG (or corresponding adjacency matrix) of $M$ \\
$Y$ & Univariate outcome & $\mathcal{C}_M$ & Causal CPDAG (or corresponding adjacency matrix) of $M$ \\
$\Pa_j(\mathcal{G}_M), \pa_j(\mathcal{G}_M)$ & Parents of $M_j$ in mediators $\{ M_1, \ldots, M_p\}$, its realization & $e_{a'}(x_S)$ & $\pr (A = a' \mid X_S  = x_S)$ \\
% $\mu(x_{S_1}, \ldots, x_{S_N})$ & $\E [Y \mid X_{S_1} = x_{S_1}, \ldots, X_{S_N} = x_{S_N}]$ & $\pi_{x_{S_1}, \ldots, x_{S_N}} (m_T)$ & $f_{M_T \mid X_{S_1}, \ldots, X_{S_N}} (m_T \mid  X_{S_1} = x_{S_1}, \ldots, X_{S_N} = x_{S_N})$ \\
$\mu(\boldsymbol\cdot)$ & Conditional expectation of $Y$ given $\boldsymbol\cdot$ & $\pi_{\boldsymbol\cdot} (m_T)$ & Conditional density of $M_T$ given $\boldsymbol\cdot$ evaluated at $m_T$ \\
$S^{\text{eff, nonpar}}(\boldsymbol\cdot)$ & Efficient score for $\boldsymbol\cdot$ on full model & $\pn [\boldsymbol\cdot]$ & $\frac{1}{n} \sum_{i = 1}^n [\boldsymbol\cdot]_i$ \\

$\kappa (a', C)$ & Formula \eqref{def_kappa} & $\zeta_j(a', a, C)$ & Formula \eqref{def_zeta} \\
$\varrho_j(a', M_j, C \, ; \, \mathcal{G}_M)$ & Formula \eqref{def_varrho} & $\big\langle g_{\boldsymbol{\cdot}, u} (\boldsymbol{\cdot}, v) \big\rangle$ & Difference in $g_{\boldsymbol{\cdot}, u} (\boldsymbol{\cdot}, v)$ evaluated at different exposure defined in Equation \eqref{def_diff_notation} \\
$\tau_{\boldsymbol{\cdot} \, ; \, S} (C, a', M_T)$ & Formula \eqref{def_tau}  \\
\hline
\hline
\end{longtable}

\section{Explicit formulas for Algorithm \ref{alg_QR} and \ref{alg_fast_QR}}\label{app_alg}

We first give the formulas in Algorithm \ref{alg_QR}, the formulas in Algorithm \ref{alg_fast_QR} will be shown in the proof of Proposition \ref{pro_fast} later in this section. For calculating the formulas in Algorithm \ref{alg_QR}, by Monte Carlo approximation, we have
\[
    \begin{aligned}
        \widehat\zeta_j^{\mathscr{M}_0}(a', 0, C) & = \int_{\mathcal{M}} \widehat{\mu}(C, 1, m) \widehat\pi_{C, a'}(m_j) \widehat\pi_{C, 0}(m_{-j}) \, \mathrm{d} m \\
        & = \int_{\mathcal{M}} \widehat{\mu}(C, 1, m) \frac{\widehat\pi_{C, a'}(m_j) \widehat\pi_{C, 0}(m_{-j})}{\widehat\pi_{C, 1}(m)} \widehat\pi_{C, 1}(m) \, \mathrm{d} m \\
        & \approx \frac{1}{N} \sum_{\ell = 1}^N \widehat{\mu}(C, 1, M_{\widehat{\pi}_{C, 0}}^{(\ell)}) \frac{\widehat\pi_{C, a'}(M_{\widehat{\pi}_{C, 0}, j}^{(\ell)}) \widehat\pi_{C, 0}(M_{\widehat{\pi}_{C, 0}, -j}^{(\ell)})}{\widehat\pi_{C, 1}(M_{\widehat{\pi}_{C, 1}}^{(\ell)})}.
    \end{aligned}
\]
where the sampling density $\widehat{\pi}_{C, 1} (m)$ is calculated by $\widehat{\pi}_{C, 1} (m) = \widehat{\pi}_{C, 1}(m_j) \widehat{\pi}_{C, 1, m_j}(m_{-j})$ and does not need to estimate additionally.
Similarly, we have
\[
    \begin{aligned}
        %& \int_{\mathcal{M}_j} \widehat\mu (C, 1, m_j, M_{-j}) \Big( \widehat\pi_{C, 1}(m_j)  - \widehat\pi_{C, 0}(m_j) \Big) \, \mathrm{d} m_j \\
        \big\langle \widehat\tau_{C, \boldsymbol\cdot \, ; \, j}(C, 1, M_{-j}) \big\rangle \approx \frac{1}{N} \sum_{\ell = 1}^N \left[  \widehat\mu (C, 1, M_{j, \widehat\pi_{C, 1}}^{(\ell)}, M_{-j}) -  \widehat\mu (C, 1, M_{j, \widehat\pi_{C, 0}}^{(\ell)}, M_{-j}) \right]
    \end{aligned}
\]
and
\begin{small}
\[
    \begin{aligned}
        %& \int_{\mathcal{M}_{-j}} \widehat{\mu} (C, 1, M_{j}, m_{-j})  \widehat{\pi}_{C, 0} (m_{-j}) \, \mathrm{d} m_{-j} 
        \widehat{\tau}_{C, 0 \, ; \, -j} (C, 1, M_{j} ) \approx \frac{1}{N} \sum_{\ell = 1}^N \widehat{\mu} (C, 1, M_{j}, M_{-j, \widehat{\pi}_{C, 0}}^{(\ell)}).
    \end{aligned}
\]
Note that we can rewrite $\widehat{DM}_j^{\text{QR}} = \frac{1}{n} \sum_{i = 1}^n \widehat{QR}_{DM_j}(X_i)$ with
\[
    \begin{aligned}
        \widehat{QR}_{DM_j}(X) & :=  \frac{\mathds{1}(A = 1)}{\widehat{e}_1(C)} \frac{\big\langle\widehat\pi_{C, \bcdot}(M_{-j}) \big\rangle}{\widehat\pi_{C, 1, M_j}(M_{-j})} \big[ Y - \widehat\mu (C, 1, M) \big] + \frac{\mathds{1}(A = 0)}{\widehat{e}_0(C)}  \big\langle \widehat\tau_{C, \boldsymbol\cdot \, ; \, j}(C, 1, M_{-j}) \big\rangle \\
        %& ~~~~~ + \frac{\mathds{1}(A = 0)}{\widehat{e}_0(C)}  \int_{\mathcal{M}_j} \widehat\mu (C, 1, m_j, M_{-j}) \Big( \widehat\pi_{C, 1}(m_j)  - \widehat\pi_{C, 0}(m_j) \Big) \, \mathrm{d} m_j  \\
        % & ~~~~~ +  \bigg[ \frac{\mathds{1}(A = 1)}{\widehat{e}_1(C)} -  \frac{\mathds{1}(A = 0)}{\widehat{e}_0(C)} \bigg]  \int_{\mathcal{M}_{-j}} \widehat{\mu} (C, 1, M_{j}, m_{-j})  \widehat{\pi}_{C, 0} (m_{-j}) \, \mathrm{d} m_{-j} \\
        & ~~~~~ +  \bigg\langle \frac{\mathds{1}(A = \bcdot)}{\widehat{e}_{\bcdot}(C)} \bigg\rangle \widehat{\tau}_{C, 0\, ; \, -j} (C, 1, M_{j}) + \bigg[ 1 -  \frac{\mathds{1}(A = 0)}{\widehat{e}_0(C)} - \frac{\mathds{1}(A = 1)}{\widehat{e}_1(C)} \bigg] \widehat\zeta_j^{\mathscr{M}_0}(1, 0, C) \\
        & ~~~~~  + \left[ 2 \frac{\mathds{1}(A = 0)}{\widehat{e}_0(C)} - 1 \right] \widehat\zeta_j^{\mathscr{M}_0}(0, 0, C).
    \end{aligned}
\]
\end{small}
Therefore, we have the approximation as $\widehat{DM}_j^{\text{QR}} \approx \frac{1}{n} \sum_{i = 1}^n \widehat{QR}_{DM_j}^{\text{MC}}(X_i)$ with the following explicit formula
\begin{small}
\begin{equation}\label{DM_QR_MC}
    \begin{aligned}
        \widehat{QR}_{DM_j}^{\text{MC}}(X) & :=  \frac{\mathds{1}(A = 1)}{\widehat{e}_1(C)} \bigg[ \frac{\widehat\pi_{C, 1}(M_{-j}) -\widehat\pi_{C, 0}(M_{-j})}{\widehat\pi_{C, 1, M_j}(M_{-j})} \bigg] \big\{ Y - \widehat\mu (C, 1, M) \big\} \\
        & ~~~~~ + \frac{\mathds{1}(A = 0)}{\widehat{e}_0(C)}   \frac{1}{N} \sum_{\ell = 1}^N \left[  \widehat\mu (C, 1, M_{j, \widehat\pi_{C, 1}}^{(\ell)}, M_{-j}) -  \widehat\mu (C, 1, M_{j, \widehat\pi_{C, 0}}^{(\ell)}, M_{-j}) \right]  \\
        & ~~~~~ +  \bigg[ \frac{\mathds{1}(A = 1)}{\widehat{e}_1(C)} -  \frac{\mathds{1}(A = 0)}{\widehat{e}_0(C)} \bigg] \frac{1}{N} \sum_{\ell = 1}^N \widehat{\mu} (C, 1, M_{j}, M_{-j, \widehat{\pi}_{C, 0}}^{(\ell)}) \\
        & ~~~~~ + \bigg[ 1 -  \frac{\mathds{1}(A = 0)}{\widehat{e}_0(C)} - \frac{\mathds{1}(A = 1)}{\widehat{e}_1(C)} \bigg] \frac{1}{N} \sum_{\ell = 1}^N \widehat{\mu}(C, 1, M_{\widehat{\pi}_{C, 0}}^{(\ell)}) \frac{\widehat\pi_{C, 1}(M_{\widehat{\pi}_{C, 0}, j}^{(\ell)}) \widehat\pi_{C, 0}(M_{\widehat{\pi}_{C, 0}, -j}^{(\ell)})}{\widehat\pi_{C, 1}(M_{\widehat{\pi}_{C, 1}}^{(\ell)})} \\
        & ~~~~~ + \left[ 2 \frac{\mathds{1}(A = 0)}{\widehat{e}_0(C)} - 1 \right] \frac{1}{N} \sum_{\ell = 1}^N \widehat{\mu}(C, 1, M_{\widehat{\pi}_{C, 0}}^{(\ell)}) \frac{\widehat\pi_{C, 0}(M_{\widehat{\pi}_{C, 0}, j}^{(\ell)}) \widehat\pi_{C, 0}(M_{\widehat{\pi}_{C, 0}, -j}^{(\ell)})}{\widehat\pi_{C, 1}(M_{\widehat{\pi}_{C, 1}}^{(\ell)})}.
    \end{aligned}
\end{equation}
\end{small}
Similarly, note that the double integral can be rewritten as
\[
    \begin{aligned}
        \E \big[ \varrho_j (a', M_j, C \, ; \, \mathcal{G}_M) \mid C \big] & = \int \mu \big( a', m_j, \pa_j(\mathcal{G}_M), C \big) \pi_{C, a'} \big(\pa_j(\mathcal{G}_M)  \big) \pi_C (m_j) \, \mathrm{d} (m_j, \pa_j(\mathcal{G}_M)) \\
        & \approx \frac{1}{N^2} \sum_{\ell_1 = 1}^N \sum_{\ell_2 = 1}^N \mu \big( a', M_{j, \pi_{C}}^{(\ell_1)}, \Pa_{j, \pi_{C, a', \mathcal{G}_M}}^{(\ell_2)}, C\big).
    \end{aligned}
\]
Thus, given the fact $\widehat{TM}_j^{\text{QR}} (\mathcal{G}_M)= \frac{1}{n} \sum_{i = 1}^n \widehat{QR}_{TM_j}(X_i \, ; \, \mathcal{G}_M)$ with
\begin{small}
\[
    \begin{aligned}
        \widehat{QR}_{TM_j}(X \, ; \, \mathcal{G}_M) & := \left\langle \frac{\mathds{1}(A = \bcdot)}{\widehat{e}_{\bcdot}(C)} \big[ Y - \widehat{\mu} (C, \bcdot) \big] -  \frac{\mathds{1} (A = \bcdot)}{\widehat{e}_{\bcdot}(C)} \widehat\pi_{C, \bcdot} \big(\Pa_j (\mathcal{G}_M) \big) \Big( Y - \widehat{\mu} (C, \bcdot, \Pa_j(\mathcal{G}_M), M_j)\Big) \right\rangle \\
        & ~~~~~~~~~ - \left\langle  \frac{\mathds{1} (A = \bcdot)}{\widehat{e}_{\bcdot}(C)} \Big\{ \widehat\tau_{C \, ; \, j} \big( C, 1, \Pa_j (\mathcal{G}_M) \big) - {\E} \big[ \widehat\varrho_j (\bcdot, M_j, C\, ; \, \mathcal{G}_M) \mid C  \big] \Big\} \right\rangle  \\
        & ~~~~~~~~~ + \left\langle \widehat{\mu}(C, \bcdot) - \widehat\tau_{C, \bcdot \, ; \, \pa_j(\mathcal{G}_M)} (C, \bcdot , M_j) \right\rangle,
    \end{aligned}
\]
\end{small}
we have the approximation $\widehat{DM}_j^{\text{QR}} \approx \frac{1}{n} \sum_{i = 1}^n \widehat{QR}_{TM_j}^{\text{MC}}(X_i)$, where 
\begin{small}
\begin{equation}\label{TM_QR_MC}
    \begin{aligned}
        \widehat{QR}_{TM_j}^{\text{MC}}(X \, ; \, \mathcal{G}_M) := & \Bigg[ \frac{\mathds{1}(A = 1)}{\widehat{e}_1(C)} \Big\{ Y - \widehat{\mu} (C, 1) \Big\} - \frac{\mathds{1}(A = 0)}{\widehat{e}_0(C)} \Big\{ Y - \widehat{\mu} (C, 0) \big] \Big\} \Bigg] \\
        & - \Bigg[ \frac{\mathds{1} (A = 1)}{\widehat{e}_1(C)} \widehat\pi_{C, 1} \big(\Pa_j (\mathcal{G}_M) \big) \Big\{ Y - \widehat{\mu} (C, 1, \Pa_j(\mathcal{G}_M), M_j)\Big\} \\
        & ~~~~~~~~~~~~~~~~~  - \frac{\mathds{1} (A = 0)}{\widehat{e}_0(C)}  \widehat\pi_{C, 0} \big(\Pa_j (\mathcal{G}_M) \big) \Big\{ Y - \widehat{\mu} (C, 0, \Pa_j(\mathcal{G}_M), M_j) \Big\} \\
        & + \frac{1}{N^2} \sum_{\ell_1 = 1}^N \sum_{\ell_2 = 1}^N \Big[ \widehat\mu \big(C, 1, M_{j, \widehat\pi_{C}}^{(\ell_1)}, \Pa_{j} (\mathcal{G}_M)\big) - \widehat\mu \big(C, 1, M_{j, \widehat\pi_{C}}^{(\ell_1)}, \Pa_{j, \widehat\pi_{C, 1, \mathcal{G}_M}}^{(\ell_2)}\big) \Big] \\
        & - \frac{1}{N^2} \sum_{\ell_1 = 1}^N \sum_{\ell_2 = 1}^N \Big[ \widehat\mu \big( C, 0, M_{j, \widehat\pi_{C}}^{(\ell_1)}, \Pa_{j} (\mathcal{G}_M) \big) - \widehat\mu \big(C, 0, M_{j, \widehat\pi_{C}}^{(\ell_1)}, \Pa_{j, \widehat\pi_{C, 0, \mathcal{G}_M}}^{(\ell_2)}, C\big) \Big] \Bigg] \\
        & + \widehat{\mu}(C, 1) - \widehat{\mu}(C, 0)  - \bigg[ \frac{1}{N} \sum_{\ell = 1}^N \Big[ \widehat\mu \big( C, 1, \Pa_{j, \widehat\pi_{C, 1, \mathcal{G}_M}}^{(\ell)}, M_j \big) - \widehat\mu \big( C, 0, \Pa_{j, \widehat\pi_{C, 0, \mathcal{G}_M}}^{(\ell)}, M_j \big) \bigg].
    \end{aligned}
\end{equation}
\end{small}
\noindent \textbf{Proof of Proposition \ref{pro_fast}:}

\begin{proof}
    Denote $\phi(\bcdot \, ; \mu, \Sigma)$ as the density of distribution $\mathcal{N} (\mu, \Sigma)$. Similar to \eqref{fast_imple_eq1}, we have
    \[
        \begin{aligned}
            & \int_{\mathcal{M}_j} \mu (C, 1, m_j, M_{-j}) \pi_{C, 1} (m_j) \, \mathrm{d} m_j \\
		= & \int_{\mathcal{M}_j} \Big[ \beta_{YC} C + \alpha_{YA}  + \beta_{YM, j} m_j + \beta_{YM, -j}^{\top} M_{-j} \Big]  \phi \Big(m_j \, ; \, \big[{\Theta}_{MC} C \big]_j + \theta_{MA, j}, \, \big[ \var(e_M) \big]_{jj} \Big) \, \mathrm{d} m_j \\
		= & \beta_{YC} C + \alpha_{YA}+ \beta_{YM, j} \Big\{ \big[\widehat{\Theta}_{MC} C \big]_j \Big\} + \beta_{YM, -j}^{\top} M_{-j},
        \end{aligned}
    \]
    thus
    \[
    	\begin{aligned}
    		\big\langle \tau_{C, \boldsymbol\cdot \, ; \, j}(C, 1, M_{-j}) \big\rangle & = \Big[ \beta_{YC} C + \alpha_{YA}+ \beta_{YM, j} \Big\{ \big[{\Theta}_{MC} C \big]_j + \theta_{MA, j}\Big\} + \beta_{YM, -j}^{\top} M_{-j} \Big]  \\
            & \qquad - \Big[ \beta_{YC} C + \alpha_{YA}+ \beta_{YM, j} \big[{\Theta}_{MC} C \big]_j + \beta_{YM, -j}^{\top} M_{-j} \Big] = \beta_{YM, j}\theta_{MA, j}.
    	\end{aligned}
    \]
    Similarly, we have 
    \[
        \begin{aligned}
            & \int_{\mathcal{M}_{-j}} {\mu} (C, 1, M_j, m_{-j}) \pi_{C, 0} (m_{-j}) \, \mathrm{d} m_{-j} \\
            = & \int_{\mathcal{M}_j} \Big[ \beta_{YC} C + \alpha_{YA}  + \beta_{YM, j} M_j + \beta_{YM, -j}^{\top} m_{-j} \Big]  \phi \Big(m_{-j} \, ; \, \big[{\Theta}_{MC} C \big]_{-j}, \, \big[ \var(e_M) \big]_{-j, -j} \Big) \, \mathrm{d} m_{-j} \\
		= & \beta_{YC} C + \alpha_{YA}+ \beta_{YM, j} M_j + \beta_{YM, -j}^{\top} \big[{\Theta}_{MC} C \big]_{-j}
        \end{aligned}
    \]
    and
    \[
        \begin{aligned}
            & \zeta_j(a', 0, C)\\
            = & \int_{\mathcal{M}} \mu(C, 1, m) \pi_{C, a'}(m_j) \pi_{C, 0} (m_{-j}) \, \mathrm{d} m \\
            = & \int_{\mathcal{M}} \Big[ \beta_{YC} C + \alpha_{YA}  + \beta_{YM, j} m_j + \beta_{YM, -j}^{\top} m_{-j} \Big] \phi \Big(m_{j} \, ; \, \big[{\Theta}_{MC} C \big]_{j} + \theta_{MA, j} a', \, \big[ \var(e_M) \big]_{j, j} \Big) \\
            & \qquad \qquad \qquad \qquad \phi \Big(m_{-j} \, ; \, \big[{\Theta}_{MC} C \big]_{-j} , \, \big[ \var(e_M) \big]_{- j, - j} \Big) \, \mathrm{d} m_j \, \mathrm{d} m_{- j} \\
            = & \beta_{YC} C + \alpha_{YA} + \beta_{YM, j} \Big\{ \big[{\Theta}_{MC} C \big]_{j} + \theta_{MA, j} a' \Big\} + \beta_{YM, -j}^{\top} \big[{\Theta}_{MC} C \big]_{-j}.
        \end{aligned}
    \]
    Therefore, under Assumption \ref{ass_slsem} and Gaussian assumption,
    % \[
    %     \widehat{DM}_j^{\text{DR}} = \pn \widehat{DR}_{DM_j}^{\text{fast}}(X) + \widehat{DM}_j^{\text{OLS}},
    % \]
    \[
        \widehat{DM}_j^{\text{QR}} = \pn \widehat{QR}_{DM_j}^{\text{fast}}(X) + \widehat{DM}_j^{\text{OLS}},
    \]
    where the (estimated) quadruply adjustment terms $\widehat{QR}_{DM_j}(X)$ is defined as
    % \begin{equation}\label{DM_DR_fast}
    % \begin{aligned}
    %     & \widehat{DR}_{DM_j}^{\text{fast}}(X) = \widehat{DR}_{DM_j}^{\text{fast}} (C, A, M, Y) \\
    %     = &  \frac{\mathds{1}(A = 1)}{\widehat{e}_1(C)} \bigg[ \frac{\phi \big(M_{-j} ; \big[ \widehat{\Theta}_{MC} C \big]_{- j} + \widehat\theta_{MA, - j}, \big[ \widehat\var(e_M) \big]_{-j, -j} \big) -\phi \big(M_{-j} ; \big[ \widehat{\Theta}_{MC} C \big]_{- j}, \big[ \widehat\var(e_M) \big]_{-j, -j} \big)}{\phi \big(M_{-j} \mid M_j ; \big[ \widehat{\Theta}_{MC} C + \widehat\theta_{MA}  \big]_{-j \, \mid \, j}, \big[ \widehat\var(e_M) \big]_{-j \, \mid \, j} \big)} \bigg] \\
    %     & ~~~~~~~~~~~~  \times \Big\{ Y - \big[ \widehat\beta_{YC} C + \widehat\alpha_{YA} + \widehat\beta_{YM}^{\top} M \big] \Big\},
    % \end{aligned}
    % \end{equation}
    % and
    \begin{small}
    \begin{equation}\label{DM_QR_fast}
    \begin{aligned}
        & \widehat{QR}_{DM_j}^{\text{fast}}(X) = \widehat{QR}_{DM_j}^{\text{fast}} (C, A, M, Y) \\
        = &  \frac{\mathds{1}(A = 1)}{\widehat{e}_1(C)} \bigg[ \frac{\phi \big(M_{-j} ; \big[ \widehat{\Theta}_{MC} C \big]_{- j} + \widehat\theta_{MA, - j}, \big[ \widehat\var(e_M) \big]_{-j, -j} \big) -\phi \big(M_{-j} ; \big[ \widehat{\Theta}_{MC} C \big]_{- j}, \big[ \widehat\var(e_M) \big]_{-j, -j} \big)}{\phi \big(M_{-j} \mid M_j ; \big[ \widehat{\Theta}_{MC} C + \widehat\theta_{MA}  \big]_{-j \, \mid \, j}, \big[ \widehat\var(e_M) \big]_{-j \, \mid \, j} \big)} \bigg] \\
        & ~~~~~~~~~~~~~~~~~~~~~~~~~~~~~~ \times \Big\{ Y - \big[ \widehat\beta_{YC} C + \widehat\alpha_{YA} + \widehat\beta_{YM}^{\top} M \big] \Big\} \\
        & ~~~~~ + \bigg[ \frac{\mathds{1}(A = 1)}{\widehat{e}_1(C)} -  \frac{\mathds{1}(A = 0)}{\widehat{e}_0(C)} \bigg] \Big[ \widehat\beta_{YC} C + \widehat\alpha_{YA}+ \widehat\beta_{YM, j} M_j + \widehat\beta_{YM, -j}^{\top} \big[{\widehat\Theta}_{MC} C \big]_{-j} \Big] \\
        & ~~~~~ - \bigg[ \frac{\mathds{1}(A = 1)}{\widehat{e}_1(C)} \Big[ \widehat\beta_{YC}^{\top} C + \widehat\alpha_{YA} + \widehat\beta_{YM, j} \Big\{ \big[{\widehat\Theta}_{MC} C \big]_{j} + \widehat\theta_{MA, j}  \Big\} + \widehat\beta_{YM, -j}^{\top} \big[{\widehat\Theta}_{MC} C \big]_{-j} \Big] \\
        & ~~~~~~~~~~~~~ - \frac{\mathds{1}(A = 0)}{\widehat{e}_0(C)} \Big[ \widehat\beta_{YC}^{\top} C + \widehat\alpha_{YA} + \widehat\beta_{YM, j}  \big[{\widehat\Theta}_{MC} C \big]_{j}  + \widehat\beta_{YM, -j}^{\top} \big[{\widehat\Theta}_{MC} C \big]_{-j} \Big] \bigg].
    \end{aligned}
    \end{equation}
    \end{small}
    Similarly, we can show that 
    \begin{small}
    \[
        \begin{aligned}
            \pi_C(M_j) & = e_0(C) \pi_{C, 0} (M_j) + e_1(C) \pi_{C, 1} (M_j) \\
            & = e_0(C) \phi \Big(m_{j} \, ; \, \big[{\Theta}_{MC} C \big]_{j}, \big[ \var(e_M) \big]_{j, j} \Big) + e_1(C) \phi \Big(m_{j} \, ; \, \big[{\Theta}_{MC} C \big]_{j} + \theta_{MA, j}, \big[ \var(e_M) \big]_{j, j} \Big),
        \end{aligned}
    \]
    \[
        \begin{aligned}
            \pi_{C, a'} \big( \Pa_j (\mathcal{G}_M)\big) = \phi \Big( \Pa_j (\mathcal{G}_M)  \, ; \, \big[{\Theta}_{MC} C \big]_{k : k \in \pa_j(\mathcal{G}_M)} + \theta_{MA,  k \in \pa_j(\mathcal{G}_M)} a', \big[ \var(e_M) \big]_{kk : k \in \pa_j(\mathcal{G}_M)} \Big),
        \end{aligned}
    \]
    \[
        \begin{aligned}
            %& \int_{\mathcal{M}_j} \mu(C, a', \Pa_j (\mathcal{G}_M), m_j) \pi_C(m_j) \, \mathrm{d} m_j \\
            & \tau_{C \, ; \, j}\big( C, a', \Pa_j (\mathcal{G}_M) \big)  =  \int_{\mathcal{M}_j} \Big\{ \eta_{YM_j} m_j + \gamma_{Y \Pa_j(\mathcal{G}_M)}^{\top} \Pa_j(\mathcal{G}_M) + \eta_{YA} a' + \gamma_{YC}^{\top} C\Big\}\\
            & ~~~~~~~ \Big\{ e_0(C) \phi \Big(m_{j} \, ; \, \big[{\Theta}_{MC} C \big]_{j}, \big[ \var(e_M) \big]_{j, j} \Big) + e_1(C) \phi \Big(m_{j} \, ; \, \big[{\Theta}_{MC} C \big]_{j} + \theta_{MA, j}, \big[ \var(e_M) \big]_{j, j} \Big) \Big\} \, \mathrm{d} m_j \\
            = & \gamma_{Y \Pa_j(\mathcal{G}_M)}^{\top} \Pa_j(\mathcal{G}_M) + \eta_{YA} a' + \gamma_{YC}^{\top} C  + \eta_{YM_j} \Big[ e_0(C) \big[{\Theta}_{MC} C \big]_{j} + e_1 (C) \big( \big[{\Theta}_{MC} C \big]_{j} + \theta_{MA, j}\big) \Big],
        \end{aligned}
    \]
    \end{small}
    % \[
    %     \begin{aligned}
    %         & \varrho_j (a', M_j, C \, ; \, \mathcal{G}_M) \\
    %         = & \int_{\mathcal{M}_{\pa_j} (\mathcal{G}_M)} \mu(C, a', \pa_j, M_j) \, \pi_{C, a'}(\pa_j)  \, \mathrm{d} \pa_j \\
    %         = & \int_{\mathcal{M}_{\pa_j} (\mathcal{G}_M)} \Big\{ \eta_{YM_j} M_j + \gamma_{Y \Pa_j(\mathcal{G}_M)} \pa_j + \eta_{YA} a' + \gamma_{YC}^{\top} C\Big\} \\
    %         & ~~~~~~~~~~~~ \phi \Big(\pa_{j} \, ; \, \big[{\Theta}_{MC} C \big]_{k: k \in \pa_j(\mathcal{G}_M)} + a' \theta_{MA, k : k \in \pa_j(\mathcal{G}_M)}, \big[ \var(e_M) \big]_{kk, k \in \pa_j(\mathcal{G}_M)} \Big) \, \mathrm{d} \pa_j \\
    %         = & \eta_{YM_j} M_j + \gamma_{Y \Pa_j(\mathcal{G}_M)}^{\top} \Big[ \big[{\Theta}_{MC} C \big]_{k: k \in \pa_j(\mathcal{G}_M)} + a' \theta_{MA, k : k \in \pa_j(\mathcal{G}_M)} \Big] + a' \eta_{YA} + \gamma_{YC}^{\top} C
    %     \end{aligned}
    % \]
    and
    \begin{small}
    \[
        \begin{aligned}
            & {\E} \big[ \varrho_j (a', M_j, C \, ; \, \mathcal{G}_M) \mid C \big] \\
            = & \int_{\mathcal{M}_{\pa_j} (\mathcal{G}_M)} \, \mathrm{d} \pa_j(\mathcal{G}_M) \int_{\mathcal{M}_j} {\mu} \big(C, a', \pa_j(\mathcal{G}_M), m_j \big) {\pi}_{C, a'} \big(\pa_j( \mathcal{G}_M) \big)  {\pi}_{C} (m_j)  \, \mathrm{d} m_j \\
            = & \int_{\mathcal{M}_{\pa_j} (\mathcal{G}_M)} \, \mathrm{d} \pa_j(\mathcal{G}_M)  \int_{\mathcal{M}_j} \Big\{ \eta_{YM_j} m_j + \gamma_{Y \Pa_j(\mathcal{G}_M)}^{\top} \pa_j + \eta_{YA} a' + \gamma_{YC}^{\top} C\Big\} \\
            & ~~~~ \times \phi \Big( \pa_j  \, ; \, \big[{\Theta}_{MC} C \big]_{k : k \in \pa_j(\mathcal{G}_M)} + \theta_{MA,  k \in \pa_j(\mathcal{G}_M)} a', \big[ \var(e_M) \big]_{kk : k \in \pa_j(\mathcal{G}_M)} \Big) \\
            & ~~~~\times\Big\{ e_0(C) \phi \Big(m_{j} \, ; \, \big[{\Theta}_{MC} C \big]_{j}, \big[ \var(e_M) \big]_{j, j} \Big) + e_1(C) \phi \Big(m_{j} \, ; \, \big[{\Theta}_{MC} C \big]_{j} + \theta_{MA, j}, \big[ \var(e_M) \big]_{j, j} \Big) \Big\} \, \mathrm{d} m_j \\
            = &  \eta_{YA} a' + \gamma_{YC}^{\top} C + \gamma_{Y \Pa_j(\mathcal{G}_M)}^{\top} \Big\{ \big[{\Theta}_{MC} C \big]_{k : k \in \pa_j(\mathcal{G}_M)} + \theta_{MA,  k \in \pa_j(\mathcal{G}_M)} a' \Big\} \\
            & ~~~~~~~~~~ + \eta_{YM_j} \Big\{ e_0(C)  \big[{\Theta}_{MC} C \big]_{j} + e_1 (C) \big[{\Theta}_{MC} C \big]_{j} +  e_1 (C) \theta_{MA, j} \Big\}.
        \end{aligned}
    \]
    \end{small}
    Note that $\mu (C, a') = \gamma_{YA}^{\dag} a' + \eta_{YC}^{\dag \top} C$.
    Hence, we conclude
    % \[
    %     \widehat{TM}_j^{\text{DR}} (\mathcal{G}_M) = \pn \widehat{DR}_{TM_j}^{\text{fast}} (X; \mathcal{G}_M) + \widehat{TM}_j^{\text{OLS}},
    % \]
    % and 
    \[
        \widehat{TM}_j^{\text{QR}} (\mathcal{G}_M) = \pn \widehat{QR}_{TM_j}^{\text{fast}} (X; \mathcal{G}_M) + \widehat{TM}_j^{\text{OLS}},
    \]
    with the estimated adjustment term $\widehat{QR}_{TM_j}^{\text{fast}} (X; \mathcal{G}_M)$ defined as
    % \begin{equation}\label{IM_DR_fast}
    % \begin{aligned}
    %     & \widehat{DR}_{TM_j}^{\text{fast}} (X; \mathcal{G}_M) = \widehat{DR}_{TM_j}^{\text{fast}} (C, A, M, Y; \mathcal{G}_M)\\
    %     = & \Bigg[ \frac{\mathds{1}(A = 1)}{\widehat{e}_1(C)} \Big\{ Y - \widehat\gamma_{YA}^{\dag} - \widehat\eta_{YC}^{\dag \top} C\Big\} - \frac{\mathds{1}(A = 0)}{\widehat{e}_0(C)} \Big\{ Y - \widehat\eta_{YC}^{\dag \top} C \big] \Big\} \Bigg] \\
    %     & - \Bigg[ \frac{\mathds{1} (A = 1)}{\widehat{e}_1(C)} \phi \Big( \Pa_j (\mathcal{G}_M)  \, ; \, \big[{\Theta}_{MC} C \big]_{k : k \in \pa_j(\mathcal{G}_M)} + \theta_{MA,  k \in \pa_j(\mathcal{G}_M)}, \big[ \var(e_M) \big]_{kk : k \in \pa_j(\mathcal{G}_M)} \Big) \\
    %     & ~~~~~~~~~~ \times \Big\{ Y - \widehat\eta_{Y M_j} M_j - \widehat\gamma_{Y \Pa_j(\mathcal{G}_M) }^{\top} \Pa_j(\mathcal{G}_M) - \widehat\eta_{YA} A - \widehat\gamma_{YC}^{\top} C \Big\} \\
    %     & ~~~~~~~~~~  - \frac{\mathds{1} (A = 0)}{\widehat{e}_0(C)}\phi \Big( \Pa_j (\mathcal{G}_M)  \, ; \, \big[{\Theta}_{MC} C \big]_{k : k \in \pa_j(\mathcal{G}_M)}, \big[ \var(e_M) \big]_{kk : k \in \pa_j(\mathcal{G}_M)} \Big) \\
    %     & ~~~~~~~~~~ \times \Big\{ Y  - \widehat\eta_{Y M_j} M_j - \widehat\gamma_{Y \Pa_j(\mathcal{G}_M) }^{\top} \Pa_j(\mathcal{G}_M) - \widehat\gamma_{YC}^{\top} C \Big\},
    % \end{aligned}
    % \end{equation}
    % and
    \begin{small}
    \begin{equation}\label{IM_QR_fast}
    \begin{aligned}
        & \widehat{QR}_{TM_j}^{\text{fast}} (X; \mathcal{G}_M) 
        = \Bigg[ \frac{\mathds{1}(A = 1)}{\widehat{e}_1(C)} \Big\{ Y - \widehat\gamma_{YA}^{\dag} - \widehat\eta_{YC}^{\dag \top} C\Big\} - \frac{\mathds{1}(A = 0)}{\widehat{e}_0(C)} \Big\{ Y - \widehat\eta_{YC}^{\dag \top} C \big] \Big\} \Bigg] \\
        & - \Bigg[ \frac{\mathds{1} (A = 1)}{\widehat{e}_1(C)} \phi \Big( \Pa_j (\mathcal{G}_M)  \, ; \, \big[\widehat{\Theta}_{MC} C \big]_{k : k \in \pa_j(\mathcal{G}_M)} + \widehat\theta_{MA,  k \in \pa_j(\mathcal{G}_M)}, \big[ \widehat{\var}(e_M) \big]_{kk : k \in \pa_j(\mathcal{G}_M)} \Big) \\
        & ~~~~~~~~~~ \times \Big\{ Y - \widehat\eta_{Y M_j} M_j - \widehat\gamma_{Y \Pa_j(\mathcal{G}_M) }^{\top} \Pa_j(\mathcal{G}_M) - \widehat\eta_{YA} A - \widehat\gamma_{YC}^{\top} C \Big\} \\
        & ~~~~~~~~~~  - \frac{\mathds{1} (A = 0)}{\widehat{e}_0(C)}\phi \Big( \Pa_j (\mathcal{G}_M)  \, ; \, \big[\widehat{\Theta}_{MC} C \big]_{k : k \in \pa_j(\mathcal{G}_M)}, \big[ \widehat{\var}(e_M) \big]_{kk : k \in \pa_j(\mathcal{G}_M)} \Big)\\
        & ~~~~~~~~~~ \times \Big\{ Y  - \widehat\eta_{Y M_j} M_j - \widehat\gamma_{Y \Pa_j(\mathcal{G}_M) }^{\top} \Pa_j(\mathcal{G}_M) - \widehat\gamma_{YC}^{\top} C \Big\} \\
        & + \frac{\mathds{1} (A = 1)}{\widehat{e}_1(C)} \Big\{  \widehat\gamma_{Y \Pa_j(\mathcal{G}_M)}^{\top} \Big[ \Pa_j(\mathcal{G}_M) - \big[\widehat{\Theta}_{MC} C \big]_{k: k \in \pa_j(\mathcal{G}_M)} - \widehat\theta_{MA, k : k \in \pa_j(\mathcal{G}_M)} \Big] \Big\} \\
        & - \frac{\mathds{1} (A = 0)}{\widehat{e}_0(C)} \Big\{  \widehat\gamma_{Y \Pa_j(\mathcal{G}_M)}^{\top} \Big[ \Pa_j(\mathcal{G}_M) - \big[\widehat{\Theta}_{MC} C \big]_{k: k \in \pa_j(\mathcal{G}_M)} \Big] \Big\} \Bigg]
    \end{aligned}
    \end{equation}
    \end{small}
    where $\phi(\bcdot \mid X_{S_1} \, ; \, \mu_{S_2 \mid S_1}, \Sigma_{S_2 \mid S_1})$ is the density for the conditional distribution $X_{S_2}  \mid X_{S_1}$ where $(X_{S_1}^{\top}, X_{S_2}^{\top})^{\top} \, \sim \, \mathcal{N} (\mu, \Sigma)$.
    Finally, by $\widehat{IM}_j^{\text{QR}} (\mathcal{G}_M) = \widehat{TM}_j^{\text{QR}} (\mathcal{G}_M) - \widehat{DM}_j^{\text{QR}}$, we conclude the proposition when both linear structures hold in Assumption \ref{ass_slsem}. Finally, it is noteworthy, as seen in Theorem \ref{thm_QR_nor}, that even if only one linear assumption holds in Assumption \ref{ass_slsem}, our estimator remains consistent.
\end{proof}

\section{Technique Proofs}\label{proof_A}

In the proof, to avoid any misunderstanding, we abbreviate the conditional density \(f_{Z_2 \mid Z_1} (z_2 \mid z_1)\) as \(f(z_2 \mid z_1)\), and omit the interval of the integral, assuming the interval is the support of the integrand variable.

\subsection{The Proofs of Section \ref{sec_semi}} 
~\\

\noindent \textbf{Proof of Theorem \ref{thm_direct_exp}:}

\begin{proof}
    Fixing a DAG $\mathcal{G}_M$, we drop the DAG index in these quantities for simplicity. Note that the joint intervention density can be written as 
    \[
        \begin{aligned}
            f \big( y \mid do(A = a, \, M_j = m_j)\big) & = \int \frac{f(y, c, a, m_j, \pa_j)}{f(a \mid c) f(m_j \mid c, a, \pa_j)} \, \mathrm{d} c \, \mathrm{d} \pa_j \\
            & = \int \frac{f(y, c, a, m_j, \pa_j)}{f(m_j, c, a, \pa_j)} \frac{f(c, a, \pa_j)}{f(a, c)} f(c) \, \mathrm{d} c \, \mathrm{d} \pa_j \\
            &= \int f(y \mid c, a, m_j, \pa_j) f(\pa_j \mid c, a) f(c) \, \mathrm{d} c \, \mathrm{d} \pa_j,
        \end{aligned}
    \]
    where the first equation is applying $\Pa(A) = C$, and $\Pa(M_j) = (C, A, \Pa_j)$ with Theorem 6 in \cite{kuroki1999identifiability}. Thus, 
    \[
         \begin{aligned}
             \E \big[ Y \mid do(A = a, \, M_j = m_j), C \big] & = \int \E \big[ Y \mid A = a, M_j = m_j, {\Pa}_j = \pa_j, C \big] f(\pa_j \mid a, C) \, \mathrm{d} \pa_j \\
             & = \varrho_j (a, m_j, C),
         \end{aligned}
    \]
    then
    \[
        \begin{aligned}
            \E \big[ Y \mid do(A = 1, M_j), C = c \big] - \E \big[ Y \mid do(A = 0, M_j), C = c \big] = \varrho_j (1, M_j, c) - \varrho_j (0, M_j, c).
        \end{aligned}
    \]
    Denote
    \[
         \begin{aligned}
             TM_j(c) := & \Big\{ \E \big[ Y \mid A = 1, C = c\big] - \E \big[ Y \mid A = 0, C = c\big] \Big\} \\
             & - \Big\{ \E \big[ Y \mid do(A = 1, M_j), C = c \big] - \E \big[ Y \mid do(A = 0, M_j), C = c \big] \Big\}.
         \end{aligned}
    \]
    Since $TM_j = \E TM_j(C)$, we have 
    \[
        \begin{aligned}
            TM_j & = \E \big[ TM_j(C) \big] \\
            & = \Big\{ \E \big[ Y \mid A = 1, C \big] - \E \big[ Y \mid A = 0, C \big] \Big\} - \Big\{ \varrho_j (1, M_j, C) - \varrho_j (0, M_j, C) \Big\}.
        \end{aligned}
    \]
    The formula for $DM_j$ is just by the definition.
\end{proof}

\noindent \textbf{Proof of Theorem \ref{thm_EIF}:}

\begin{proof}
    Let $F_{X; t}$ denote a one dimensional regular parametric submodel of $\mathscr{M}_{\text{nonpar}}$, with $F_{X; 0}=F_X$, and let $\E_t$ be the expectation with respect to $F_{X; t}$. 
    Denote $U$ the score of $F_{X ; t}$ at $t=0$ and $\nabla_{t=0}$ denoting differentiation with respect to $t$ at $t=0$. For a fixed DAG $\mathcal{G}_M$, we drop $\mathcal{G}_M$ in $\varrho_j (a', M_j, C; \mathcal{G}_M)$ as $\varrho_j (a', M_j, C)$. Note that we can write 
    \begin{small}
    \[
        \begin{aligned}
            & \E \varrho_j (a', M_j, C ) \\
            = & \int f(m_j, c) \mathrm{d} (m_j, c)\int \mu(C, a', \pa_j, M_j) \, \pi_{C, a'}(\pa_j)  \, \mathrm{d} \pa_j \\
            = & \int \E \big[ Y \mid A = a', M_j = m_j, {\Pa}_j = \pa_j, C = c\big] f(\pa_j \mid a', c) f(m_j \mid c) f(c) \, \mathrm{d} \pa_j \, \mathrm{d} m_j \, \mathrm{d} c.
        \end{aligned}  
    \]
    Denote $\mathrm{d} (c, \pa_j, m_j) = \mathrm{d} \mu$, then we can get 
    \[
        \begin{aligned}
            & \left. \frac{\partial}{\partial t} \E_t \varrho_j (a', M_j, C) \right|_{t = 0} \\
            = & \int \nabla_{t = 0} \E_t \big[ Y \mid A = a', M_j = m_j, {\Pa}_j = \pa_j, C = c\big] f(\pa_j \mid a', c) f(m_j \mid c) f(c) \mathrm{d} \mu \, \mathrm{d} F_{X, t} \\
            & ~~~~ + \int \E \big[ Y \mid A = a', M_j = m_j, {\Pa}_j = \pa_j, C = c\big] \nabla_{t = 0} f_t(\pa_j \mid a', c) f(m_j \mid c) f(c) \, \mathrm{d} \mu \, \mathrm{d} F_{X, t} \\
            & ~~~~ + \int \E \big[ Y \mid A = a', M_j = m_j, {\Pa}_j = \pa_j, C = c\big]  f(\pa_j \mid a', c) \nabla_{t = 0} f_t(m_j \mid c) f(c) \, \mathrm{d} \mu \, \mathrm{d} F_{X, t} \\
            & ~~~~ + \int \E \big[ Y \mid A = a', M_j = m_j, {\Pa}_j = \pa_j, C = c\big]  f(\pa_j \mid a', c) f(m_j \mid c) \nabla_{t = 0} f_t(c) \, \mathrm{d} \mu \, \mathrm{d} F_{X, t} \\
            =: & \sum_{k = 1}^4 A_{j, k}.
        \end{aligned}
    \]
    \end{small}
    Here $A_{j, k}, k = 1, 2, 3, 4$ can be calculated straightforwardly by the method in \cite{tchetgen2012semiparametric}. First,
    \begin{small}
    \[
        \begin{aligned}
            & A_{j, 1} = \int \nabla_t \E_t \big[ Y \mid A = a', M_j = m_j, {\Pa}_j = \pa_j, C = c\big] f(\pa_j \mid a', c) f(m_j \mid c) f(c) \, \mathrm{d} \pa_j \, \mathrm{d} (m_j, c) \, \mathrm{d} F_{X, t} \\
            % = & \E \Bigg[ U \int \frac{\mathds{1}(a' = A, m_j = M_j, \pa_j = \Pa_j, c = C)}{f(c, a', \pa_j, m_j)} \Big\{ Y -  \E \big[ Y \mid A = a', M_j = m_j, {\Pa}_j = \pa_j, C = c\big] \Big\} \, \mathrm{d} \pa_j \, \mathrm{d} (m_j, c) \Bigg] \\
            = & \E \Bigg[ U \int \frac{\mathds{1}(a' = A, m_j = M_j, \pa_j = \Pa_j, c = C)}{f(c) f (a' \mid c) f(\pa_j \mid a', c) f(m_j \mid \pa_j, a', c)} \\
            & ~~~~~ \times \Big\{ Y -  \E \big[ Y \mid A = a', M_j = m_j, {\Pa}_j = \pa_j, C = c\big] \Big\}  f(\pa_j \mid a', c) f(m_j, c) \, \mathrm{d} \pa_j \, \mathrm{d} (m_j, c) \Bigg] \\
            = & \E \Bigg[ U \int \frac{\mathds{1}(a' = A, m_j = M_j, \pa_j = \Pa_j, c = C) f(m_j \mid c)}{f(a' \mid c)f(m_j \mid \pa_j, a', c)} \\
            & ~~~~~~~~~~~~~~~~~~ \times \Big\{ Y -  \E \big[ Y \mid A = a', M_j = m_j, {\Pa}_j = \pa_j, C = c\big] \Big\}  \, \mathrm{d} \pa_j \, \mathrm{d} (m_j, c) \Bigg] \\
            = & \E \Bigg[ U \frac{\mathds{1}(A = a')}{e_{a'} (C)} \pi_{C, a'} ( \Pa_j)  \Big\{ Y -  \E \big[ Y \mid A = a', M_j, {\Pa}_j, C\big] \Big\}  \Bigg].
        \end{aligned}
    \]
    \end{small}
    Similarly, one can easily obtain that
    \begin{small}
    \[
        \begin{aligned}
            & A_{j, 2} = \int  \E \big[ Y \mid A = a', M_j = m_j, {\Pa}_j = \pa_j, C = c\big] \nabla_t f_t (\pa_j \mid a', c) f(m_j, c) \, \mathrm{d} \pa_j \, \mathrm{d} (m_j, c) \, \mathrm{d} F_{X, t} \\
            % = & \E \Bigg[ U \int\E \big[ Y \mid A = a', M_j = m_j, {\Pa}_j = \pa_j, C = c\big] \\
            % & ~~~~~~~~~~ \times \frac{\mathds{1}(a' = A, c = C)}{f(a', c)} \Big\{ \mathds{1} (\pa_j = \Pa_j) - f(\pa_j \mid a', c) \Big\} f(m_j, c) \, \mathrm{d} \pa_j \, \mathrm{d} (m_j, c) \Bigg] \\
            = & \E \Bigg[ U \int\E \big[ Y \mid A = a', M_j = m_j, {\Pa}_j = \pa_j, C = c\big] \\
            & ~~~~~~~~~~ \times \frac{\mathds{1}(a' = A, c = C)}{f(a' \mid c) f(c)} \Big\{ \mathds{1} (\pa_j = \Pa_j) - f(\pa_j \mid a', c) \Big\} f(m_j \mid c) f(c) \, \mathrm{d} \pa_j \, \mathrm{d} m_j \, \mathrm{d} c \Bigg] \\
            % = & \E \Bigg[ U \frac{\mathds{1} (A = a')}{e_{a'}(C)} \int \E \big[ Y \mid C, A = a', \Pa_j, M_j = m_j\big] f(m_j \mid C) \, d m_j \\
            % & ~~~~~~~~ - U \frac{\mathds{1}(A = a')}{e_{a'}(C)} \int \E \big[ Y \mid C, A = a', M_j = m_j, {\Pa}_j = \pa_j \big]  f(\pa_j \mid a', C) f(m_j \mid C)  \, \mathrm{d} \pa_j \, \mathrm{d} m_j\Bigg] \\
            = & \E \Bigg[ U \frac{\mathds{1} (A = a')}{e_{a'}(C)} \bigg\{  \int \E \big[ Y \mid C, A = a', \Pa_j, M_j = m_j\big] f(m_j \mid C) \, d m_j - \E \big[ \varrho_j (a', M_j, C) \mid C \big] \bigg\} \Bigg],
        \end{aligned}
    \]
    \end{small}
    \begin{small}
    \[
        \begin{aligned}
            A_{j, 3} & =  \int  \E \big[ Y \mid A = a', M_j = m_j, {\Pa}_j = \pa_j, C = c\big] f (\pa_j \mid a', c) \nabla_t f_t (m_j \mid c) f(c) \, \mathrm{d} \pa_j \, \mathrm{d} (m_j, c) \, \mathrm{d} F_{X, t}\\
            & = \E \Bigg[ U \int  \E \big[ Y \mid A = a', M_j = m_j, {\Pa}_j = \pa_j, C = c\big] f (\pa_j \mid a', c) \\
            & ~~~~~~~~~~~~~~~~~~~~~~~~~~~ \times \frac{\mathds{1} (c = C)}{f(c)}\Big\{ \mathds{1} (m_j = M_j) - f(m_j \mid c) \Big\} \, f(c) \, \mathrm{d} m_j \, \mathrm{d} \pa_j \mathrm{d} c \Bigg] \\
            & = \E \Big[ U \big\{ \varrho_j(a', M_j, C) - \E \big[ \varrho_j(a', M_j, C) \mid C \big] \big\} \Big],
        \end{aligned}
    \]
    \end{small}
    and
    \begin{small}
    \[
        \begin{aligned}
            A_{j, 4} & = \int \E \big[ Y \mid A = a', M_j = m_j, {\Pa}_j = \pa_j, C = c\big] f (\pa_j \mid a', c)  f (m_j \mid c) \nabla_t f_t(c) \, \mathrm{d} \pa_j \, \mathrm{d} (m_j, c)\, \mathrm{d} F_{X, t} \\
            & = \E \Bigg[ U \int \E \big[ Y \mid A = a', M_j = m_j, {\Pa}_j = \pa_j, C = c\big] \\
            & ~~~~~~~~~~~~~~~~~~~~ \times f (\pa_j \mid a', c)  f (m_j \mid c) \big\{ \mathds{1} (c = C) - f(c)\big\} \, \mathrm{d} \pa_j \, \mathrm{d} (m_j, c) \Bigg] \\
            & = \E \Big[ U \big\{  \E \big[ \varrho_j(a', M_j, C) \mid C \big] - \E \varrho_j(a', M_j, C) \big\} \Big].
        \end{aligned}
    \]
    Similarly, we can decompose 
    \[
        \begin{aligned}
            & \left. \frac{\partial}{\partial t} \E_t \zeta_j (a', 0, C) \right|_{t = 0} \\
            = & \int \nabla_{t = 0} \E_t \big[ Y \mid A = 1, M = m, C = c \big] f (m_{-j} \mid A = 0, C = c )  f (m_{j} \mid A = 1, C = c ) f(c) \, \mathrm{d} \mu \, \mathrm{d} F_{X, t} \\
            & ~ + \int  \E_t \big[ Y \mid A = 1, M = m, C = c \big] \nabla_{t = 0} f (m_{-j} \mid A = 0, C = c )  f (m_{j} \mid A = 1, C = c ) f(c) \, \mathrm{d} \mu \, \mathrm{d} F_{X, t} \\
            & ~ + \int  \E_t \big[ Y \mid A = 1, M = m, C = c \big]  f (m_{-j} \mid A = 0, C = c ) \nabla_{t = 0} f (m_{j} \mid A = 1, C = c ) f(c) \, \mathrm{d} \mu \,\mathrm{d} F_{X, t} \\
            & ~ + \int  \E_t \big[ Y \mid A = 1, M = m, C = c \big]  f (m_{-j} \mid A = 0, C = c )f (m_{j} \mid A = 1, C = c )  \nabla_{t = 0}  f(c) \, \mathrm{d} \mu \,\mathrm{d} F_{X, t} \\
            =: & \sum_{k = 1}^4 B_{j, k},
        \end{aligned}
    \]
    \end{small}
    and similarly verify that
    \begin{small}
    \[
        B_{j, 1} = \E \Bigg[ U \frac{\mathds{1}(A = 1)}{f( A = 1 \mid C)} \frac{f (M_{j} \mid A = 0, C) f (M_{-j} \mid A = a', C ) }{f(M_j \mid A = 1, C, M_{-j})f(M_{-j} \mid A = 1, C)} \Big\{ Y - \E \big[ Y \mid A = 1, M, C\big] \Big\} \Bigg],
    \]
    \end{small}
    \begin{small}
    \[
        B_{j, 2} = \E \Bigg[ U \frac{\mathds{1}(A = 0)}{f(A = 0 \mid C)} \bigg\{ \int \E \big[ Y \mid A = 1, M_{-j}, C, M_j = m_j\big]  f (m_{j} \mid A = a', C = c ) \, \mathrm{d} m_j - \zeta_j (a', 0, C) \bigg\} \Bigg],
    \]
    \end{small}
    \begin{small}
    \[
        B_{j, 3} = \E \Bigg[ U \frac{\mathds{1}(A = a')}{f(A = a' \mid C)} \bigg\{ \int \E \big[ Y \mid A = 1, M_{j}, C, M_{-j} = m_{-j}\big]  f (m_{-j} \mid A = 0, C  ) \, \mathrm{d} m_{-j} - \zeta_j (a', 0, C) \bigg\}\Bigg],
    \]
    \end{small}
    and
    \[
        B_{j, 4} = \E \Big[ U \big(  \zeta_j (a', 0, C)  - \E  \zeta_j (a', 0, C) \big) \Big].
    \]
    Finally, the efficient score for $\E \kappa(a', C)$ equal to 
    \[
        \frac{\mathds{1}(A = a')}{e_{a'} (C)} \Big\{ Y - \kappa(a', C) \Big\} + \kappa(a', C) - \E \kappa(a', C)
    \]
    has been well studied in various literature, see Section 2.2 in \cite{tchetgen2012semiparametric} for example. Thus, we conclude the results in the theorem.
\end{proof}

\noindent \textbf{Proof of Corollary \ref{cor_eff_DM_IM}:}

\begin{proof}
    The results are directly derived due to the linear property of efficient scores.
\end{proof}

\subsection{The Proofs of Section \ref{sec_straight_strategy}} 
~\\

Let $\E_{\xi} [\bcdot]$ represent the expected value computed with respect to the random variable $\xi$, while treating other variables as constants.
Employing graphical techniques, we can demonstrate the lemma as follows. 

\begin{lemma}\label{lem_nat_eff}
    Suppose the model satisfies Assumption \ref{ass_slsem}, then natural effects defined in Definition \ref{def_na} satisfy
    \[
        \begin{aligned}
            TE = \E_C \big[ \E \{ Y \mid do(A = 1), C\} - \E \{ Y \mid do(A = 0), C \} \big],
        \end{aligned}
    \]
    \[
         DE = \E_C \big[ \E\{Y \mid do(A=1, M=m^{(0)}), C \}- \E\{Y \mid do(A=0), C\} \big],
    \]
   and
   \[
       IE = \E_C \big[ \E\{Y \mid do(A=0, M=m^{(1)}), C\} - \E\{Y \mid do(A=0), C\} \big].
   \]

\end{lemma}

\noindent \textbf{Proof of Lemma \ref{lem_nat_eff}:}

\begin{proof}
     Introduce $A^{\dag}$ has the same law as $A$ except that $A^{\dag} \indep C$, then by the definition of $do$ operator,
    \begin{equation*}
        \begin{aligned}
            \E_C \big[ \E \{ Y \mid do(A = a), C\} \big] & = \E_C \E \{ Y \mid A^{\dag} = a, C \} \\
            & = \E \{ Y \mid A^{\dag} = a \} = \E \{ Y \mid do(A = a)\}.
        \end{aligned}
    \end{equation*}
    for any $a \in \{ 0, 1\}$ This furthermore implies
    \[
        \begin{aligned}
             TE & = \E \{ Y \mid do(A = 1)\} - \E \{ Y \mid do(A = 0)\} \\
             & =  \E_C \big[ \E \{ Y \mid do(A = 1), C\} \big] - \E_C \big[ \E \{ Y \mid do(A = 0 ), C\} \big].
        \end{aligned}
    \]
    This proves the result of $TE$. We can similarly prove the result for $DE$ and $IE$.
\end{proof}

Lemma \ref{lem_nat_eff} shows the definition of natural effects defined in \cite{pearl2000causality} can actually be written as the average on confounders.
 We require some DAG lemmas in order to study the formula for indirect effects of mediators. To the best of our knowledge, these lemmas are also novel, and they can be helpful resources for pertinent research.

\begin{lemma}\label{lem_dag}
    Suppose $B$ is the weight matrix of a DAG on $\{ X_1, \ldots, X_q\}$. Define $\Delta_{-i}  := \big[ (I_q - B^{\top})^{-1} u \big]_{-i} - (I_{q - 1} - B^{\top}_{-i, -i})^{-1} u _{-i}$, then
    \begin{itemize}
        \item[(1)] (Path Representation)
        \begin{equation*}
            \Delta_{-i} = u_i \Bigg[ \frac{\partial}{\partial x_i} \E [X_j \mid do_B (X_i = x_i) ] \Bigg]_{j \neq i} +  \Bigg[ \sum_{k \neq i} u_k \frac{\partial}{\partial x_k} \E [X_j \mid do_{B_{-i, -i}}(X_k = x_k)]\Bigg]_{j \neq i},
        \end{equation*}
        where $do_{A}$ is the $do$ operator on the DAG weight matrix $A$. 
        
        \item[(2)] (Matrix Expression)
        \begin{equation*}
            \Delta_{-i} = u_i \big[ (I_q - B^{\top})^{-1} \big]_{-i, i} + \big[ (I_q - B^{\top})^{-1} \big]_{-i, i} \big[ (I_q - B^{\top})^{-1} \big]_{i, -i} u_{-i} .
        \end{equation*}
    \end{itemize}
\end{lemma}

\noindent Before proving this lemma, we need the following result for digraph (do not require acyclic). We denote $|w(x)|$ as the product of the weights of the edges for any path $x$.

\begin{lemma}[\cite{yu2019dag}]\label{lem_digra}
    If $A$ is an adjacency matrix of a weighted digraph on $m$ vertices, then
    \[
        \left[A^{k}\right]_{i j}=\sum_{x \in \pi_{i j}^{+}(k)}|w(x)| - \sum_{y \in \pi_{i j}^{-}(k)}|w(y)|,
    \]
    where $\pi_{i j}^{+}(k)$ is the set of path with positive direction from $i$ to $j$ and $\pi_{i j}^{-}(k)$ is the set of path with negative direction from $i$ to $j$ with length $k$.
\end{lemma}

\noindent \textbf{Proof of Lemma \ref{lem_dag}:}

\begin{proof}
    Denote $j$ as the major index. And for any nonempty set $S = \{ l_1, \ldots, l_{|S|}\}$, we denote
    \[
        \big[ u_j \big]_{S} = \big[ u_j \big]_{j \in S} = (u_{l_1}, \ldots, u_{l_{|S|}})^{\top}
    \]
    as a vector with dimension $|S|$. Specially, define $[u_j] := [u_j]_{j \in [q]}$. By lemma \ref{lem_digra}, for any $u \in \mathbb{R}^q$, we have 
    \[
        \begin{aligned}
            (I - B^{\top})^{-1} u & = u + \sum_{m = 1}^{\infty} \big[ B^{\top} \big]^m u  \\
            & = [u_j] + \sum_{m = 1}^{\infty}\bigg[ \sum_{k = 1}^q \big[ (B^m)^{\top}\big]_{jk} u_k\bigg] \\
            & = [u_j] + \sum_{m = 1}^{\infty}\bigg[ \sum_{k = 1}^q [B^m]_{kj} u_k\bigg] \\
            \alignedoverset{\text{By applying Lemma \ref{lem_digra}}}{=} [u_j] + \sum_{m = 1}^{\infty} \bigg[ \sum_{k = 1}^q u_k \sum_{\Pi \in \pi_{kj}^{(m)}}  \big| w_B (\Pi)\big|  \bigg] \\
            & = [u_j] + \sum_{k = 1}^q u_k \bigg[  \sum_{m = 1}^{\infty} \sum_{\Pi \in \pi_{kj}^{(m)}} \big| w_{B}  (\Pi) \big| \bigg] \\
            & = \Bigg[u_j + \sum_{k = 1}^q u_k \sum_{\exists \,  \Pi \in \pi_{kj}} \big| w_B (\Pi) \big| \Bigg],
        \end{aligned}
    \]
    where we drop the superscript in the path $\pi_{kj}^{+}$ by the fact that there is no negative path in DAG. Similarly, one can prove 
    \[
        (I - B_{-i, -i}^{\top})^{-1} u_{-i} = \Bigg[u_j + \sum_{k = 1}^p u_k \sum_{\exists \, \Pi \in \pi_{kj}  \text{ and } i \notin \Pi} \big| w_B (\Pi) \big| \Bigg]_{j \neq i}.
    \]
    Therefore, 
    \[
        \begin{aligned}
            \Delta_{-i} & = \big[ (I - B^{\top})^{-1} u \big]_{-i} - (I - B^{\top}_{-i, -i})^{-1} u _{-i} \\
            & = \Bigg[u_j + \sum_{k = 1}^q u_k \sum_{\exists \,  \Pi \in \pi_{kj}} \big| w_B (\Pi) \big| \Bigg]_{j \neq i} - \Bigg[u_j + \sum_{k = 1}^p u_k \sum_{\exists \, \Pi \in \pi_{kj}  \text{ and } i \notin \Pi} \big| w_B (\Pi) \big| \Bigg]_{j \neq i} \\
            & = \Bigg[\sum_{k = 1}^p u_k \sum_{\exists \, \Pi \in \pi_{kj} \text{ and } i \in \Pi} \big| w_B (\Pi) \big| \Bigg]_{j \neq i} \\
            & = \Bigg[u_i \sum_{\exists \, \Pi \in \pi_{ij}} \big| w_B (\Pi) \big| + \sum_{k \neq i}^q u_k \sum_{\exists \, \Pi \in \pi_{kj}  \text{ and } i \notin \Pi} \big| w_B (\Pi) \big| \Bigg]_{j \neq i} \\
            & = \Bigg[u_i \sum_{\exists \, \Pi \in \pi_{ij}} \big| w_B (\Pi) \big| \Bigg]_{j \neq i } + \Bigg[ \sum_{k \neq i}^q u_k \sum_{\exists \, \Pi \in \pi_{kj}  \text{ and } i \notin \Pi} \big| w_B (\Pi) \big| \Bigg]_{j \neq i}\\ 
            & =: \Delta_{-i, 1} + \Delta_{-i, 2}. 
        \end{aligned}
    \]
    By the definition of $do$ operator and Proposition 3.1 in the Supplementary of \cite{nandy2017estimating}, one have
    \[
        \begin{aligned}
            \Delta_{-i, 1} & = \Bigg[u_i \sum_{\exists \, \Pi \in \pi_{ij}} \big| w_B (\Pi) \big| \Bigg]_{j \neq i } \\
            & = u_i  \Bigg[\sum_{\exists \, \Pi \in \pi_{ij}} \big| w_B (\Pi) \big| \Bigg]_{j \neq i } \\
            & = u_i \Bigg[ \frac{\partial}{\partial x_i}\E [X_j \mid do_B(X_i = x_i)]\Bigg]_{j \neq i} .
        \end{aligned}
    \]
    Similarly,
    \[
        \begin{aligned}
            \Delta_{-i, 2} & = \Bigg[ \sum_{k \neq i}^q u_k \sum_{\exists \, \Pi \in \pi_{kj}  \text{ and } i \notin \Pi} \big| w_B (\Pi) \big| \Bigg]_{j \neq i} \\
            & = \Bigg[ \sum_{k \neq i}^q u_k \sum_{\exists \, \Pi \in \pi_{kj}(B_{-i, -i})} \big| w_{B_{-i, -i}} (\Pi) \big| \Bigg]_{j \neq i}  \\
            & = \Bigg[ \sum_{k \neq i} u_k   \frac{\partial}{\partial x_k}\E [X_j \mid do_{B_{-i, -i}}(X_k = x_k)] \Bigg]_{j \neq i},
        \end{aligned}
    \]
    This proves the path representation. For the matrix representation, we use Woodbury matrix identity
    \[
        \left(A_{-j, -j}^{-1}\right)^{-1}=\left[A^{-1}\right]_{-j, -j}-\left[A^{-1}\right]_{-j,  j}\left(\left[A^{-1}\right]_{jj}\right)^{-1}\left[A^{-1}\right]_{j -j}.
    \]
    Thus, take $A = (I_q - B^{\top})$, we have 
    \[
        \begin{aligned}
            \left(A_{-i, -i}^{-1}\right)^{-1} & = \left( I_{q - 1} - B_{-i, -i} \right)^{-1} \\
            & = \big[ ( I_q - B^{\top} )^{-1} \big]_{-i, -i} -  \big[ ( I_q - B^{\top} )^{-1} \big]_{-i, i} \Big(  \big[ ( I_q - B^{\top} )^{-1} \big]_{i, i}\Big)^{-1} \big[ ( I_q - B^{\top} )^{-1} \big]_{i, - i} \\
            & = \big[ ( I_q - B^{\top} )^{-1} \big]_{-i, -i} -  \big[ ( I_q - B^{\top} )^{-1} \big]_{-i, i} \big[ ( I_q - B^{\top} )^{-1} \big]_{i, - i}.
        \end{aligned}
    \]
    where we use the fact that $[A^{-1}]_{i, i} \equiv 1$ for $i \in [p]$ as $B^{\top}$ is acyclic. Then,
    \[
        \begin{aligned}
            & \left( I_{q - 1} - B_{-i, -i} \right)^{-1} u_{-i} \\
            = & \big[ ( I_q - B^{\top} )^{-1} \big]_{-i, -i} u_{-i} -  \big[ ( I_q - B^{\top} )^{-1} \big]_{-i, i} \big[ ( I_q - B^{\top} )^{-1} \big]_{i, - i} u_{-i}.
        \end{aligned}
    \]
    On the other hand, note that, for any $j \neq i$, $(A u)_j = \sum_{k = 1}^q  A_{j k} u_k$. Thus,
    \[
        \begin{aligned}
            [Au]_{-i} & = \big[ A u \big]_{j \neq i} = \left[ A_{j i} u_i + \sum_{k \neq i} A_{j k} u_k \right]_{j \neq i} \\
            & = u_i [ A ]_{-i, i} + A_{-i, -i} u_{-i}.
        \end{aligned}
    \]
    Note that 
    \[
        \begin{aligned}
            & \big[ (I_q - B^{\top})^{-1} u \big]_{-i} \\
            = & u_{i} \big[ ( I_q - B^{\top} )^{-1} \big]_{-i, i} + \big[ ( I_q - B^{\top} )^{-1} \big]_{-i, - i} u_{-i}.
        \end{aligned}
    \]
    Combine the two results, we obtain that 
    \[
        \begin{aligned}
            \Delta_{-i} & = \big[ (I_q - B^{\top})^{-1} u \big]_{-i} - (I_{q - 1} - B^{\top}_{-i, -i})^{-1} u _{-i} \\
            & = \left( u_{i} \big[ ( I_q - B^{\top} )^{-1} \big]_{-i, i} + \big[ ( I_q - B^{\top} )^{-1} \big]_{-i, - i} u_{-i} \right) \\
            & ~~~~~~~~~~~~~~~~~~~~- \left( \big[ ( I_q - B^{\top} )^{-1} \big]_{-i, -i} u_{-i} -  \big[ ( I_q - B^{\top} )^{-1} \big]_{-i, i} \big[ ( I_q - B^{\top} )^{-1} \big]_{i, - i} u_{-i} \right) \\
            & = u_{i} \big[ ( I_q - B^{\top} )^{-1} \big]_{-i, i}  + \big[ ( I_q - B^{\top} )^{-1} \big]_{-i, i} \big[ ( I_q - B^{\top} )^{-1} \big]_{i, - i} u_{-i}.
        \end{aligned}
    \]
    This implies the matrix expression.
\end{proof}

An important result from Lemma \ref{lem_dag} is that it guarantees the following two identities. See the following corollary.

\begin{corollary}\label{cor_equal_cor}
     Suppose the conditions are the same as Lemma \ref{lem_digra}, then we have
     \begin{equation}\label{ident_1}
        \Big[ \E [X_j \mid X_i \cup \Pa_B (X_i)]_1 \mathds{1}(j \notin \Pa_B (X_i))\Big]_{j \neq i} = \big[ (I - B^{\top})^{-1} \big]_{-i, i},
    \end{equation}
    and
    \begin{equation}\label{ident_2}
        \Big[ \E [X_i \mid X_j \cup \Pa_B (X_j)]_1 \mathds{1}(i \notin \Pa_B (X_j))\Big]_{j \neq i}  = \big[ (I - B^{\top})^{-1} \big]_{i, -i}.
    \end{equation}
\end{corollary}

\noindent \textbf{Proof of Corollary \ref{cor_equal_cor}:}

\begin{proof}
    Compare the two expressions in Lemma \ref{lem_digra}, take arbitrary $u \in \mathbb{R}^q$, we have
    \[
         \big[ (I - B^{\top})^{-1} \big]_{-i, i} = \Bigg[ \frac{\partial}{\partial x_i} \E [X_j \mid do_B (X_i = x_i) ] \Bigg]_{j \neq i}.
    \]
    On the other hand, the Pearl's back-door adjustment implies
    \[
        \frac{\partial}{\partial x_i} \E [X_j \mid do_B (X_i = x_i) ] = \E [X_j \mid X_i \cup \Pa_B (X_i)]_1 \mathds{1} (j \notin \Pa_B(X_i)).
    \]
    The two results give \eqref{ident_1}. Similarly, we can get \eqref{ident_2}.
\end{proof}

\noindent \textbf{Proof of Proposition \ref{thm_unique}:}

\begin{proof}
    Consider the M-estimator defined as 
    \begin{equation}\label{thm_unique_1}
        \big[ \Theta^* \quad \theta^*\big] \in \arg \min_{[\Theta \ \theta] \in \mathbb{R}^{p \times t}} \operatorname{tr} \E\big[ M - \Theta C - \theta A \big] \big[ M - \Theta C - \theta A \big]^{\top}
    \end{equation}
    Since the distribution of $X$ is non-degenerate, 
    \[
        \begin{aligned}
            & \operatorname{tr} \E\big[ M - \Theta C - \theta A \big] \big[ M - \Theta C - \theta A \big]^{\top} \\
            = & \E \operatorname{tr} \big[ M - \Theta C - \theta A \big] \big[ M - \Theta C - \theta A \big]^{\top} \\
            = & \E \big[ M - \Theta C - \theta A \big]^{\top} \big[ M - \Theta C - \theta A \big] \\
            = & \E \big[ M - \Theta_{MC} C - \theta_{MA} A + (\Theta_{MC} - \Theta) C + (\theta_{MA} - \theta) A \big]^{\top} \\
            & ~~~~~~~~~~~~~~~~~~~~~~~~~~~~~~~~~~~ \big[ M - \Theta_{MC} C - \theta_{MA} A + (\Theta_{MC} - \Theta) C + (\theta_{MA} - \theta) A \big] \\
            = & \E \big[ e_M + (\Theta_{MC} - \Theta) C + (\theta_{MA} - \theta) A \big]^{\top}  \big[ e_M + (\Theta_{MC} - \Theta) C + (\theta_{MA} - \theta) A \big] \\
            = & \E e_M^{\top} e_M + \underbrace{\E \big[ (\Theta_{MC} - \Theta) C + (\theta_{MA} - \theta) A \big]^{\top}  \big[(\Theta_{MC} - \Theta) C + (\theta_{MA} - \theta) A \big]}_{\geq 0}.
        \end{aligned}
    \]
    Therefore, $ \big[ \Theta_{MC} \quad \theta_{MA}\big] $ is a solution of \eqref{thm_unique_1}. It is enough to show that the solution $ \big[ \Theta^* \quad \theta^*\big] $ is unique over MEC. Indeed, use the fact in Assumption \ref{ass_slsem} that the components of $\epsilon$ are mutually independent, 
    \[
        \begin{aligned}
            \E \big[ M - \Theta C - \theta A \big] \big[ M - \Theta C - \theta A \big]^{\top} & = \E e_M e_M^{\top} \\
            & = (I - B_{MM}^{\top}) \cov (\epsilon) (I - B_{MM}) \\
            & = \Big( \frac{1}{\sigma_{\ell_1}^2}  [I - B_{MM}]_{\ell_1 :}^{\top}  [I - B_{MM}]_{:\ell_2} \Big)_{1 \leq \ell_1, \ell_2  \leq p}.
        \end{aligned}
    \]
    When $\ell_1  = \ell_2 = \ell$, $\frac{1}{\sigma_{\ell_1}^2}  [I - B_{MM}]_{\ell_1 :}^{\top}  [I - B_{MM}]_{:\ell_2} = \frac{1}{\sigma_{\ell}^2} \sum_{k = 1}^p \big( B_{MM}\big)_{k \ell}^2 $. Thus, we know that
    \[
        \operatorname{tr} \E \big[ M - \Theta C - \theta A \big] \big[ M - \Theta C - \theta A \big]^{\top} = \sum_{1 \leq  j, k \leq p} \frac{1}{\sigma_j^2} (B_{MM})_{kj}^2.
    \]
    %\iffalse
    %for $\ell_1 \neq \ell_2$
    %\[
        %\begin{aligned}
            %\frac{1}{\sigma_{\ell_1}^2}  [I - B_{MM}]_{\ell_1 :}^{\top}  [I - B_{MM}]_{:\ell_2} = \frac{1}{\sigma_{\ell_1}^2} \left[\sum_{k = 1}^p \big( B_{MM}\big)_{k, \ell_1} \big( B_{MM}\big)_{k, \ell_2} - (B_{MM})_{\ell_1, \ell_2} - (B_{MM})_{\ell_2, \ell_1}\right].
        %\end{aligned}
    %\]
    %Now take two $B_{MM}^{(1)}$ and $B_{MM}^{(2)}$ from the MEC. Consider $M_i$ and $M_j$ are connected with spurious direction, and $\big[ B_{MM}^{(1)} \big]_{ij} = \big[ B_{MM}^{(2)} \big]_{ji} = 0$. Then we know that $M_i$ and $M_j$ is not on any immoralities. And even for $\ell_2 \neq i, j$
    %\[
        %\begin{aligned}
            %& \frac{1}{\sigma_{i}^2}  [I - B_{MM}^{(1)}]_{i :}^{\top}  [I - B_{MM}^{(1)}]_{:\ell_2} - \frac{1}{\sigma_{i}^2}  [I - B_{MM}^{(2)}]_{i:}^{\top}  [I - B_{MM}^{(2)}]_{:\ell_2} \\
            %= & 
        %\end{aligned}
    %\]
    %\fi
    Since any DAG in the same MEC shares the same skeleton, we know that the above optimization problem \eqref{thm_unique_1} is unique over MEC. The unique solution comes from the theory of M-estimator, see \cite{van2014asymptotically}.
\end{proof}

\noindent \textbf{Proof of Proposition \ref{thm_par_exp}:}

\begin{proof}
    The basic idea is implying Lemma \ref{lem_nat_eff}. Denote the conditional version of the natural direct effect as
    \[
        DE(c) := \E\{Y \mid do(A=1, M=m^{(0)}), C = c\}- \E\{Y \mid do(A=0), C = c\}.
    \]
    Note that under Assumption \ref{ass_slsem}, we can write 
    \[
        A \leftarrow h(C, \epsilon_A),
    \]
    \[
        M \leftarrow \left(I_{t-1}-B_{M M}^{\top}\right)^{-1} B_{M C}^{\top} C+\left(I_{t-1}-B_{M M}^{\top}\right)^{-1} \beta_{M A} A + \left(I_{t-1}-B_{M M}^{\top}\right)^{-1} \epsilon_M.
    \]
    Thus, we have 
    \[
        \begin{aligned}
            \E \{Y \mid d o (A = 1, M=m^{(0)}), C=c \} & = \E \{ Y \mid A = 1, M = m^{(0)}, C = c\} \\
            & = \beta_{Y C}^{\top} c +\alpha_{Y A} + \beta_{Y M}^{\top} m^{(0)},
        \end{aligned}
    \]
    and
    \[
        \begin{aligned}
            \E\{Y \mid do(A = 0), C = c\} & = \E\{Y \mid A = 0, C = c\} = \beta_{Y C}^{\top} C + \beta_{Y M}^{\top} m^{(0)}.
        \end{aligned}
    \]
    These equations together with the first display in Lemma \ref{lem_nat_eff} indicate
    \[
        \begin{aligned}
            DE & = \E_C DE (C) \\
            & = \E_C \Big[ \left( \beta_{Y C}^{\top} C +\alpha_{Y A} + \beta_{Y M}^{\top} m^{(0)} \right) - \left(\beta_{Y C}^{\top} C  + \beta_{Y M}^{\top} m^{(0)}\right) \Big] \\
            & = \E_C \alpha_{YA} = \alpha_{YA}.
        \end{aligned}
    \]
    And similarly, 
    \[
        \begin{aligned}
            IE & = \E_C IE(C) \\
            & = \E_C \Big[ \mathrm{E} \{Y \mid A = 0, M=m^{(1)}, C \}-\mathrm{E}\{Y \mid A = 0, C\} \Big] \\
            & = \E_C \Big[ \beta_{Y C}^{\top} C +  \beta_{Y M}^{\top} m^{(1)} - \left(\beta_{Y C}^{\top} C + \beta_{Y M}^{\top} m^{(0)}\right)  \Big] \\
            & = \E_C \Big[ \beta_{Y M}^{\top}\left(m^{(1)}-m^{(0)}\right) \Big] \\
            \alignedoverset{\text{by } C \, \indep \, \epsilon_M }{=} \E_C \bigg[ \left(I_{t-1}-B_{M M}^{\top}\right)^{-1} B_{M C}^{\top} C + \left(I_{t-1}-B_{M M}^{\top}\right)^{-1} \beta_{M A} \cdot 1 \\
            & ~~~~~~~~~~~~~ - \left( \left(I_{t-1}-B_{M M}^{\top}\right)^{-1} B_{M C}^{\top} C +\left(I_{t-1}-B_{M M}^{\top}\right)^{-1} \beta_{M A} \cdot 0 \right) \bigg] \\
            & = \E_C \left[ \beta_{Y M}^{\top}\left(I-B_{M M}^{\top}\right)^{-1} \beta_{M A} \right] = \beta_{Y M}^{\top}\left(I-B_{M M}^{\top}\right)^{-1} \beta_{M A}.
        \end{aligned}
    \]
    This completes part (i).
    \vspace{1ex}
    
    For proving part (ii), we first note that Theorem F.4 in \cite{cai2020anoce} implies $DM_j$ defined in Definition \ref{def_med} is equivalent to Definition 3.2 in  \cite{cai2020anoce}, which derives the expression $DM_j$ directly from Part (i). On the other hand, recall we have $TM_j = TM_j(c) = IE - IE_{\mathcal{G}_{(-j)}}$. Hence, $TM_j$ also has the exactly same as Definition 3.3 in \cite{cai2020anoce} by Corollary F.1 in \cite{cai2020anoce}. Therefore, we denote the total effect's condition version as
    \[
        TE(c) = \E \{ Y \mid do(A = 1), C = c\} - \E \{Y \mid do(A = 0), C = c \}.
    \]
    Then, by Theorem E.2 in \cite{watson2023heterogeneous}, that is $TM_j(c)$ can be interpreted as the effect of treatment $A$ on the outcome $Y$ that is mediated by the mediator $M_j$, or inversely as the change in total treatment effect caused by $M_j$ being removed from the causal graph, i.e.
    \[
        TM_j(c) = TE(c) - TE_{\mathcal{G}_{(-j)}}(c) = IE(c) - IE_{\mathcal{G}_{(-j)}} (c) = IE - IE_{\mathcal{G}_{(-j)}}.
    \]
    Similar to Lemma \ref{lem_nat_eff}, one can show $TM_j  = \E_C TM_j(C)$. Thus,
    \begin{equation}\label{TM_sub_graph}
        \begin{aligned}
            TM_j & = \E_C TM_j(C) \\
            & \qquad - \operatorname{cov} (\E \{Y \mid do(M_j = m_j + 1), C \} - \E \{Y \mid do(M_j = m_j), C\} , \Delta_j(C)) \\
            & = \E_C TM_j(C) = IE - IE_{\mathcal{G}_{(-j)}}.
        \end{aligned}
    \end{equation}
    Therefore,
    \[
        \begin{aligned}
            IM_j & = TM_j - DM_j = IE - IE_{\mathcal{G}_{(-j)}} - DM_j \\
            & = \beta_{Y  M}^{\top} \big( I - B_{M  M}^{\top}\big)^{-1} \beta_{M  A} -  \beta_{Y  M_{-j}}^{\top} \big( I - B_{M_{-j}  M_{-j}}^{\top}\big)^{-1} \beta_{M_{-j}  A} - DM_j \\
            & = \beta_{YM, -j}^{\top} \theta_{MA, -j} - \beta^{\top}_{YM_-j} \theta_{M_{-j}A} = \beta_{YM, -j}^{\top} \big( \theta_{MA, -j} - \theta_{M_{-j}A} \big).
        \end{aligned}
    \]
    Finally, for proving (ii'), note that $Y$ is still linear about $C$, $A$, and $M$, given $C = c$, the proof argument in proof of Theorem 3.1 and Theorem 3.2 of \cite{nandy2017estimating} still holds. Thus, fixed $C = c$, Proposition 2.1 in \cite{chakrabortty2018inference} guarantees, that is 
    \[
        \begin{aligned}
            \eta_j (c) & = \E \big[ M_j \mid A \, \cup \, \{ C = c \} \big]_1 \times \E \big[ Y \mid M_j, \operatorname{Pa} (M_j), A, \{C = c\} \big]_1 \\
            \alignedoverset{\text{by \eqref{reg_exp}}}{=} \theta_{MA, j} \times \E \big[ Y \mid M_j, \operatorname{Pa} (M_j), A, \{C = c\} \big]_1
        \end{aligned}
    \]
    for any fixed $C = c$, which combined with the expression of $DM_j$ gives the result for the alternative expression for $IM_j$. Hence, we complete our proof.
\end{proof}

\subsection{The Proofs of Section \ref{sec_mr}} 
~\\

\noindent \textbf{Proof of Alternative Strategies:}\label{proof_alter_strategies}

\begin{proof}
Just by the definition, we can prove the equivalence for alternative strategy 1 as follows: 
\[
    \begin{aligned}
        \E \bigg[ \frac{\mathds{1}(A = a')}{f(A = a' \mid c)} Y\bigg] & = \int \frac{\mathds{1}(a = a')}{f(a = a' \mid c)} y f(c, a, y) \mathrm{d} (a, c) \\
        & =  \int \frac{1}{f(a = a' \mid c)} y f(y \mid a' , c) f(a = a' \mid c) f(c) \, \mathrm{d} (a, c) \\
        & = \int y f(y \mid a' , c) f(c) \, \mathrm{d} (a, c) \\
        & = \E \kappa(a', C),
    \end{aligned}
\]
\begin{small}
\[
    \begin{aligned}
        & \E \bigg[ \frac{\mathds{1}(A = 0)}{f(A = 0 \mid C)} \bigg\{ \int \E \big[ Y \mid A = 1, M_{-j}, C, M_j = m_j\big]  f (m_{j} \mid A = a', C = c ) \, \mathrm{d} m_j \bigg\} \bigg] \\
        = & \int \frac{\mathds{1}(a = 0)}{f(A = a \mid C)} f(c, a, m_{-j}) \, \mathrm{d} c \, \mathrm{d} a \, \mathrm{d} m_{-j} \\
        & ~~~~~~~~~~~~~~~~~ \times \bigg\{ \int \E \big[ Y \mid A = 1, M_{-j} = m_{-j}, C = c, M_j = m_j\big]  f (m_{j} \mid A = a', C = c ) \, \mathrm{d} m_j \bigg\} \\
        = & \int \frac{1}{f(A = 0 \mid C)} f(c) f(A = 0 \mid c) f(m_{-j} \mid C = c, A = 0) \, \mathrm{d} c \, \mathrm{d} m_{-j} \\
        & ~~~~~~~~~~~~~~~~~ \times  \bigg\{ \int \E \big[ Y \mid A = 1, M_{-j} = m_{-j}, C = c, M_j = m_j\big]  f (m_{j} \mid A = a', C = c ) \, \mathrm{d} m_j \bigg\} \\
        = &  \int \E \big[ Y \mid A = 1, M = m, C\big]  f (m_{j} \mid C = c, A = a' ) f(m_{-j} \mid C = c, A = 0) f(c)\, \mathrm{d} m \, \mathrm{d} c \\
        = & \E \big[ \zeta_j (a', 0, C) \big]
    \end{aligned}
\]
\end{small}
and 
\begin{small}
\[
    \begin{aligned}
         & \E \bigg[ \frac{\mathds{1}(A = a')}{f( A = a' \mid c )} \left\{ \int \E \big[ Y \mid C, A = a', \Pa_j(\mathcal{G}_M), M_j = m_j\big] f(m_j \mid C) \, \mathrm{d} m_j\right\} \bigg] \\
         = & \int  \frac{\mathds{1}(A = a')}{f( A = a' \mid c )} f(c, a, \pa_j) \, \mathrm{d} c \, \mathrm{d} a  \, \mathrm{d} \pa_j \\
         & ~~~~~~~~~~~~~~~~~ \times   \left\{ \int \E \big[ Y \mid C = c, A = a', \Pa_j(\mathcal{G}_M) = \pa_j , M_j = m_j\big] f(m_j \mid C = c) \, \mathrm{d} m_j\right\} \\
         = & \int \E \big[ Y \mid C = c, A = a', \Pa_j(\mathcal{G}_M) = \pa_j , M_j = m_j\big] f(\pa_j \mid A = a', c) f(m_j, c) \, \mathrm{d} \pa_j \, \mathrm{d} (m_j, c) \\
         = & \int \varrho_j (a', m_j, c \, ; \, \mathcal{G}_M)\, \mathrm{d} (m_j, c) = \E \big[ \varrho_j (a', M_j, C \, ; \, \mathcal{G}_M) \big].
    \end{aligned}
\]
\end{small}
Thus, all identities for alternative strategy 1 hold. Similarly, for Strategy 2 and Strategy 3, we have 
\begin{small}
\[
     \begin{aligned}
         & \E \bigg[ \frac{\mathds{1}(A = a')}{f(A = a' \mid C)} \bigg\{ \int \E \big[ Y \mid A = 1, M_{j}, C, M_{-j} = m_{-j}\big]  f (m_{-j} \mid A = 0, C  ) \, \mathrm{d} m_{-j} \bigg\} \bigg] \\
         = & \int \frac{\mathds{1}(A = a')}{f(A = a' \mid C)} f(c, a, m_j) \, \mathrm{d} c \, \mathrm{d} a \, \mathrm{d} m_j \\
         & ~~~~~~~~~~~~ \times \bigg\{  \int \E \big[ Y \mid A = 1, M_{j} = m_j, C = c, M_{-j} = m_{-j}\big]  f (m_{-j} \mid A = 0, C = c ) \, \mathrm{d} m_{-j} \bigg\} \\
         = & \int \E \big[ Y \mid A = 1, C = c, M = m \big]  f (m_{-j} \mid A = 0, c) f(m_j \mid A = a', c) f(c)\, \mathrm{d} m_{-j} \mathrm{d} c \, \mathrm{d} m_j \\
         = & \E \big[\zeta_j(a', 0, C) \big],
     \end{aligned}
\]
\end{small}
\begin{small}
\[
    \begin{aligned}
        & \E \bigg[ \frac{\mathds{1}(A = 1)}{f( A = 1 \mid C)} \frac{f (M_{-j} \mid A = a', C ) }{f(M_{-j} \mid A = 1, C, M_j)} Y \bigg] \\
        = & \int \frac{\mathds{1}(a = 1)}{f( a = 1 \mid c)} \frac{f (m_{j} \mid a = 1, c) f (m_{-j} \mid A = a', c ) }{f(m_j \mid a = 1, c)f(m_{-j} \mid a = 1, c, m_j)} y f (c, a, m, y) \, \mathrm{d} \mu \\
        = &  \int \frac{1}{f( a = 1 \mid c)} \frac{f (m_{j} \mid a = 1, c) f (m_{-j} \mid A = a', c ) }{f(m \mid a = 1, c)} \\
        & ~~~~~~~~~~~~ \times y f(y \mid a = 1, c, m) f(m \mid a = 1, c) f(a = 1 \mid c) f(c) \, \mathrm{d} \mu \\
        = & \int  y f(y \mid a = 1, c, m) f (m_{j} \mid a = 1, c) f (m_{-j} \mid A = a', c ) f (c) \,  \mathrm{d} \mu \\
        = & \int \E [Y \mid A = 1, M = m, C = c]  f (m_{j} \mid a = 1, c) f (m_{-j} \mid A = a', c ) f (c) \,  \mathrm{d} \mu = \E \big[ \zeta_j (a', 0, C)\big],
    \end{aligned}
\]
\end{small}
and
\begin{small}
\[
    \begin{aligned}
        & \E \bigg[ \frac{\mathds{1}(A = a')}{f(A = a' \mid c)} f\big(\Pa_j (\mathcal{G}_M) \mid C, A = a'\big)  Y \bigg] \\
        = & \int \frac{\mathds{1}(a = a')}{f(a = a' \mid c)} f(\pa_j \mid c, a = a' )  y f(c, a, m, y)  \,  \mathrm{d} \mu \\
        = & \int \frac{1}{f(a = a' \mid c)} f(\pa_j \mid c, a = a' )  y f(y \mid a = a', c, m_j, \pa_j) f(m_j \mid a = a', c) f(a = a' \mid c) f(c) \, \mathrm{d} \mu \\
        = & \int y f(y \mid c, a = a', m_j, \pa_j) f(\pa_j \mid c, a = a' ) f(m_j, c) \, \mathrm{d} \pa_j \, \mathrm{d} (m_j, c) \\
        = & \E \big[ \varrho_j (a', M_j, C \, ; \, \mathcal{G}_M)\big],
    \end{aligned}
\]
\end{small}
which complete the calculations.
\end{proof}

\subsection{Proof of Section \ref{sec_asy_ols}}

Denote $a = \plim a_n$ if $a_n \, \xlongrightarrow{\pr} \, a$.

\begin{lemma}\label{lem_simple_reg}
    Let $X$ are mean-zero random vector and $Y$ are univariate random variable. Suppose $\E (\epsilon \mid X) = 0$ and $\var (\epsilon \mid X) = \sigma^2$ is free of $X$ in the regression 
    \[
        Y = \beta^{\top} X + \epsilon = \beta_1^{\top} X_1 + \beta_2^{\top} X_2 + \epsilon.
    \]
    Suppose $X \in \mathbb{R}^p$, and $\widehat{\beta}$ is the OLS estimator with sample size $p < n$,
    \[
        \sqrt{n} \big( \widehat{\beta}_1 - \beta_1 \big) = \left[ \plim  \widehat{\Gamma}_{X_1, X_2} \widehat{\Gamma}_{X_1, X_2}^{\top} \right]^{1 / 2}  \frac{1}{\sqrt{n}}  \sum_{i = 1}^n \varepsilon_i + o_p(1) \, \rightsquigarrow \, \mathcal{N} \left( 0, \plim \sum_{i = 1}^n \epsilon_i^2 \left[ \widehat{\Gamma}_{X_1, X_2} \widehat{\Gamma}_{X_1, X_2}^{\top} \right] \right),
    \]
    where $\{ \widehat{\epsilon}_i \}_{i = 1}^n$ are the residuals and $\{ \varepsilon_i\}_{i = 1}^n$ are the i.i.d. $\mathcal{N}(0, \sigma^2 I)$ error vector independent with $X$.
\end{lemma}
\begin{proof}
    We first note, we have \begin{equation}\label{reg_part_rep}
    \widehat{\beta}_1 - \beta_1 = \big[ \X_1^{\top} (I_n - P_{\X_2}) \X_1 \big]^{-1} \X_1^{\top} (I_n - P_{\X_2}) \boldsymbol{\varepsilon} = \widehat{\Gamma}_{X_1, X_2} \boldsymbol{\varepsilon}.
    \end{equation}
    \[
        \left[ n \widehat{\Gamma}_{X_1, X_{2}}  \widehat{\Gamma}_{X_1, X_{2}}^{\top} \right]^{-1} = \frac{1}{n} \mathbf{X}_1^{\top} \mathbf{X}_1 - \left[ \frac{1}{n} \mathbf{X}_1^{\top} \mathbf{X}_{2} \right] \left[ \frac{1}{n} \mathbf{X}_{2}^{\top} \mathbf{X}_{2} \right]^{-1} \left[ \frac{1}{n} \mathbf{X}_{2}^{\top} \mathbf{X}_1\right],
    \]
    we have 
    \[
        \begin{aligned}
             & \left[ n \widehat{\Gamma}_{X_1, X_{2}}  \widehat{\Gamma}_{X_1, X_{2}}^{\top} \right]^{-1} - \left[ \widehat{\var} (X_{1}) - \widehat{\cov} (X_{1}, X_{2}) \widehat{\var}^{-1} (X_{2})\widehat{\cov} (X_{ 2 }, X_{ 1}) \right]\\
             = & \left( \frac{1}{n} \mathbf{X}_1^{\top} \mathbf{X}_1 - \widehat{\var} (X_{1}) \right) \\
             & - \left( \left[ \frac{1}{n} \mathbf{X}_1^{\top} \mathbf{X}_{2} \right] \left[ \frac{1}{n} \mathbf{X}_{2}^{\top} \mathbf{X}_{2} \right]^{-1} \left[ \frac{1}{n} \mathbf{X}_{2}^{\top} \mathbf{X}_1\right] - \widehat{\cov} (X_{1}, X_{2}) \widehat{\var}^{-1} (X_{2})\widehat{\cov} (X_{ 2}, X_{ 1}) \right) \\
             & = 0,
        \end{aligned}
    \]
    which gives 
    \[
         n \widehat{\Gamma}_{X_1, X_{2}}  \widehat{\Gamma}_{X_1, X_{2}}^{\top}  = \left[ \widehat{\var} (X_{1}) - \widehat{\cov} (X_{1}, X_{2}) \widehat{\var}^{-1} (X_{2})\widehat{\cov} (X_{ 2 }, X_{ 1}) \right]^{-1}. 
    \]
    On the other hand, from \eqref{reg_part_rep}, we get that 
    \[
        \begin{aligned}
            \widehat{\beta}_{1} - \beta_{1} &= \widehat{\Gamma}_{X_{1}, X_{2}} \boldsymbol{\varepsilon}.
        \end{aligned}
    \]
    Since $\sqrt{n} \big( \widehat{\beta} - \beta \big)$ is asymptotic normal and $\E (\epsilon^2 \mid X) =\sigma^2$, its asymptotic variance is 
    \[
        \begin{aligned}
            & \left( \E  \left[ \begin{array}{cc}
                X_1 X_1^{\top}  & X_1 X_2^{\top} \\
                X_2 X_1^{\top} & X_2 X_2^{\top}
            \end{array}\right] \right)^{-1} \E \left( \epsilon_Y^2 \left[ \begin{array}{cc}
                X_1 X_1^{\top}  & X_1 X_2^{\top} \\
                X_2 X_1^{\top} & X_2 X_2^{\top}
            \end{array}\right] \right) \left( \E \left[ \begin{array}{cc}
                X_1 X_1^{\top}  & X_1 X_2^{\top} \\
                X_2 X_1^{\top} & X_2 X_2^{\top}
            \end{array}\right] \right)^{-1} \\
            = & \sigma^2 \left( \left[ \begin{array}{cc}
                \var(X_1)  & \cov(X_1, X_2) \\
                \cov(X_2, X_1) & \var (X_2)
            \end{array}\right] \right)^{-1} \\
            = & \sigma^2 \left[ \begin{array}{cc}
                \left( {\var} (X_{1}) - {\cov} (X_{1}, X_{2}) {\var}^{-1} (X_{2}) {\cov} (X_{ 2 }, X_{ 1}) \right)^{-1}  & * \\
                * & *
            \end{array}\right].
        \end{aligned}
    \]
    Combining the above result, we have 
    \[
        \begin{aligned}
            \sqrt{n} \big( \widehat\beta_1 - \beta_1 \big) \, & \rightsquigarrow \, \mathcal{N} \left( 0,  \sigma^2 \left[ {\var} (X_{1}) - {\cov} (X_{1}, X_{2}) {\var}^{-1} (X_{2}) {\cov} (X_{ 2 }, X_{ 1}) \right]^{-1} \right) \\
            & = \mathcal{N} \left( 0,  \sigma^2 \plim n \widehat{\Gamma}_{X_1, X_{2}}  \widehat{\Gamma}_{X_1, X_{2}}^{\top} \right) \\
            & = \mathcal{N} \left( 0, \plim \frac{\sum_{i = 1}^n \widehat\epsilon_i^2}{n} \plim n \widehat{\Gamma}_{X_1, X_{2}}  \widehat{\Gamma}_{X_1, X_{2}}^{\top}\right),
        \end{aligned}
    \]
    by the fact that $\frac{1}{n} \sum_{i = 1}^n \widehat\epsilon_i^2 \, \xlongrightarrow{\pr} \, \E \epsilon^2 = \sigma^2$. The asymptotic linear expression comes from the above display directly.
\end{proof}

\noindent \textbf{Proof of Theorem \ref{thm_CI_DE_IE_DM}}:
\begin{proof}
    From the regression \eqref{reg_exp}, we know that 
    \[
        \left( \begin{array}{c}
                \widehat{\beta}_{YC} \\
                \widehat{\alpha}_{YA} \\
                \widehat{\beta}_{YM}
            \end{array}\right) - \left( \begin{array}{c}
                {\beta}_{YC} \\
                {\alpha}_{YA} \\
                {\beta}_{YM}
            \end{array}\right) = \left( \left[ \begin{array}{c}
                \mathbf{C}^{\top} \\
                \mathbf{A}^{\top} \\
                \mathbf{M}^{\top}
            \end{array}\right] \left[ \begin{array}{ccc}
                \mathbf{C} & \mathbf{A} & \mathbf{M}
            \end{array}\right]  \right)^{-1} \left[ \begin{array}{c}
                \mathbf{C}^{\top} \\
                \mathbf{A}^{\top} \\
                \mathbf{M}^{\top}
            \end{array}\right] \boldsymbol\epsilon_Y.
    \]
    Thus, by \eqref{reg_part_rep}, we have 
    \[
        \widehat{\alpha}_{YA} - \alpha_{YA} = \widehat{\Gamma}_{A, (M, C)} \boldsymbol\epsilon_Y
    \]
    and
    \[
        \widehat{\beta}_{YM} - \beta_{YM} = \widehat{\Gamma}_{M, (C, A)} \boldsymbol\epsilon_Y.
    \]
    where \textbf{$\var (\epsilon_Y \mid C, A, M ) = \sigma_Y^2$ is a constant free of $C, A, M$}. Similarly, 
    \[
         \widehat{\theta}_{MA}^{\top} - \theta_{MA}^{\top} = \widehat{\Gamma}_{A, C} \mathbf{e}_{M}
    \]
    it is worthy to note that 
    \[
        \begin{aligned}
            \var (e_M \mid C, A ) & = \var \big[ (I - B_{MM}^{\top}) \epsilon_M  \mid (I - B_{CC}^{\top}) \epsilon_C, \beta_{AC}^{\top}(I - B_{CC}^{\top}) \epsilon_C + \epsilon_A \big] \\
            & = \var ((I - B_{MM}^{\top}) \epsilon_M ) = \sigma_{\epsilon_M}^2 (I - B_{MM}^{\top} ) (I - B_{MM})
        \end{aligned}
    \]
    is also a constant free of $C$ and $A$ by $\epsilon_M \indep (C, \epsilon_A)$. 
    Therefore, by Lemma \ref{lem_simple_reg}, we have
    \[
        \sqrt{n} \big( \widehat{\alpha}_{YA} - \alpha_{YA} \big) \, \rightsquigarrow \, \mathcal{N} \left( 0, \plim \sum_{i = 1}^n \widehat\epsilon_{Y, i}^2 \left[ \widehat{\Gamma}_{A, (M, C)}\widehat{\Gamma}_{A, (M, C)}^{\top} \right] \right),
    \]
    \[
        \sqrt{n} \big( \widehat{\beta}_{YM} - \beta_{YM} \big) \, \rightsquigarrow \, \mathcal{N} \left( 0, \plim \sum_{i = 1}^n \widehat\epsilon_{Y, i}^2 \left[ \widehat{\Gamma}_{M, (C, A)}\widehat{\Gamma}_{M, (C, A)}^{\top} \right] \right),
    \]
    and
    \[
        \begin{aligned}
            \sqrt{n} & \left( \begin{array}{c}
                \widehat{\alpha}_{YA} - \alpha_{YA}  \\
                 \widehat{\beta}_{YM} - \beta_{YM}
            \end{array}\right)= \left( \begin{array}{c}
                 \widehat{\Gamma}_{A, (M, C)} \\
                 \widehat{\Gamma}_{M, (C, A)}
            \end{array}\right) \boldsymbol\epsilon_Y \\ & \rightsquigarrow \, N \left( 0, \left[ \begin{array}{cc}
                * & \plim \sum_{i = 1}^n \widehat\epsilon_{Y, i}^2 \left[ \widehat{\Gamma}_{A, (M, C)}\widehat{\Gamma}_{M, (C, A)}^{\top} \right] \\
                \plim \sum_{i = 1}^n \widehat\epsilon_{Y, i}^2 \left[ \widehat{\Gamma}_{M, (C, A)} \widehat{\Gamma}_{A, (M, C)}^{\top}\right] & *
            \end{array}\right] \right).
        \end{aligned}
    \]
    Besides,
    \[
        \sqrt{n} \big( \widehat\theta_{MA} - \theta_{MA}\big) \, \rightsquigarrow \, \mathcal{N} \left( 0, \plim \widehat{\mathbf{e}}_M^{\top} \widehat{\mathbf{e}}_M \Gamma_{A, C} \Gamma_{A, C}^{\top} \right).
    \]
    Hence, we have 
    \[
        \begin{aligned}
            & \sqrt{n} \left( \begin{array}{c}
                \widehat{DE}^{OLS} - DE \\
                \widehat{IE}^{OLS} - IE
            \end{array}\right) = \left( \begin{array}{c}
                1 \\
                \theta_{MA}^{\top} \quad \beta_{YM}^{\top} 
            \end{array}\right) \left( \begin{array}{c}
                 \widehat{\alpha}_{YA} - \alpha_{YA} \\
                 \widehat\beta_{YM} - \beta_{YM} \\
                 \widehat\theta_{MA} - \widehat\theta_{MA}
            \end{array}\right) \\
            & \rightsquigarrow \, N \left( 0, \left[ \begin{array}{cc}
                \plim \sum_{i = 1}^n \widehat\epsilon_{Y, i}^2 \left[ \widehat{\Gamma}_{A, (M, C)}\widehat{\Gamma}_{A, (M, C)}^{\top} \right] & \plim \sum_{i = 1}^n \widehat\epsilon_{Y, i}^2 \left[ \widehat{\Gamma}_{A, (M, C)}\widehat{\Gamma}_{M, (C, A)}^{\top} \right] \theta_{MA} \\
                \plim \sum_{i = 1}^n \widehat\epsilon_{Y, i}^2 \theta_{MA}^{\top} \left[ \widehat{\Gamma}_{M, (C, A)} \widehat{\Gamma}_{A, (M, C)}\right] & {\theta}_{MA}^{\top} \bigg[ \plim \sum_{i = 1}^n \widehat\epsilon_{Y, i}^2 \left[ \widehat{\Gamma}_{M, (C, A)}\widehat{\Gamma}_{M, (C, A)}^{\top} \right] \bigg] \theta_{MA} \\
                 & + \beta_{YM}^{\top} \bigg[  \plim \widehat{\mathbf{e}}_M^{\top} \widehat{\mathbf{e}}_M \Gamma_{A, C} \Gamma_{A, C}^{\top} \bigg] \beta_{YM} 
            \end{array}\right]\right)
        \end{aligned}
    \]
    by $\epsilon_Y \indep e_M$, which gives the result for $\widehat{DE}^{OLS}$ as well as the result for $\widehat{IE}^{OLS}$ when $\theta_{MA}$ and $\beta_{YM}$ are not both equal to zero. When $\theta_{MA} = \beta_{YM} = 0$, we write 
    \[
        \begin{aligned}
            \frac{\sqrt{n}(\widehat{IE}^{\text{OLS}} - IE)}{\sqrt{\widehat\beta_{YM}^{\top} \widehat{\Sigma}_{\theta_{MA}} \widehat\beta_{YM} + \widehat\theta_{MA}^{\top} \widehat{\Sigma}_{\beta_{YM}} \widehat\theta_{MA}}} = \frac{\sqrt{n}\widehat\theta_{MA}^{\top} \sqrt{n} \beta_{YM}}{\sqrt{\sqrt{n}\widehat\beta_{YM}^{\top} \widehat{\Sigma}_{\theta_{MA}} \sqrt{n}\widehat\beta_{YM} + \sqrt{n}\widehat\theta_{MA}^{\top} \widehat{\Sigma}_{\beta_{YM}} \sqrt{n}\widehat\theta_{MA}}}
        \end{aligned}
    \]
    which leads to the result by the continuous mapping theorem. Similar method applies to the proof of $\widehat{DM}_j^{OLS}$. Therefore, we have 
    Then
    \[
        \frac{\sqrt{n}(\widehat{DE}^{\text{OLS}} - DE)}{\sqrt{\widehat{\Gamma}_{A, (M, C)} \widehat{\Gamma}_{A, (M, C)}^{\top} \sum_{i = 1}^n \widehat\epsilon_{Y, i}^2}} \, \rightsquigarrow \, \mathcal{N} (0, 1).
    \]
    and
    \[
        \begin{aligned}
            & \frac{\sqrt{n}(\widehat{IE}^{\text{OLS}} - IE)}{\sqrt{\widehat\beta_{YM}^{\top} \widehat{\Sigma}_{\theta_{MA}} \widehat\beta_{YM} + \widehat\theta_{MA}^{\top} \widehat{\Sigma}_{\beta_{YM}} \widehat\theta_{MA}}} \, 
            \rightsquigarrow \, \left\{ \begin{array}{ll}
                \displaystyle\frac{Z_{\theta}^{\top} Z_{\beta}}{\sqrt{Z_{\theta}^{\top} \Sigma_{\beta}Z_{\theta} + Z_{\beta}^{\top} \Sigma_{\theta}Z_{\beta}}} \qquad , & \text{if } \theta_{MA} = \beta_{YM} = 0, \\
                \mathcal{N}(0, 1), & \text{otherwise.}
            \end{array}\right.
        \end{aligned}
    \]
    where
    \[
        \left( \begin{array}{c}
         Z_\theta \\
         Z_\beta
    \end{array} \right) \, \sim \, \mathcal{N} \left( 0, \left( \begin{array}{cc}
        \Sigma_{\theta_{MA}} & \\
         & \Sigma_{\beta_{YM}}
    \end{array} \right) \right), \qquad \Sigma_{\theta_{MA}} = \plim \widehat\Sigma_{\theta_{MA}}, \qquad \Sigma_{\beta_{YM}} = \plim \widehat\Sigma_{\beta_{YM}}.
    \]
    Note that the $\frac{Z_{\theta}^{\top} Z_{\beta}}{\sqrt{Z_{\theta}^{\top} \Sigma_{\beta}Z_{\theta} + Z_{\beta}^{\top} \Sigma_{\theta}Z_{\beta}}}$ is much more concentrated around zero compared to the standard normal distribution \cite{chakrabortty2018inference}, we get the "$\geq$" instead of "$=$" for the confidence interval of $\widehat{IE}^{OLS}$. Similar result can be applied to $\widehat{DM}_j^{OLS}$.
\end{proof}

\noindent \textbf{Proof of Theorem \ref{thm_CI_IM}:}
\begin{proof}
    Let $\mu = \E X$ is the mean of $X$. Denote
    \[
        \xi_j(\mathcal{G}_{M}) :=  \E^{\text{reg}} \big[ Y \mid M_j \cup \Pa_j (\mathcal{G}_{M}) \cup A \cup C \big]_1
    \]
    and corresponding quantities
    \[
        \overline{\xi}_j (\mathcal{C}_M) := \frac{1}{\# \operatorname{MEC}(\mathcal{C}_M)} \sum_{\mathcal{G}_{M} \in \operatorname{MEC}({\mathcal{C}}_{ M})}  \E^{\text{reg}} \big[ Y \mid M_j \cup \Pa_j (\mathcal{G}_{M}) \cup A \cup C \big]_1
    \]
    and
    \[
        \widetilde{\xi}_j (\mathcal{C}_M) := \frac{1}{\# \operatorname{MEC}(\mathcal{C}_M)} \sum_{\mathcal{G}_{M} \in \operatorname{MEC}({\mathcal{C}}_{ M})}  \widehat{\xi} (\mathcal{G}_{M}), \qquad \widehat{\xi} (\mathcal{G}_{M}) := \widehat\E^{\text{reg}} \big[ Y \mid M_j \cup \Pa_j (\mathcal{G}_{M}) \cup A \cup C \big]_1
    \]
    From (ii') in Proposition \ref{thm_par_exp}, we have 
    \[
        \begin{aligned}
            & \widehat{{IM}}_j^{OLS} - \overline{IM}_j \\ 
            = & \widehat{\theta}_{MA, j} \Big( \widetilde{\xi}_j (\widehat{\mathcal{C}}_M )- \widehat{\beta}_{YM, j} \Big) - {\theta}_{MA, j}  \Big( \overline{\xi}_j ({\mathcal{C}}_{M} )- {\beta}_{YM, j} \Big) \\
            = & \widehat{\theta}_{MA, j} \Big( \big(\widetilde{\xi}_j (\widehat{\mathcal{C}}_M ) - \overline{\xi}_j ({\mathcal{C}}_{M} ) \big) - \big( \widehat{\beta}_{YM, j} - {\beta}_{YM, j} \big) \Big) + \big( \widehat{\theta}_{MA, j}  - {\theta}_{MA, j}\big)  \Big( \overline{\xi}_j ({\mathcal{C}}_{M} )- {\beta}_{YM, j} \Big) \\
            = & \widehat{\theta}_{MA, j} \Big(\widetilde{\xi}_j (\widehat{\mathcal{C}}_M ) - \widetilde{\xi}_j ({\mathcal{C}}_{M} )\Big) + \widehat{\theta}_{MA, j}  \Big( \widetilde{\xi}_j ({\mathcal{C}}_{M} ) - \overline{\xi}_j ({\mathcal{C}}_{M} ) \Big) \\
            & ~~~~~~~~~~~~~~~~~~~~~~~~ - \widehat{\theta}_{MA, j} \big( \widehat{\beta}_{YM, j} - {\beta}_{YM, j} \big) +  \big( \widehat{\theta}_{MA, j}  - {\theta}_{MA, j}  \big)  \Big( \overline{\xi}_j ({\mathcal{C}}_{M} )- {\beta}_{YM, j} \Big) \\
            =: & D_{j, 1} + D_{j, 2} - D_{j, 3} + D_{j, 4}
        \end{aligned}
    \]
    for $j \in [p]$. What we want to do next is finding the asymptotic expression of $\widehat{{IM}}_j - \overline{IM}_j $. First, for any non-negative sequence $\{ a_n \}$, 
    \[
        \begin{aligned}
            \pr \left( a_n \max_{j \in [p]}\big| D_{j, 1} \big | > \varepsilon \right) & \leq \pr \Big( \exists\, j \in [p] : \widetilde{\xi}_j (\widehat{\mathcal{C}}_M ) - \widetilde{\xi}_j ({\mathcal{C}}_{M} ) \neq 0 \Big) \\
            & \leq \pr \Big( \widehat{\mathcal{C}}_M \neq \mathcal{C}_{ M}\Big) \overset{\text{by Assumption \eqref{ass_nor_cpdag}}}{\longrightarrow} 0.
        \end{aligned}
    \]
    Thus, we have $\max_{j \in [p]} D_{j, 1} = o_p(a_n^{-1})$. We first consider the case at least one $\theta_{MA, j}$ or $\overline{\xi}_j ({\mathcal{C}}_{M} )- {\beta}_{YM, j}$ is not equal to zero. From the proof of Theorem \ref{thm_CI_DE_IE_DM} and Lemma \ref{lem_simple_reg}, we know that 
    \[
        D_{j, 2} = {\theta}_{MA, j} \Big( \widetilde{\xi}_j ({\mathcal{C}}_{M} ) - \overline{\xi}_j ({\mathcal{C}}_{M} ) \Big) + o_p \left( \frac{1}{n}\right),
    \]
    \[
        \begin{aligned}
            D_{j, 3} & = \theta_{MA, j} \Big[ \widehat{\Gamma}_{M, (C, A)} \boldsymbol\epsilon_Y \Big]_j+ o_p \left( \frac{1}{n}\right) \\
            & = \theta_{MA, j} \left[ \plim \widehat{\Gamma}_{M, (C, A)} \widehat{\Gamma}_{M, (C, A)}^{\top} \right]_{jj}^{1 / 2} \frac{1}{n} \sum_{i = 1}^n \epsilon_Y + o_p \left( \frac{1}{n}\right),
        \end{aligned}
    \]
    and
    \[
        \begin{aligned}
            D_{j, 4} & = \big( \widehat{\theta}_{MA, j}^{\top}  - {\theta}_{MA, j}^{\top}  \big)  \Big( \overline{\xi}_j ({\mathcal{C}}_{M} )- {\beta}_{YM, j} \Big) \\
            & = \left[ \widehat{\Gamma}_{A, C} \mathbf{e}_M \right]_{: j} \Big( \overline{\xi}_j ({\mathcal{C}}_{M} )- {\beta}_{YM, j} \Big) + o_p \left( \frac{1}{n}\right) \\
            & = \Big( \overline{\xi}_j ({\mathcal{C}}_{M} )- {\beta}_{YM, j} \Big) \left[ \widehat{\Gamma}_{A, C} \boldsymbol{\epsilon}_M (I_p - B_{MM})^{-1}\right]_{:j}  + o_p \left( \frac{1}{n}\right) \\
            & =  \Big( \overline{\xi}_j ({\mathcal{C}}_{M} )- {\beta}_{YM, j} \Big) \left[ \plim \widehat{\Gamma}_{A, C} \widehat{\Gamma}_{A, C}^{\top}\right]^{1 / 2}  \\
            & ~~~~~~~~~~~~~~~~~~~~~~~~~~~~~~\Big[ (I_p - B_{MM}^{\top})^{-1} \Sigma_M (I_p - B_{MM})^{-1} \Big]_{jj}^{1 / 2} \frac{1}{n} \sum_{i = 1}^n \varepsilon_{M, j, i} + o_p \left( \frac{1}{n}\right).
        \end{aligned}
    \]
    uniformly in $j \in [p]$, where $\varepsilon_M := \Sigma_M^{-1 / 2} \epsilon_M$. We denote $X_{S_{M, j, 1}}, \ldots, X_{S_{M, j, L_{\text{distinct}, j}}}$ be distinct parent sets of $M_j$ with $m_{j, 1}, \ldots, m_{j, L_{\text{distinct}, j}}$ times. Let 
    \[
        \begin{aligned}
            \mathcal{S}_j & := \big\{ S_{j r} = (j + t, 1, \ldots, t, X_{S_{M, j, r}}) : r \in \{1, \ldots, L_{\text{distinct}, j} \}, \, t \notin X_{S_{M, j, r}}\big\} \\
            & ~~~~~~~~ \cup \, \big\{ S_{j r} = (j + t, 1, \ldots, t - 1, X_{S_{M, j, r}}) : r \in \{1, \ldots, L_{\text{distinct}, j} \}, \, t \in X_{S_{M, j, r}}\big\} \boldsymbol\epsilon_Y
        \end{aligned}
    \]
    with $a_{S_{jr}} = (m_{ j r} / L_j, 0, \ldots, 0)^{\top} \in \mathbb{R}$ and $L_j = \sum_{r = 1}^{L_{\text{distinct}, j}} m_{j r}$.  Since 
    \[
        \sum_{S_{j r} \in \mathcal{S}_j} \| a_{S_{j r}} \|_2 = \sum_{r = 1}^{L_{\text{distinct}, j}} \frac{m_{ jr}}{L_j} = 1
    \]
    for any $j \in [p]$, we can apply REMARK 5.2 in \cite{chakrabortty2018inference}, we have
    \[
        \widetilde{\xi}_j ({\mathcal{C}}_{M} ) - \overline{\xi}_j ({\mathcal{C}}_{M} ) =  \frac{1}{n} \sum_{i = 1}^n W_{i}^{(j)}({\mathcal{C}}_{M} ) + o_p \left(\frac{1}{\sqrt{n}} \right) 
    \]
    where 
    \[
        \begin{aligned}
            & W^{(j)}({\mathcal{C}}_{M} ) \\
            = & \frac{1}{L_j} \sum_{\ell=1}^{L_j} {e}_{1,\left|S_{j \ell}\right|}^{\top} \left(\Sigma_{S_{j \ell} S_{j \ell}}\right)^{-1}  \left({X}_{S_{j \ell}}-{\mu}_{S_{j \ell}}\right) \Big( Y - \E Y  -\Sigma_{p S_{j \ell}} \left(\Sigma_{S_{j \ell} S_{j \ell}}\right)^{-1} \left({X}_{S_{j \ell}}-{\mu}_{S_{j \ell}}\right) \Big) \\
            = & \frac{1}{L_j} \sum_{\ell=1}^{L_j} \frac{m_r}{L_j \sigma_{t + j \mid \left(1, \ldots, t, S_{M, j,  \ell}\right)}^2}\left(\left(M_{j, i} - \E M_j \right)-{\beta}_{t + j \mid\left(1, \ldots, t, S_{M, j, \ell} \right)}^{\top}\left(X_{\left(1, \ldots, t, S_{M, \ell, j}\right), i}-{\E X}_{\left(1, \ldots, t, S_{M, \ell, j}\right)}\right)\right) \\
            & ~~~~~~ \times \Big( Y - \E Y  -\Sigma_{p S_{j \ell}} \left(\Sigma_{S_{j \ell} S_{j \ell}}\right)^{-1} \left({X}_{S_{j \ell}}-{\mu}_{S_{j \ell}}\right) \Big)
        \end{aligned}
    \]
    where $e_{i, k}$ the $i$-th column of the $k \times k$ identity matrix and ${\beta}_{t + j \mid\left(1, \ldots, t, S_{M, j, \ell}\right)}^{\top} := \left(\Sigma\right)_{j\left(1, \ldots, t, S_{M, j , \ell}\right)} \penalty 0 \left(\Sigma_{\left(1, \ldots, t, S_{M, j, \ell}\right)\left(1, \ldots, t,S_{M, j, \ell}\right)}\right)^{-1}$. Therefore, we can obtain
    \[
        \begin{aligned}
            & \sqrt{n} \big( \widehat{{IM}}_j^{OLS} - \overline{IM}_j \big)\\
            = & \frac{1}{\sqrt{n}} \sum_{i = 1}^n \theta_{MA, j} \Big\{ W_{i}^{(j)}({\mathcal{C}}_{M} ) + \left[ \plim \widehat{\Gamma}_{M, (C, A)} \widehat{\Gamma}_{M, (C, A)}^{\top} \right]_{jj}^{1 / 2} \epsilon_{Y, i} \Big\} \\
            & + \Big( \overline{\xi}_j ({\mathcal{C}}_{M} )- {\beta}_{YM, j} \Big) \left[ \plim \widehat{\Gamma}_{A, C} \widehat{\Gamma}_{A, C}^{\top}\right]^{1 / 2} \Big[ (I_p - B_{MM}^{\top})^{-1} \Sigma_M (I_p - B_{MM})^{-1} \Big]_{jj}^{1 / 2} \frac{1}{\sqrt{n}} \sum_{i = 1}^n \varepsilon_{M, j, i} \\
            & + o_p(1) \, \rightsquigarrow \, \mathcal{N} \big(0, \sigma_{\overline{IM}_j}^2 \big) 
        \end{aligned}
    \]
    where 
    \[
        \begin{aligned}
            & \sigma_{\overline{IM}_j}^2 := \E \bigg[ \theta_{MA, j} \Big\{ W^{(j)}({\mathcal{C}}_{M} ) + \left[ \plim \widehat{\Gamma}_{M, (C, A)} \widehat{\Gamma}_{M, (C, A)}^{\top} \right]_{jj}^{1 / 2} \epsilon_{Y} \Big\} \\
            & + \Big( \overline{\xi}_j ({\mathcal{C}}_{M} )- {\beta}_{YM, j} \Big) \left[ \plim \widehat{\Gamma}_{A, C} \widehat{\Gamma}_{A, C}^{\top}\right]^{1 / 2} \Big[ (I_p - B_{MM}^{\top})^{-1} \Sigma_M (I_p - B_{MM})^{-1} \Big]_{jj}^{1 / 2}\varepsilon_{M, j} \bigg]^2,
        \end{aligned}
    \]
    which implies $\sigma_{\overline{IM}_j}^2$ can be estimated by 
    \begin{equation}\label{est_sigma_IM}
        \begin{aligned}
            \widehat\sigma_{\overline{IM}_j}^2 & = \frac{1}{n} \sum_{i = 1}^n \bigg[ \widehat\theta_{MA, j} \Big\{ \widehat{W}_i^{(j)}({\widehat{\mathcal{C}}}_{M} ) + \left[ \widehat{\Gamma}_{M, (C, A)} \widehat{\Gamma}_{M, (C, A)}^{\top} \right]_{jj}^{1 / 2} \widehat\epsilon_{Y, i} \Big\} \\
            & + \Big( \widetilde{\xi}_j (\widehat{\mathcal{C}}_{M} )- \widehat{\beta}_{YM, j} \Big) \left[  \widehat{\Gamma}_{A, C} \widehat{\Gamma}_{A, C}^{\top}\right]^{1 / 2} \Big[ (I_p - \widehat{B}_{MM}^{\top})^{-1} \widehat\Sigma_M (I_p - \widehat{B}_{MM})^{-1} \Big]_{jj}^{1 / 2}\widehat\varepsilon_{M, j, i} \bigg]^2
        \end{aligned}
    \end{equation}
    where $\widehat{B}_{MM}$ can be estimated incidentally by a version of PC algorithm \citep{harris2013pc}. Finally, when $\theta_{MA, j} = \overline{\xi}_j ({\mathcal{C}}_{M} )- {\beta}_{YM, j} = 0$, we will have $ \lim_{n \rightarrow \infty} \pr \big( \sqrt{n} \big| \widehat{IM}_j^{OLS} - IM_j \big| \geq \widehat\sigma_{\overline{IM}_j} \Phi^{-1} (1 - \alpha / 2)  \big) > 1 - \alpha$ by using the similar argument in the proof of Theorem \ref{thm_CI_DE_IE_DM}.
\end{proof}

\noindent \textbf{Proof of Theorem \ref{thm_QR_nor}:}

\begin{proof}
    We will prove the consistency for $\widehat{DM}_j^{\text{QR}}$ and semiparametric efficiency for $\widehat{TM}_j^{\text{avg}, \, \text{QR}}$. The proof of other parts can be similarly derived. We first prove the consistency.\\
    
    \noindent \textbf{Consistency of} $\widehat{DM}_j^{\text{QR}}$:\\
    
    Note that $\widehat{DM}_j^{\text{QR}} - DM_j = \pn \big\langle \widehat{S}^{\text{eff, nonpar}} \big( \E \zeta_j (\bcdot, 0, C) \big)  \big\rangle$,
    % \[
    %     \begin{aligned}
    %         \widehat{DM}_j^{\text{QR}} - DM_j & = \pn \Big[ \widehat{S}^{\text{eff, nonpar}} \big( \E \zeta_j (1, 0, C) \big) - \widehat{S}^{\text{eff, nonpar}} \big( \E \zeta_j (0, 0, C) \big) \Big], 
    %     \end{aligned}
    % \]
    so it is sufficient to prove $\pn \widehat{S}^{\text{eff, nonpar}} \big( \E \zeta_j (a', 0, C) \big) \, \xlongrightarrow{\pr} \, 0$ for any $a' \in \{ 0, 1\}$. Denote $\pi_{x_S}^{(T)} (m_T) = \pi_{x_S} (m_T)$, Then we write 
    \[
        \begin{aligned}
            & \psi_j^{(0)} \big( \mu, \pi_{C, a'}^{(j)}, \pi_{C, 0}^{(-j)}\big) = \int_{\mathcal{M}} \mu (C, 1, m) \pi_{C, 1} (m_j) \pi_{C, 0}(m_{-j}) \, \mathrm{d} m - \E  \zeta_j (a', 0, C) \\
            & \psi_j^{(1)} \big( e_0, \mu, \pi_{C, a'}^{(j)}\big) = \frac{\mathds{1}(A = 0)}{{e}_0(C)} \bigg\{ \int_{\mathcal{M}_j} \mu (C, 1, m_j, M_{-j}) \pi_{C, a'}(m_j) \, \mathrm{d} m_j - \zeta_j (a', 0, C) \bigg\} \\
            & \psi_j^{(2)} \big( {e}_{a'}, \mu, \pi_{C, 0}^{(-j)} \big) = \frac{\mathds{1}(A = a')}{{e}_{a'} (C)} \bigg\{ \int_{\mathcal{M}_{-j}} \mu(C, 1, M_j, m_{-j}) \pi_{C, 0} (m_{-j}) \, \mathrm{d} m_{-j} - \zeta_j (a', 0, C) \bigg\}  \\
            & \psi_j^{(3)} \big( e_1,  \pi_{C, a'}^{(-j)}, \pi_{C, 1, M_{j}}^{(-j)}, \mu \big) = \frac{\mathds{1}(A = 1)}{e_1(C)} \frac{ \pi_{C, a'}(M_{-j}) }{\pi_{C, 1, M_{j}}(M_{-j})} \Big\{ Y - \mu (C, 1, M) \Big\} 
        \end{aligned}
    \]
    then $\pn  S^{\text{eff, nonpar}} \big( \E \zeta_j (a', 0, C) \big)  = \sum_{\ell = 0}^3 \pn \psi_j^{(\ell)}(\cdot)$. We break the proof into four parts, \textbf{Part $\boldsymbol\ell$} gives the consistency under different model $\mathscr{M}_{j, \, \ell}$.
    \vspace{1ex}

    \textbf{Part 0.} When model $\mathscr{M}_0$ is correctly specified, i.e., $\widehat\mu$, $\widehat\pi_{C, a'}^{(j)}$, and $\widehat\pi_{C, 0}^{(-j)}$ are consistency. We have 
    \[
        \pn \psi_j^{(0)} \big( \widehat\mu, \widehat\pi_{C, a'}^{(j)}, \widehat\pi_{C, 0}^{(-j)}\big) = \pn \Big[ \psi_j^{(0)} \big( \widehat\mu, \widehat\pi_{C, a'}^{(j)}, \widehat\pi_{C, 0}^{(-j)}\big) - \psi_j^{(0)} \big( \mu, \pi_{C, a'}^{(j)}, \pi_{C, 0}^{(-j)}\big) \Big] + \pn \psi_j^{(0)} \big( \mu, \pi_{C, a'}^{(j)}, \pi_{C, 0}^{(-j)}\big),
    \]
    with $\pn \psi_j^{(0)} \big( \mu, \pi_{C, a'}^{(j)}, \pi_{C, 0}^{(-j)}\big) \, \xlongrightarrow{\pr} \, P \psi_j^{(0)} \big( \mu, \pi_{C, a'}^{(j)}, \pi_{C, 0}^{(-j)}\big) = 0$. For another part, we will use empirical processes technique. Indeed, we have 
    \[
        \begin{aligned}
            & \bigg| \pn \Big[ \psi_j^{(0)} \big( \widehat\mu, \widehat\pi_{C, a'}^{(j)}, \widehat\pi_{C, 0}^{(-j)}\big) - \psi_j^{(0)} \big( \mu, \pi_{C, a'}^{(j)}, \pi_{C, 0}^{(-j)}\big) \Big] \bigg| \\ 
            \leq & \, \bigg| \pn \int_{\mathcal{M}} \big\{ \widehat\mu (C, 1, m) - \mu (C, 1, m) \big\} \pi_{C, a} (m_j) \pi_{C, 0}(m_{-j}) \, \mathrm{d} m \bigg|\\
            & \, + \bigg| \pn \int_{\mathcal{M}}  \mu (C, 1, m) \big\{ \widehat\pi_{C, a} (m_j) - \pi_{C, a} (m_j) \big\} \pi_{C, 0}(m_{-j}) \, \mathrm{d} m \bigg|\\
            & + \bigg| \pn \int_{\mathcal{M}}  \mu (C, 1, m)  \pi_{C, a} (m_j) \big\{ \widehat\pi_{C, 0}(m_{-j}) - \pi_{C, 0}(m_{-j}) \big\} \, \mathrm{d} m \bigg|.
        \end{aligned}
    \]
    For sufficient small constant $\varepsilon>0$, we define a set of functions $\mathcal{U}(\varepsilon)$ that contains conditional expectation $\mu_{\varepsilon}$ such that
    \[
        \E \int_{\mathcal{M}} \big\{ \mu_{\varepsilon} (C, 1, m) - \mu (C, 1, m) \big\}^2 \, \mathrm{d} m \leq \varepsilon^2.
    \]
    then we consider  
    \begin{small}
    \[
        \begin{aligned}
             & \E \sup_{\mu_{\varepsilon}  \, \in \, \mathcal{U}(\varepsilon)} \bigg| \pn \int_{\mathcal{M}} \big\{ \mu_{\varepsilon} (C, 1, m) - \mu (C, 1, m) \big\} \pi_{C, a} (m_j) \pi_{C, 0}(m_{-j}) \, \mathrm{d} m \bigg| \\
             \leq & \E \sup_{\mu_{\varepsilon} \, \in \, \mathcal{U}(\varepsilon)} \bigg[ \pn \int_{\mathcal{M}} \big| \mu_{\varepsilon}  (C, 1, m) - \mu (C, 1, m) \big|\pi_{C, a} (m_j) \pi_{C, 0}(m_{-j}) \, \mathrm{d} m \bigg] \\
             \leq & \pn \sup_{\mu_{\varepsilon} \, \in \, \mathcal{U}(\varepsilon)} \E \int_{\mathcal{M}} \big| \mu_{\varepsilon}  (C, 1, m) - \mu (C, 1, m) \big|\pi_{C, a} (m_j) \pi_{C, 0}(m_{-j}) \, \mathrm{d} m \\
             & \, + \E \sup_{\mu_{\varepsilon} \, \in \, \mathcal{U}(\varepsilon) }  \pn \bigg[ \int_{\mathcal{M}} \big| \mu_{\varepsilon}  (C, 1, m) - \mu (C, 1, m) \big|\pi_{C, a} (m_j) \pi_{C, 0}(m_{-j}) \, \mathrm{d} m \\
             & ~~~~~~~~~~~~~~~~~~~~~~~~~~~~ - \E \int_{\mathcal{M}} \big| \mu_{\varepsilon}  (C, 1, m) - \mu (C, 1, m) \big|\pi_{C, a} (m_j) \pi_{C, 0}(m_{-j}) \, \mathrm{d} m \bigg] \\
             \alignedoverset{\text{Cauchy Inequality and $\ell^2$ assumption}}{\leq} \varepsilon^2 \\
             & \, + \E \sup_{\mu_{\varepsilon} \, \in \, \mathcal{U}(\varepsilon) }  \pn \bigg[ \int_{\mathcal{M}} \big| \mu_{\varepsilon}  (C, 1, m) - \mu (C, 1, m) \big|\pi_{C, a} (m_j) \pi_{C, 0}(m_{-j}) \, \mathrm{d} m \\
            & ~~~~~~~~~~~~~~~~~~~~~~~~~~~~ - \E \int_{\mathcal{M}} \big| \mu_{\varepsilon}  (C, 1, m) - \mu (C, 1, m) \big|\pi_{C, a} (m_j) \pi_{C, 0}(m_{-j}) \, \mathrm{d} m \bigg] \\
            \alignedoverset{\text{Corollary 5.1 in \cite{chernozhukov2014gaussian}}}{\lesssim} \varepsilon^2 + n^{-1 / 2} \sqrt{n^\vartheta} \varepsilon \big( \log (\varepsilon) + \log n \big) \\
            \alignedoverset{\text{Let } \varepsilon  \rightarrow 0 \text{ with Assumption \ref{ass_fun_class}}}{\longrightarrow} 0.
        \end{aligned}
    \]
    \end{small}
    Therefore, we must have $\pn \int_{\mathcal{M}} \big\{ \widehat\mu (C, 1, m) - \mu (C, 1, m) \big\} \pi_{C, a} (m_j) \pi_{C, 0}(m_{-j}) \, \mathrm{d} m = o_p(1)$. Similarly, we have 
    \[
        \begin{aligned}
            & \pn \int_{\mathcal{M}}  \mu (C, 1, m) \big\{ \widehat\pi_{C, a} (m_j) - \pi_{C, a} (m_j) \big\} \pi_{C, 0}(m_{-j}) \, \mathrm{d} m \\
            = & \pn \int_{\mathcal{M}}  \mu (C, 1, m)  \pi_{C, a} (m_j) \big\{ \widehat\pi_{C, 0}(m_{-j}) - \pi_{C, 0}(m_{-j}) \big\} \, \mathrm{d} m = o_p(1),
        \end{aligned}
    \]
    which implies
    \[
        \pn \Big[ \psi_j^{(0)} \big( \widehat\mu, \widehat\pi_{C, a'}^{(j)}, \widehat\pi_{C, 0}^{(-j)}\big) - \psi_j^{(0)} \big( \mu, \pi_{C, a'}^{(j)}, \pi_{C, 0}^{(-j)}\big) \Big] = o_p(1)
    \]
    This yields $\pn \psi_j^{(0)} \big( \widehat\mu, \widehat\pi_{C, a'}^{(j)}, \widehat\pi_{C, 0}^{(-j)}\big) = o_p(1)$. Note that the expressions of $\psi_j^{(3)}(\cdot)$ and Assumption \ref{ass_pos} ensure
    \[
        \begin{aligned}
            \Big| \pn \psi_j^{(3)} \big( \widehat{e}_1, \widehat{\pi}_{C, a'}^{(-j)}, \widehat{\pi}_{C, 1, M_{j}}^{(-j)},  \widehat\mu \big) \Big| \lesssim \bigg| \frac{1}{n} \sum_{i = 1}^n \big\{ Y_i - \widehat\mu (C_i, 1, M_i)\big\} \bigg| = o_p(1),
        \end{aligned}
    \]
    by $\widehat{\mu}(\bcdot)$ is correctly estimated, which gives $ \pn \psi_j^{(3)} \big( \widehat{e}_1, \widehat{\pi}_{C, 0}^{(j)}, \widehat{\pi}_{C, a'}^{(-j)}, \widehat{\pi}_{C, 1, M_{-j}}^{(j)},  \widehat\pi_{C, 1}^{(-j)}, \widehat\mu \big) = o_p(1)$. Similarly,
    \[
        \begin{aligned}
            & \Big| \pn \psi_j^{(1)} \big( \widehat{e}_0, \widehat\mu, \widehat\pi_{C, a'}^{(j)}\big) \Big|\\
            \lesssim & \pn \bigg[ \int_{\mathcal{M}_j} \widehat\mu (C, 1, m_j, M_{-j}) \widehat\pi_{C, a'}(m_j) \, \mathrm{d} m_j - \zeta_j (a', 0, C) \bigg] \\
            \leq & \bigg| \pn \bigg[ \int_{\mathcal{M}_j} \mu (C, 1, m_j, M_{-j}) \pi_{C, a'}(m_j) \, \mathrm{d} m_j - \zeta_j (a', 0, C) \bigg] \bigg| \\
            &+ \bigg| \pn \int_{\mathcal{M}_j} \Big\{ \widehat\mu (C, 1, m_j, M_{-j}) \widehat\pi_{C, a'}(m_j) - \mu (C, 1, m_j, M_{-j}) \pi_{C, a'}(m_j) \Big\} \, \mathrm{d} m_j \bigg|.
        \end{aligned}
    \]
    The first term above is $o_p(1)$ obviously by the definition of $\zeta_j(a', 0, C)$. By using the same steps for proving $\pn \Big[ \psi_j^{(0)} \big( \widehat\mu, \widehat\pi_{C, a'}^{(j)}, \widehat\pi_{C, 0}^{(-j)}\big) - \psi_j^{(0)} \big( \mu, \pi_{C, a'}^{(j)}, \pi_{C, 0}^{(-j)}\big) \Big] = o_p(1)$ above, we can prove the second term is also $o_p(1)$ by both $\widehat{\mu}$ and $\widehat\pi_{C, a'}(m_j)$ are correctly estimated. This yields $ \pn \psi_j^{(1)} \big( \widehat{e}_0, \widehat\mu, \widehat\pi_{C, a'}^{(j)}\big) = o_p(1)$, and similarly gives $\pn \psi_j^{(2)} \big( \widehat{e}_{a'}, \widehat\mu, \widehat\pi_{C, 0}^{(-j)} \big) = o_p(1)$. The proof for \textbf{Part 0} is thus completed.
    \vspace{1ex}

    \textbf{Part 1.} When the model $\mathscr{M}_{j, \, 1}$ is correctly specified, we have consistency estimator $\widehat{e}_{a'}$, $\widehat\mu$, and $\widehat\pi_{C, a'}^{(j)}$. 
    % By using the same empirical process theory in Part 0, we have 
    % \[
    %     \begin{aligned}
    %         & \pn \psi_j^{(1)} \big( \widehat{e}_{a'}, \widehat\mu, \widehat\pi_{C, a'}^{(j)}\big) \\
    %         = & \underbrace{\pn \Big[ \psi_j^{(1)} \big( \widehat{e}_{a'}, \widehat\mu, \widehat\pi_{C, a'}^{(j)}\big) - \psi_j^{(1)} \big( {e}_{a'}, \mu, \pi_{C, a'}^{(j)}\big) \Big]}_{= o_p(1) \text{ by the consistency of } \widehat{e}_{a'}, \widehat\mu, \text{ and } \widehat\pi_{C, a'}^{(j)}} + \underbrace{\pn \psi_j^{(1)} ( {e}_{a'}, \mu, \pi_{C, a'}^{(j)})}_{= o_p(1) \text{ by } P \psi_j^{(1)} \big( {e}_{a'}, \mu, \pi_{C, a'}^{(j)}\big) = 0} = o_p(1).
    %     \end{aligned}
    % \]
    % The consistency of $\widehat{\mu}$ and Assumption \ref{ass_pos} again ensure $ \pn \psi_j^{(3)} \big( \widehat{e}_1, \widehat{\pi}_{C, 0}^{(j)}, \widehat{\pi}_{C, a'}^{(-j)}, \widehat{\pi}_{C, 1, M_{-j}}^{(j)}, \penalty 0 \widehat\pi_{C, 1}^{(-j)}, \widehat\mu \big) = o_p(1)$. 
    We first handle $\psi_j^{(0)}(\cdot) + \psi_j^{(1)}(\cdot)$. We rewrite
    \[
        \begin{aligned}
            & \psi_j^{(0)} \big( \mu, \pi_{C, a'}^{(j)}, \pi_{C, 0}^{(-j)}\big) + \psi_j^{(1)} \big( {e}_{0}, \mu, \pi_{C, a'}^{(j)} \big) \\
            = & \zeta_j (a', 0, C) - \E  \zeta_j (a', 0, C) +  \frac{\mathds{1}(A = 0)}{{e}_0(C)} \bigg\{ \int_{\mathcal{M}_j} \mu (C, 1, m_j, M_{-j}) \pi_{C, a'}(m_j) \, \mathrm{d} m_j - \zeta_j (a', 0, C) \bigg\} \\
            = & \frac{\zeta_j (a', 0, C)}{e_{0} (C)} \big[ \mathds{1} (A = 0) - e_{0}(C)\big] \\
            & ~~~~~~~~~~ + \frac{\mathds{1}(A = 0)}{{e}_0(C)} \bigg\{ \int_{\mathcal{M}_j} \mu (C, 1, m_j, M_{-j}) \pi_{C, a'}(m_j) \, \mathrm{d} m_j - \E  \zeta_j (a', 0, C).
        \end{aligned}
    \]
    Therefore, 
    \[
        \begin{aligned}
            & \Big| \pn \big[ \psi_j^{(0)} \big( \widehat\mu, \widehat\pi_{C, a'}^{(j)}, \widehat\pi_{C, 0}^{(-j)}\big) + \psi_j^{(1)} \big( \widehat{e}_{0}, \widehat\mu, \widehat\pi_{C, a'}^{(j)} \big) \big] \Big| \\
            \lesssim & \, \Bigg| \pn \bigg[ \frac{ \int_{\mathcal{M}} \widehat\mu (C, 1, m) \widehat\pi_{C, a} (m_j) \widehat\pi_{C, 0}(m_{-j}) \, \mathrm{d} m}{\widehat{e}_{0} (C)} \big\{ \mathds{1} (A = 0) - \widehat{e}_{0}(C)\big\} \bigg] \Bigg| \\
            & ~~~~~~~~ + \Bigg| \pn \bigg[ \frac{\mathds{1}(A = 0)}{\widehat{e}_0(C)} \int_{\mathcal{M}_j} \widehat\mu (C, 1, m_j, M_{-j}) \widehat\pi_{C, a'}(m_j) \, \mathrm{d} m_j - \E  \zeta_j (a', 0, C) \bigg] \Bigg|.
        \end{aligned}
    \]
    By the bounded assumption in Assumption \ref{ass_fun_class}, we have
    \[
        \begin{aligned}
            & \Bigg| \pn \bigg[ \frac{ \int_{\mathcal{M}} \widehat\mu (C, 1, m) \widehat\pi_{C, a} (m_j) \widehat\pi_{C, 0}(m_{-j}) \, \mathrm{d} m}{\widehat{e}_{0} (C)} \big\{ \mathds{1} (A = 0) - \widehat{e}_{0}(C)\big\} \bigg] \Bigg| \\
            \lesssim & \Big| \pn \big\{ \mathds{1} (A = 0) - \widehat{e}_{0}(C)\big\} \Big| = o_p(1)
        \end{aligned}
    \]
    by $\widehat{e}_{0} (\cdot)$ convergences to ${e}_{0} (\cdot)$ in $\ell^2$ norm and $\E \big\{ \mathds{1} (A = 0) - {e}_{0}(C)\big\} = 0$. On the other hand, we have 
    \[
        \begin{aligned}
            & \pn \bigg[ \frac{\mathds{1}(A = 0)}{\widehat{e}_0(C)} \bigg\{ \int_{\mathcal{M}_j} \widehat\mu (C, 1, m_j, M_{-j}) \widehat\pi_{C, a'}(m_j) \, \mathrm{d} m_j - \E  \zeta_j (a', 0, C) \bigg] \\
            \alignedoverset{\text{by consistency of $\widehat{e}_{0}$, $\widehat\mu$, and $\widehat\pi_{C, a'}^{(j)}$}}{=} \pn \bigg[ \frac{\mathds{1}(A = 0)}{{e}_0(C)} \int_{\mathcal{M}_j} \mu (C, 1, m_j, M_{-j}) \pi_{C, a'}(m_j) \, \mathrm{d} m_j - \E  \zeta_j (a', 0, C) \bigg] + o_p(1)\\
            & = \E \bigg[ \frac{\mathds{1}(A = 0)}{{e}_0(C)} \int_{\mathcal{M}_j} \mu (C, 1, m_j, M_{-j}) \pi_{C, a'}(m_j) \, \mathrm{d} m_j - \E  \zeta_j (a', 0, C) \bigg] + o_p(1) \\
            \alignedoverset{\text{by  \eqref{strategy_1_DM}}}{=} o_p(1),
        \end{aligned}
    \]
    which gives $\pn \big[ \psi_j^{(0)} \big( \widehat\mu, \widehat\pi_{C, a'}^{(j)}, \widehat\pi_{C, 0}^{(-j)}\big) + \psi_j^{(1)} \big( \widehat{e}_{0}, \widehat\mu, \widehat\pi_{C, a'}^{(j)} \big) \big] = o_p(1)$. It remains to show $\pn \big[ \psi_j^{(2)} \big( \widehat{e}_{a'}, \widehat\mu, \widehat\pi_{C, 0}^{(-j)} \big) + \psi_j^{(3)} \big( \widehat{e}_1, \widehat\pi_{C, a'}^{(-j)}, \widehat\pi_{C, 1, M_{j}}^{(-j)}, \widehat\mu \big) \big] = o_p(1)$. Indeed, the consistency of $\widehat\mu$ ensures $\pn \big[ \psi_j^{(3)} \big( \widehat{e}_1, \widehat\pi_{C, a'}^{(-j)}, \widehat\pi_{C, 1, M_{j}}^{(-j)}, \widehat\mu \big) \big] = o_p(1)$ as stated in \textbf{Part 0}, and 
    \[
        \begin{aligned}
            & \Big| \pn \big[ \psi_j^{(2)} \big( \widehat{e}_{a'}, \widehat\mu, \widehat\pi_{C, 0}^{(-j)} \big) \big] \Big| \\
            & = \bigg| \pn \bigg[ \frac{\mathds{1}(A = a')}{\widehat{e}_{a'} (C)} \bigg\{ \int_{\mathcal{M}_{-j}} \widehat\mu(C, 1, M_j, m_{-j}) \widehat\pi_{C, 0} (m_{-j}) \, \mathrm{d} m_{-j}  \\
            &~~~~~~~~~~~~~~~~~~~~~~~~~ - \int_{\mathcal{M}} \widehat\mu (C, 1, m) \widehat\pi_{C, 1} (m_j) \widehat\pi_{C, 0}(m_{-j}) \, \mathrm{d} m \bigg\} \bigg] \bigg| \\
            & \leq \pn \bigg[ \frac{\mathds{1}(A = a')}{\widehat{e}_{a'} (C)} \int_{\mathcal{M}_{-j}} \widehat\pi_{C, 0} (m_{-j}) \, \mathrm{d} m_{-j} \bigg| \int_{\mathcal{M}_j} \widehat\mu(C, 1, m_j, m_{-j}) \big[ \pi_{C, 1} (m_j) - \widehat\pi_{C, 1} (m_j) \big] \, \mathrm{d} m_j  \bigg| \bigg] \\
            \alignedoverset{\text{by $\ell^2$ assumption and Cauchy}}{\lesssim} \pn \bigg| \int_{\mathcal{M}_j} \big[ \pi_{C, 1} (m_j) - \widehat\pi_{C, 1} (m_j) \big]^2 \, \mathrm{d} m_j  \bigg| \\
            & = \E \bigg| \int_{\mathcal{M}_j} \big[ \pi_{C, 1} (m_j) - \widehat\pi_{C, 1} (m_j) \big]^2 \, \mathrm{d} m_j  \bigg| + o_p(1) = o_p(1),
        \end{aligned}
    \]
    which completes the proof required for \textbf{Part 1}.
    \vspace{1ex}

    \textbf{Part 2.} Exactly the same as \textbf{Part 1}, one can prove when model $\mathscr{M}_{j, \, 2}$ is correctly specified, we have 
    \[
        \begin{aligned}
            & \bigg| \pn \Big[ \psi_j^{(0)} \big( \widehat\mu, \widehat\pi_{C, a'}^{(j)}, \widehat\pi_{C, 0}^{(-j)}\big) + \psi_j^{(2)} \big( \widehat{e}_{a'}, \widehat\mu, \widehat\pi_{C, 0}^{(-j)} \big) \Big] \bigg| \\
            = & \bigg| \pn \bigg[ \int_{\mathcal{M}} \widehat\mu (C, 1, m) \widehat\pi_{C, 1} (m_j) \widehat\pi_{C, 0}(m_{-j}) \, \mathrm{d} m - \E  \zeta_j (a', 0, C) \\
            & ~~~~~~~~~~~~~~~~~~~~~~~~ +  \frac{\mathds{1}(A = a')}{\widehat{e}_{a'} (C)} \bigg\{ \int_{\mathcal{M}_{-j}} \widehat\mu(C, 1, M_j, m_{-j}) \widehat\pi_{C, 0} (m_{-j}) \, \mathrm{d} m_{-j} - \zeta_j (a', 0, C) \bigg\} \bigg] \bigg| \\
            \leq & \bigg| \pn \frac{\int_{\mathcal{M}} \widehat\mu (C, 1, m) \widehat\pi_{C, a} (m_j) \widehat\pi_{C, 0}(m_{-j}) \, \mathrm{d} m}{e_{a'} (C)} \big[ \mathds{1} (A = a') - \widehat{e}_{a'}(C)\big] \bigg| \\
            & ~~~~~~~~~~~~~~~~~~~~~~~~~ + \bigg| \pn \bigg[ \frac{\mathds{1}(A = a')}{\widehat{e}_{a'} (C)} \int_{\mathcal{M}_{-j}} \widehat\mu(C, 1, M_j, m_{-j}) \widehat\pi_{C, 0} (m_{-j}) \, \mathrm{d} m_{-j} - \E  \zeta_j (a', 0, C) \bigg] \bigg| \\
            \lesssim & \E \big[ \mathds{1} (A = a') - \widehat{e}_{a'}(C) \big]^2 \\
            & ~~~~~~ + \E \bigg[ \frac{\mathds{1}(A = a')}{{e}_{a'} (C)} \int_{\mathcal{M}_{-j}} \mu(C, 1, M_j, m_{-j}) \pi_{C, 0} (m_{-j}) \, \mathrm{d} m_{-j} - \E  \zeta_j (a', 0, C) \bigg] + o_p(1) \\
            \alignedoverset{\text{by \eqref{strategy_3_DM}}}{=} o_p(1) + 0 + o_p(1) = o_p(1).
        \end{aligned}
    \]
    Similarly, we can prove 
    \[
        \begin{aligned}
            & \Big| \pn \psi_j^{(1)} (\widehat{e}_0, \widehat{\mu}, \widehat\pi_{C, a'}^{(j)})\Big| \\
            \leq & \bigg|  \pn \bigg[ \frac{\mathds{1}(A = 0)}{{e}_0(C)} \bigg\{ \int_{\mathcal{M}_j} \widehat\pi_{C, a'}(m_j) \, \mathrm{d} m_j \int_{\mathcal{M}_{-j}} \widehat\mu (C, 1, m_j, m_{-j}) \big[ \pi_{C, 0} (m_{-j}) - \widehat\pi_{C, 0} (m_{-j})\big] \, \mathrm{d} m_{-j} \bigg] \bigg| \\
            \lesssim & \pn \bigg| \int_{\mathcal{M}_{-j}} \big[ \pi_{C, 0} (m_{-j}) - \widehat\pi_{C, 0} (m_{-j})\big]^2 \, \mathrm{d} m_{-j} \bigg| = o_p(1)
        \end{aligned}
    \]
    and $\pn \big[ \psi_j^{(3)} \big( \widehat{e}_1, \widehat\pi_{C, a'}^{(-j)}, \widehat\pi_{C, 1, M_{j}}^{(-j)}, \widehat\mu \big) \big] = o_p(1)$,
    which completes the proof of this part.
    \vspace{1ex}

    \textbf{Part 3.} When the model $\mathscr{M}_{j, \, 3}$ is correctly specified, we have $\widehat{e}_1$, $\widehat\pi_{C, a'}^{(-j)}$, and $\widehat\pi_{C, 1, M_{j}}^{(-j)}$ are consistent, as well as $\widehat\pi_{C, a', \Pa_j (\mathcal{G}_M)}^{(j)}$ thus $\widehat\pi_{C, a'}^{(j)}$, but not $\widehat{\mu}$. In this case, we first rewrite
    \[
         \begin{aligned}
             & \psi_j^{(0)} \big( \mu, \pi_{C, a'}^{(j)}, \pi_{C, 0}^{(-j)}\big) + \psi_j^{(3)} \big( e_1, \pi_{C, a'}^{(-j)}, \pi_{C, 1, M_{j}}^{(-j)}, \mu \big) \\
             = & \bigg[ \frac{\mathds{1}(A = 1)}{e_1(C)} \frac{ \pi_{C, a'}(M_{-j}) }{\pi_{C, 1, M_j}(M_{-j})} Y - \E \zeta_j(a', 0, C) \bigg] \\
             & ~~~~~~~~ +  \int_{\mathcal{M}} \mu (C, 1, m) \pi_{C, 1} (m_j) \pi_{C, a'}(m_{-j}) \, \mathrm{d} m - \frac{\mathds{1}(A = 1)}{e_1(C)} \frac{\pi_{C, a'}(M_{-j}) }{\pi_{C, 1, M_j}(M_{-j})} \mu(C, 1, M).
         \end{aligned}
    \]
    By using the same empirical processes technique in \textbf{Part 0}, we can prove
    \begin{small}
    \[
        \begin{aligned}
            & \bigg| \pn \Big[ \psi_j^{(0)} \big( \widehat\mu, \widehat\pi_{C, a'}^{(j)}, \widehat\pi_{C, 0}^{(-j)}\big) + \psi_j^{(3)} \big( \widehat{e}_1, \widehat\pi_{C, a'}^{(-j)}, \widehat\pi_{C, 1, M_{j}}^{(-j)}, \widehat\mu \big) \Big] \bigg| \\
            \leq &  \bigg|  \pn \bigg[ \frac{\mathds{1}(A = 1)}{\widehat{e}_1(C)} \frac{ \widehat\pi_{C, a'}(M_{-j}) }{\widehat\pi_{C, 1, M_j}(M_{-j})} Y - \E \zeta_j(a', 0, C) \bigg] \bigg| \\
            & ~~~ + \bigg| \pn \bigg[ \int_{\mathcal{M}} \widehat\mu (C, 1, m) \widehat\pi_{C, 1} (m_j) \widehat\pi_{C, a'}(m_{-j}) \, \mathrm{d} m - \frac{\mathds{1}(A = 1)}{\widehat{e}_1(C)} \frac{\widehat\pi_{C, a'}(M_{-j}) }{\widehat\pi_{C, 1, M_j}(M_{-j})} \widehat\mu(C, 1, M) \bigg] \bigg| \\
            = & \Bigg| \E \bigg[ \frac{\mathds{1}(A = 1)}{e_1(C)} \frac{ \pi_{C, a'}(M_{-j}) }{\pi_{C, 1, M_j}(M_{-j})} Y - \E \zeta_j(a', 0, C) \bigg] + o_p(1) \Bigg| \\
            & + \Bigg| \E \bigg[ \int_{\mathcal{M}} \widehat\mu (C, 1, m) \pi_{C, 1} (m_j) \pi_{C, a'}(m_{-j}) \, \mathrm{d} m - \frac{\mathds{1}(A = 1)}{e_1(C)} \frac{\pi_{C, a'}(M_{-j}) }{\pi_{C, 1, M_j}(M_{-j})} \widehat\mu(C, 1, M) \bigg] + o_p(1) \Bigg| \\
            \alignedoverset{\text{by \eqref{strategy_1_DM}}}{=}  \Bigg| \E \bigg[ \int_{\mathcal{M}} \widehat\mu (C, 1, m) \pi_{C, 1} (m_j) \pi_{C, a'}(m_{-j}) \, \mathrm{d} m - \frac{\mathds{1}(A = 1)}{e_1(C)} \frac{\pi_{C, a'}(M_{-j}) }{\pi_{C, 1, M_j}(M_{-j})} \widehat\mu(C, 1, M) \bigg]  \Bigg| + o_p(1) \\
            = & \Bigg| \E \bigg[ \int_{\mathcal{M}} \widehat\mu (C, 1, m) \pi_{C, 1} (m_j) \pi_{C, a'}(m_{-j}) \, \mathrm{d} m -\E \bigg[ \frac{\mathds{1}(A = 1)}{e_1(C)} \frac{\pi_{C, a'}(M_{-j}) }{\pi_{C, 1, M_j}(M_{-j})} \widehat\mu(C, 1, M) \mid C \bigg] \bigg] \Bigg| + o_p(1) \\
            = & 0 + o_p(1) = o_p(1),
        \end{aligned}
    \]
    \end{small}
    where we use the fact that 
    \[
        \begin{aligned}
            & \E \bigg[ \frac{\mathds{1}(A = 1)}{e_1(C)} \frac{\pi_{C, a'}(M_{-j}) }{\pi_{C, 1, M_j}(M_{-j})} \widehat\mu(C, 1, M) \mid C \bigg] \\
            = & \int_{\mathcal{M}} \frac{1}{f(a = 1 \mid C)} \frac{f(m_{-j} \mid C, a = a') }{f(m_{-j} \mid C, a = 1, m_j)} \widehat\mu(C, 1, m) \mathds{1}(a = 1) f(a, m \mid C) \, \mathrm{d} m \\
            = & \int_{\mathcal{M}} \frac{1}{f(a = 1 \mid C)} \frac{f(m_{-j} \mid C, a = a') }{f(m_{-j} \mid C, a = 1, m_j)} \widehat\mu(C, 1, m) \\
            & ~~~~~~~~~~~~~ \times f(a = 1\mid C) f(m_j \mid C, a = 1) f(m_{-j} \mid C, a = 1, m_j)  \, \mathrm{d} m \\
            = & \int_{\mathcal{M}} \widehat\mu (C, 1, m) f(m_j \mid C, a = 1) f(m_{-j} \mid C, a = a')  \, \mathrm{d} m = \int_{\mathcal{M}} \widehat\mu (C, 1, m) \pi_{C, 1} (m_j) \pi_{C, a'}(m_{-j}) \, \mathrm{d} m.
        \end{aligned}
    \]
    We can also prove that 
    \[
        \Big| \pn \psi_j^{(1)} (\widehat{e}_0, \widehat{\mu}, \widehat\pi_{C, a'}^{(j)})\Big| \lesssim \pn \bigg| \int_{\mathcal{M}_{-j}} \big[ \pi_{C, 0} (m_{-j}) - \widehat\pi_{C, 0} (m_{-j})\big]^2 \, \mathrm{d} m_{-j} \bigg| = o_p(1)
    \]
    and 
    \[
        \begin{aligned}
            \Big| \pn \big[ \psi_j^{(2)} \big( \widehat{e}_{a'}, \widehat\mu, \widehat\pi_{C, 0}^{(-j)} \big) \big] \Big| & \lesssim \pn \bigg| \int_{\mathcal{M}_j} \big[ \pi_{C, 1} (m_j) - \widehat\pi_{C, 1} (m_j) \big]^2 \, \mathrm{d} m_j  \bigg| \\
            & = \pn \bigg| \int_{\mathcal{M}_j} \int_{\mathcal{M}_{\pa_j} (\mathcal{G}_M)}\big[ \pi_{C, 1, \pa_j(\mathcal{G}_M)} (m_j) - \widehat\pi_{C, 1, \pa_j(\mathcal{G}_M)} (m_j) \big]^2 \, \mathrm{d} m_j  \bigg| \\
            & = o_p(1).
        \end{aligned}
    \]
    Thus, we finish proof of \textbf{Part 3}. The proof of consistency of $\widehat{DM}_j^{\text{QR}}$ is hence completed. \\

    \noindent \textbf{Semiparametric Efficiency of} $\widehat{TM}_j^{\text{avg}, \, \text{QR}}$:\\

    Next, we want to prove the semiparametric efficiency for $\widehat{TM}_j^{\text{avg}, \, \text{QR}}$. Here we have all estimators in model $\mathscr{M}_{j, \, \ell}, \ell = 0, 1, 2, 3$ are correct. We use $\widehat{P}$ to denote these correct estimated functions or estimated distribution, which is a consistent estimator of the true law of $X$, $P$. First note that, for any fixed $\mathcal{G}_M \in \operatorname{MEC}(\widehat{\mathcal{C}}_M)$ and $j \in [p]$, $\widehat{TM}_j^{\text{QR}}(\mathcal{G}_M)$ can be written as the following one-step estimators
    \[
        \begin{aligned}
            & \widehat{TM}_j^{\text{QR}}(\mathcal{G}_M) \\
            = & \mathbb{P}_n \Bigg[ \frac{\mathds{1}(A = 1)}{\widehat{e}_1(C)} \Big\{ Y - \widehat{\mu} (C, 1) \Big\} - \frac{\mathds{1}(A = 0)}{\widehat{e}_0(C)} \Big\{ Y - \widehat{\mu} (C, 0) \big] \Big\} \Bigg] \\
            & - \mathbb{P}_n \Bigg[ \frac{\mathds{1} (A = 1)}{\widehat{e}_1(C)} \widehat{\pi}_{C, 1} \big( \Pa_j(\mathcal{G}_M) \big) \Big\{ Y - \widehat{\mu} (C, 1, \Pa_j(\mathcal{G}_M), M_j)\Big\} \\
        & ~~~~~~~~~~~~~~~~~  - \frac{\mathds{1} (A = 0)}{\widehat{e}_0(C)} \widehat{\pi}_{C, 0} \big( \Pa_j(\mathcal{G}_M) \big)  \Big\{ Y - \widehat{\mu} (C, 0, \Pa_j(\mathcal{G}_M), M_j) \Big\} \\
        & +\int_{\mathcal{M}_j}  \bigg( \frac{\mathds{1} (A = 1)}{\widehat{e}_1(C)} \widehat{\mu} ( C, 1, \Pa_j(\mathcal{G}_M), m_j) - \frac{\mathds{1} (A = 0)}{\widehat{e}_0(C)} \widehat{\mu} ( C, 0, \Pa_j(\mathcal{G}_M), m_j)\bigg) \widehat{\pi}_C(m_j)  \, \mathrm{d} m_j \\
        & - \bigg( \frac{\mathds{1} (A = 1)}{\widehat{e}_1 (C)} \widehat{\E} \big[ \varrho_j (1, M_j, C\, ; \, \mathcal{G}_M) \mid C \big] - \frac{\mathds{1} (A = 0)}{\widehat{e}_0 (C)} \widehat{\E} \big[ \varrho_j (0, M_j, C \, ; \, \mathcal{G}_M) \mid C \big] \bigg) \Bigg] + \widehat{TM}_j^{\mathscr{M}_0} \\
        =: & \Big[ \pn \upvarphi_{j, 1}^{\text{QR}} (X, \widehat{P} \, ; \, \mathcal{G}_M ) + \uppsi_{j, 1} (\widehat{P} \, ; \, \mathcal{G}_M) \Big] - \Big[ \pn \upvarphi_{j, 0}^{\text{QR}} (X, \widehat{P} \, ; \, \mathcal{G}_M) + \uppsi_{j, 0} (\widehat{P} \, ; \, \mathcal{G}_M) \Big]
        \end{aligned}
    \]
    where $\uppsi_{j, a'}$ with $a' = 0, 1$ is defined as 
    \[
        \begin{aligned}
            \uppsi_{j, a'}: F_X \times \mathcal{G}_M \, \mapsto \, \E \Bigg[ & \bigg\{ \mu(a', C) - \int_{\mathcal{M}_{\pa_j} (\mathcal{G}_M)} \mu(C, a', \pa_j, M_j) \, \pi_{C, a'}(\pa_j)  \, \mathrm{d} \pa_j \bigg\} \Bigg].
        \end{aligned}
    \]
    and $\upvarphi_{j, a'}$ is the mapping from the observation and the underlying distribution of $X$ to the efficient score for $\uppsi_{j, a'}$, from Theorem \ref{thm_EIF}, we know that 
    \[
        \begin{aligned}
            & \upvarphi_{j, a'}^{\text{QR}} (X, Q \, ; \, \mathcal{G}_M) \\
            := & \frac{\mathds{1}(A = a')}{e_{a'}^Q (C)} \Big\{ Y - \kappa^Q (a', C) \Big\} + \kappa^Q (a', C) - \Bigg[ \frac{\mathds{1}(A = a')}{e_{a'}^Q (C)} \pi_{C, a'}^Q \big( \Pa_j (\mathcal{G}_M) \big) \Big\{  Y - \mu^Q (C, a', M) \Big\} \Bigg] \\
            & - \Bigg[ \frac{\mathds{1}(A = a')}{e_{a'}^Q (C)} \left\{ \int_{\mathcal{M}_j} \mu^Q (C, a', \Pa_j(\mathcal{G}_M), m_j) \pi_{C}^Q (m_j) \, \mathrm{d} m_j - \E_P \big[ \varrho_j^Q (a', M_j, C \, ; \, \mathcal{G}_M) \mid C \big]\right\} \\
            & \qquad \quad + \varrho_j^Q (a', M_j, C \, ; \, \mathcal{G}_M) \Bigg] - \uppsi_{j, a'} (Q \, ; \, \mathcal{G}_M) \\
            = & \upvarphi_{j, a'}^{(1)} (X, Q \, ; \, \mathcal{G}_M) - \upvarphi_{j, a'}^{(2)} (X, Q \, ; \, \mathcal{G}_M) - \uppsi_{j, a'} (Q \, ; \, \mathcal{G}_M) 
        \end{aligned}
    \]
    for arbitrary distribution $Q$ for $X$, where $e^Q_{a'}(c) := Q(A = a' \mid c)$ and similar definition for $k^Q(a', c)$, $\pi_{S}^Q (m_T)$, $\mu^{Q} (\cdot)$, and $\rho_j^Q(a', m_j, c \, ; \, \mathcal{G}_M)$, where 
   \[
       	\begin{aligned}
       		& \upvarphi_{j, a'}^{(1)} (X, Q \, ; \, \mathcal{G}_M) \\
       		:= & \frac{\mathds{1}(A = a')}{e_{a'}^Q (C)} \Big\{ Y - \kappa^Q (a', C) \Big\} + \kappa^Q (a', C) - \Bigg[ \frac{\mathds{1}(A = a')}{e_{a'}^Q (C)} \pi_{C, a'}^Q \big( \Pa_j (\mathcal{G}_M) \big) \Big\{  Y - \mu^Q (C, a', M) \Big\} \Bigg] \\
                &  + \varrho_j^Q (a', M_j, C \, ; \, \mathcal{G}_M),
       	\end{aligned}
   \]
   and
   \[
       \begin{aligned}
       	    & \upvarphi_{j, a'}^{(2)} (X, Q \, ; \, \mathcal{G}_M) \\
            := & \frac{\mathds{1}(A = a')}{e_{a'}^Q (C)} \left\{ \int_{\mathcal{M}_j} \mu^Q (C, a', \Pa_j(\mathcal{G}_M), m_j) \pi_{C}^Q (m_j) \, \mathrm{d} m_j - \E_P \big[ \varrho_j^Q (a', M_j, C \, ; \, \mathcal{G}_M) \mid C \big]\right\}.
       \end{aligned}
   \]
   We can see that $\uppsi_{j, 1} - \uppsi_{j, 0}$ the functional mapping of the underlying distribution of $X$ to $TM_j(\mathcal{G}_M)$ for a fixed DAG $\mathcal{G}_M$. Therefore, we can decompose 
    \[
        \begin{aligned}
            & \widehat{TM}_j^{\text{QR}}(\mathcal{G}_M) - TM_j(\mathcal{G}_M) \\
            = & \Big\langle \left(\mathbb{P}_n - P\right)\upvarphi_{j, \bcdot}^{\text{QR}}(X, P \, ; \, \mathcal{G}_M ) + \left(\mathbb{P}_n - P\right) \big[ \upvarphi_{j, \bcdot}^{\text{QR}}(X, \widehat{P} \, ; \, \mathcal{G}_M) -\upvarphi_{j, \bcdot}^{\text{QR}}(X, P \, ; \, \mathcal{G}_M ) \big] + R_{j, \bcdot}^{\text{QR}}(\widehat{P}, {P} \, ; \, \mathcal{G}_M) \Big\rangle \\
            % = & \Big[ \left(\mathbb{P}_n - P\right)\upvarphi_{j, 1}^{\text{QR}}(X, P \, ; \, \mathcal{G}_M ) + \left(\mathbb{P}_n - P\right) \{\upvarphi_{j, 1}^{\text{QR}}(X, \widehat{P} \, ; \, \mathcal{G}_M) -\upvarphi_{j, 1}^{\text{QR}}(X, P \, ; \, \mathcal{G}_M ) \} + R_{j, 1}^{\text{QR}}(\widehat{P}, {P} \, ; \, \mathcal{G}_M) \Big] \\
            % & - \Big[ \left(\mathbb{P}_n - P\right)\upvarphi_{j, 0}^{\text{QR}}(X, P \, ; \, \mathcal{G}_M) + \left(\mathbb{P}_n - P\right) \{\upvarphi_{j, 0}^{\text{QR}} (X, \widehat{P} \, ; \, \mathcal{G}_M) -\upvarphi_{j, 0}^{\text{QR}}(X, P \, ; \, \mathcal{G}_M ) \} + R_{j, 0}^{\text{QR}}(\widehat{P}, {P} \, ; \, \mathcal{G}_M) \Big],
        \end{aligned}
    \]
    where
    \[
        R_{j, a'}^{\text{QR}} (\widehat{P}, {P} \, ; \, \mathcal{G}_M) = \uppsi_{j, a'}(\widehat{P} \, ; \, \mathcal{G}_M) - \uppsi_{j, a'}({P} \, ; \, \mathcal{G}_M) + \int \upvarphi_{j, a'}^{\text{QR}}(x, \widehat{P} \, ; \, \mathcal{G}_M) \, \mathrm{d} P(x).
    \]
    Given Assumption \ref{ass_nor_cpdag}, by using the argument in Proof of Theorem \ref{thm_CI_IM}, we can prove
    \begin{equation}\label{TM_QR_avg_decompose}
        \begin{aligned}
            & \widehat{TM}_j^{\text{avg}, \, \text{QR}} - \overline{TM}_j \\
            = & \frac{1}{\# \operatorname{MEC}(\mathcal{C}_M)} \sum_{\mathcal{G}_M \in \operatorname{MEC}(\mathcal{C}_M)} (\mathbb{P}_n - P) \big\{ \upvarphi_{j, 1}^{\text{QR}} (X, P \, ; \, \mathcal{G}_M ) - \upvarphi_{j, 0}^{\text{QR}} (X, P \, ; \, \mathcal{G}_M ) \big\} \\
            & + \frac{1}{\# \operatorname{MEC}(\mathcal{C}_M)} \sum_{\mathcal{G}_M \in \operatorname{MEC}(\mathcal{C}_M)} \left(\mathbb{P}_n - P\right) \Big[ \{\upvarphi_{j, 1}^{\text{QR}} (X, \widehat{P} \, ; \, \mathcal{G}_M) -\upvarphi_{j, 1}^{\text{QR}} (X, P \, ; \, \mathcal{G}_M ) \} \\
            & ~~~~~~~~~~~~~~~~~~~~~~~~~~~~~~~~~~~~~~~~~~~~~~~~~~~~~ - \{\upvarphi_{j, 0}^{\text{QR}} (X, \widehat{P} \, ; \, \mathcal{G}_M) - \upvarphi_{j, 0}^{\text{QR}} (X, P \, ; \, \mathcal{G}_M ) \} \Big] \\
            & + \frac{1}{\# \operatorname{MEC}(\mathcal{C}_M)} \sum_{\mathcal{G}_M \in \operatorname{MEC}(\mathcal{C}_M)} \big\{ R_{j, 1}^{\text{QR}} (\widehat{P}, {P} \, ; \, \mathcal{G}_M) - R_{j, 0}^{\text{QR}} (\widehat{P}, {P} \, ; \, \mathcal{G}_M) \big\} + o_p (\ell_n^{-1})
        \end{aligned}
    \end{equation}
    with any $\ell_n$ such that $\lim_{n \rightarrow \infty} \ell_n = \infty$.
    The (vector-)function classes for $\upvarphi_{j, a'}^{\text{QR}} (X, P \, ; \, \mathcal{G}_M)$ with $a' \in \{ 0, 1\}$ is $q = (q_{k})_{k \in \mathcal{K}} \in \mathcal{Q} := \bigotimes_{\{ T, S\} \text{ used in \eqref{DM_MR} and \eqref{IM_MR} for a fixed $j \in [p]$}} \penalty 0 \mathcal{E} \otimes \mathcal{F}_{T \, \mid \, S} \otimes \mathcal{U}_{S}$, which is also VC-class with index the same index $\vartheta_j \in [0, 1 / 2)$ by Lemma 2.6.18 in \cite{van1996weak}. 
    On the other hand, note that the consistency in Assumption \ref{ass_conv_rate_QR} ensures there exists sufficient large $n$ such that all denominators in $\upvarphi_{j, a'}^{\text{QR}} (X, \widehat{P} \, ; \, \mathcal{G}_M)$ will be larger than $\varepsilon / 2$, then by using $|a_1 a_2 - b_1 b_2| \leq |b_1| |a_2 - b_2| + |b_2||a_1 - b_1| + |a_1 - a_2| |b_1 - b_2|$ for any real number $a_1, a_2, b_1, b_2$, we can show that
    \begin{equation}\label{T1_rem_QR}
    	\begin{aligned}
    		& \Big| (\pn - P) \big\{ \upvarphi_{j, a'}^{\text{QR}} (X, \widehat{P} \, ; \, \mathcal{G}_M) - \upvarphi_{j, a'}^{\text{QR}} (X, P \, ; \, \mathcal{G}_M) \big\} \Big| \\
    		\alignedoverset{\text{by $\ell^2$ assumption}}{\lesssim} \underbrace{ \sum_{k \in \mathcal{K}} \Big| (\pn - P) \big\{ \widehat{q}_k (X; \mathcal{G}_M) - {q}_k (X; \mathcal{G}_M)\big\} \Big|}_{\text{comes from } \upvarphi_{j, a'}^{(1)} (X, \widehat{P} \, ; \, \mathcal{G}_M) - \upvarphi_{j, a'}^{(1)} (X, {P} \, ; \, \mathcal{G}_M)} \\
    		&~~~~~~~ + \underbrace{\Bigg| (\pn - P) \int_{\mathcal{M}_{\pa_j}(\mathcal{G}_M)}  \Big[  \widehat{\pi}_{C, a'} \big(\pa_j( \mathcal{G}_M) \big) - {\pi}_{C, a'} \big(\pa_j( \mathcal{G}_M) \big) \Big] \, \mathrm{d} \pa_j (\mathcal{G}_M) \Bigg|}_{\text{comes from } \upvarphi_{j, a'}^{(2)} (X, \widehat{P} \, ; \, \mathcal{G}_M) - \upvarphi_{j, a'}^{(2)} (X, {P} \, ; \, \mathcal{G}_M)} \\
    		&~~~~~~~ + \underbrace{\Bigg| (\pn - P) \int_{\mathcal{M}_{\pa_j} (\mathcal{G}_M)} \Big[ \widehat\mu(C, a, \pa_j, M_j) - \mu(C, a, \pa_j, M_j)  \Big] \, \mathrm{d} \pa_j \Bigg|}_{\text{comes from } \upvarphi_{j, a'}^{(2)} (X, \widehat{P} \, ; \, \mathcal{G}_M) - \upvarphi_{j, a'}^{(2)} (X, {P} \, ; \, \mathcal{G}_M)} \\
    		&~~~~~~~ + \underbrace{\Bigg| (\pn - P) \int_{\mathcal{M}_{\pa_j} (\mathcal{G}_M)} \Big[ \widehat\pi_{C, a}(\pa_j) - \pi_{C, a}(\pa_j) \Big] \, \mathrm{d} \pa_j \Bigg|}_{\text{comes from } \upvarphi_{j, a'}^{(2)} (X, \widehat{P} \, ; \, \mathcal{G}_M) - \upvarphi_{j, a'}^{(2)} (X, {P} \, ; \, \mathcal{G}_M)} \\
    		&~~~~~~~ + \underbrace{\Big| (\pn - P) \big[ \widehat\mu(a', C) - \mu (a', C) \big] \Big|}_{\text{comes from } \uppsi_{j, a'} (\widehat{P} ; \mathcal{G}_M) - \uppsi_{j, a'} ({P} ; \mathcal{G}_M)}\\
            &~~~~~~~ + \underbrace{\Bigg| (\pn - P) \bigg[ \int_{\mathcal{M}_{\pa_j} (\mathcal{G}_M)} \Big[ \widehat\mu(C, a', \pa_j, M_j)  - \mu(C, a', \pa_j, M_j) \Big] \, \mathrm{d} \pa_j \bigg] \Bigg|}_{\text{comes from } \uppsi_{j, a'} (\widehat{P} ; \mathcal{G}_M) - \uppsi_{j, a'} ({P} ; \mathcal{G}_M)} \\
    	&~~~~~~~ + \underbrace{\Bigg| (\pn - P) \bigg[ \int_{\mathcal{M}_{\pa_j} (\mathcal{G}_M)} \Big[ \widehat\pi_{C, a'}(\pa_j) - \pi_{C, a'}(\pa_j) \Big] \, \mathrm{d} \pa_j \bigg] \Bigg|}_{\text{comes from } \uppsi_{j, a'} (\widehat{P} ; \mathcal{G}_M) - \uppsi_{j, a'} ({P} ; \mathcal{G}_M)}.
    	\end{aligned}
    \end{equation}
    For any $k \in \mathcal{K}$, $\widehat{q}_k$ converge with $\ell^2$-norm to $q_k$ at a rate of $n^{- \vartheta_{q_k}^*}$ in class $\mathcal{Q}$ with $\vartheta_{q_k}^* \in \{ \vartheta_{j, e}^*, \vartheta_{j, \pi}^*, \vartheta_{j, \mu}^* \}$ . Then by Corollary 5.1 in \cite{chernozhukov2014gaussian}, we have 
    \[
        \begin{aligned}
            & \E \Big| (\pn - P) \big\{ \widehat{q}_k (X; \mathcal{G}_M) - {q}_k (X; \mathcal{G}_M)\big\} \Big|\\
            \asymp & \, n^{-1/2} \sqrt{n^{\vartheta_j}} n^{- \vartheta_{q_k}^*} \big( \log  n^{- \vartheta_{q_k}^*} + \log n \big) \\
            \asymp \, & O (n^{- 1 / 2 + \vartheta_j / 2 - \vartheta_{q_k}^*} \log n) = o\big( n^{-1/2} \big)
        \end{aligned}
    \]
    by $\vartheta_j / 2 - \min\{ \vartheta_{j, e}^*, \vartheta_{j, \pi}^*, \vartheta_{j, \mu}^* \} < 0$ in Assumption \ref{ass_conv_rate_QR},
    which implies $(\pn - P) \big\{ \widehat{q}_k (X; \mathcal{G}_M) - {q}_k (X; \mathcal{G}_M)\big\} = o_p(n^{- 1 / 2})$ for any $\mathcal{G}_M \in \operatorname{MEC}(\mathcal{C}_M)$. Thus, 
    \[
        \begin{aligned}
            & (\pn - P) \big[ \widehat\mu(a', C) - \mu (a', C) \big] = o_p(n^{- 1 / 2}), \\
            & (\pn - P) \bigg[ \int_{\mathcal{M}_{\pa_j} (\mathcal{G}_M)} \Big[ \widehat\mu(C, a', \pa_j, M_j)  - \mu(C, a', \pa_j, M_j) \Big] \, \mathrm{d} \pa_j \bigg] = o_p(n^{- 1 / 2}), \\
            & (\pn - P) \bigg[ \int_{\mathcal{M}_{\pa_j} (\mathcal{G}_M)} \Big[ \widehat\pi_{C, a'}(\pa_j) - \pi_{C, a'}(\pa_j) \Big] \, \mathrm{d} \pa_j =  o_p(n^{- 1 / 2}).
        \end{aligned}
    \]
    Similarly, for any $\mathcal{G}_M \in \operatorname{MEC}(\mathcal{C}_M)$ we can prove 
    \[
        \begin{aligned}
            & \E \Bigg| (\pn - P) \int_{\mathcal{M}_{\pa_j}(\mathcal{G}_M)}  \Big[  \widehat{\pi}_{C, a'} \big(\pa_j( \mathcal{G}_M) \big) - {\pi}_{C, a'} \big(\pa_j( \mathcal{G}_M) \big) \Big] \, \mathrm{d} \pa_j (\mathcal{G}_M) \Bigg| = o (n^{-1/2}), \\
            & \E \Bigg| (\pn - P) \int_{\mathcal{M}_{\pa_j} (\mathcal{G}_M)} \Big[ \widehat\mu(C, a, \pa_j, M_j) - \mu(C, a, \pa_j, M_j)  \Big] \, \mathrm{d} \pa_j \Bigg| = o (n^{-1/2}), \\
            & \E \Bigg| (\pn - P) \int_{\mathcal{M}_{\pa_j} (\mathcal{G}_M)} \Big[ \widehat\pi_{C, a}(\pa_j) - \pi_{C, a}(\pa_j) \Big] \, \mathrm{d} \pa_j \Bigg| = o (n^{-1/2}).
        \end{aligned}
    \]
    Therefore, we conclude
    \begin{equation}\label{reminder1_final_QR}
        \begin{aligned}
            & \frac{1}{\# \operatorname{MEC}(\mathcal{C}_M)}  \sum_{\mathcal{G}_M \in \operatorname{MEC}(\mathcal{C}_M)} (\mathbb{P}_n - P) \Big[ \{\upvarphi_{j, 1}^{\text{QR}} (X, \widehat{P} \, ; \, \mathcal{G}_M) -\upvarphi_{j, 1}^{\text{QR}} (X, P \, ; \, \mathcal{G}_M ) \} \Big]  = o_p \big(n^{-1 / 2} \big), \\
            & \frac{1}{\# \operatorname{MEC}(\mathcal{C}_M)} \sum_{\mathcal{G}_M \in \operatorname{MEC}(\mathcal{C}_M)} (\mathbb{P}_n - P) \{\upvarphi_{j, 0}^{\text{QR}} (X, \widehat{P} \, ; \, \mathcal{G}_M) -\upvarphi_{j, 0}^{\text{QR}} (X, P \, ; \, \mathcal{G}_M ) \}  = o_p \big(n^{-1 / 2} \big).
        \end{aligned}
     \end{equation}
    It remains to deal with the reminder $R_{j, a'}^{\text{QR}} (\widehat{P}, P)$ for $a' \in \{ 0, 1\}$. Note that for any distribution $Q$, we rewrite
    \begin{small}
    \[
        \begin{aligned}
            & R_{j, a'}^{\text{QR}} (Q, P \, ; \, \mathcal{G}_M) \\
            = & \uppsi_{j, a'} (Q \, ; \, \mathcal{G}_M) - \uppsi_{j, a'} (P \, ; \, \mathcal{G}_M) + \int \upvarphi_{j, a'}^{\text{QR}} (x, Q \, ; \, \mathcal{G}_M) \, \mathrm{d} P(x) \\
            = & - \uppsi_{j, a'} (P \, ; \, \mathcal{G}_M) + \int \Bigg\{ \frac{\mathds{1}(a = a')}{e_{a'}^Q (c)} \Big\{ y - \kappa^Q (a', c) \Big\} + \kappa^Q (a', c) \\
            & \qquad - \Bigg[ \frac{\mathds{1}(a = a')}{e_{a'}^Q (c)} \pi_{C, a'}^Q \big( \Pa_j (\mathcal{G}_M) \big)  \Big\{  y - \mu^Q (c, a', m) \Big\} \\
            & \qquad \quad + \frac{\mathds{1}(a = a')}{e_{a'}^Q (c)} \left( \int_{\mathcal{M}_j} \mu^Q (c, a', \pa_j(\mathcal{G}_M), m_j) \pi_{c}^Q (m_j) \, \mathrm{d} m_j - \E_P \big[ \varrho_j^Q (a', M_j, c \, ; \, \mathcal{G}_M) \mid c \big]\right) \\
            & \qquad \quad + \varrho_j^Q (a', m_j, c \, ; \, \mathcal{G}_M) \Bigg] \Bigg\} \, \mathrm{d} P(x) \\
            & = - \int \Big[ \kappa (a', c) - \varrho_j(a', m_j, c \, ; \, \mathcal{G}_M) \Big] \, \mathrm{d} P(x) + \int \Bigg\{ \frac{\mathds{1}(a = a')}{e_{a'}^Q (c)} \Big\{ y - \kappa^Q (a', c) \Big\} + \kappa^Q (a', c) \\
            & \qquad - \Bigg[ \frac{\mathds{1}(a = a')}{e_{a'}^Q (c)} \frac{\pi_{c}^Q (m_j)}{\pi_{c, a', \pa_j (\mathcal{G}_M)}^Q (m_j)}\Big\{  y - \mu^Q (c, a', m) \Big\} \\
            & \qquad \quad + \frac{\mathds{1}(a = a')}{e_{a'}^Q (c)} \left( \int_{\mathcal{M}_j} \mu^Q (c, a', \pa_j(\mathcal{G}_M), m_j) \pi_{c}^Q (m_j) \, \mathrm{d} m_j - \E_P \big[ \varrho_j^Q (a', M_j, c \, ; \, \mathcal{G}_M) \mid c \big]\right) \\
            & \qquad \quad + \varrho_j^Q (a', m_j, c \, ; \, \mathcal{G}_M) \Bigg] \Bigg\} \, \mathrm{d} P(x) \\
            \alignedoverset{\text{by rearranging}}{=}\int \Bigg\{ \frac{e_{a'} (c)}{e_{a'}^Q (c)} \Big\{ \kappa(a', c) - \kappa^Q (a', c) \Big\} + \Big\{ \kappa^Q (a', c) - \kappa(a', c) \Big\} \\
            & \qquad - \Bigg[ \frac{e_{a'}(c)}{e_{a'}^Q (c)} \frac{\pi_{c}^Q (m_j)}{\pi_{c, a', \pa_j (\mathcal{G}_M)}^Q (m_j)}\Big\{  \mu (c, a', m) - \mu^Q (c, a', m) \Big\} \\
            & \qquad \quad + \frac{e_{a'}(c)}{e_{a'}^Q (c)} \left( \int_{\mathcal{M}_j} \mu^Q (c, a', \pa_j(\mathcal{G}_M), m_j) \pi_{c}^Q (m_j) \, \mathrm{d} m_j - \E_P \big[ \varrho_j^Q (a', M_j, c \, ; \, \mathcal{G}_M) \mid c \big]\right) \\
            & \qquad \quad + \varrho_j^Q (a', m_j, c \, ; \, \mathcal{G}_M) - \varrho_j (a', m_j, c \, ; \, \mathcal{G}_M) \Bigg] \Bigg\} \, \mathrm{d} P(x) \\
            % = & \int \Bigg\{ \frac{e_{a'} (c)}{e_{a'}^Q (c)} \Big\{ \kappa(a', c) - \kappa^Q (a', c) \Big\} + \Big\{ \kappa^Q (a', c) - \kappa(a', c) \Big\} \Bigg\} \, \mathrm{d} P(x) \\
            % & - \int \, \mathrm{d} P(x) \bigg[ \frac{e_{a'}(c)}{e_{a'}^Q (c)} - 1\bigg]  \Big\{  \mu (c, a', m) - \mu^Q (c, a', m) \Big\} \pi_{c, a'}^Q \big(\pa_j (\mathcal{G}_M) \big)  \, \mathrm{d} \pa_j (\mathcal{G}_M) \\
            % & + \int \mu (c, a', m) \Big\{  \pi_{c, a'} \big(\pa_j (\mathcal{G}_M) \big) - \pi_{c, a'}^Q \big(\pa_j (\mathcal{G}_M) \big) \Big\} \, \mathrm{d} \pa_j (\mathcal{G}_M) \\
            =: & R_{j, a'}^{(1)} (Q, P) + R_{j, a'}^{(2)} (Q, P \, ; \, \mathcal{G}_M) - R_{j, a'}^{(3)} (Q, P \, ; \, \mathcal{G}_M) + R_{j, a'}^{(4)} (Q, P \, ; \, \mathcal{G}_M).
        \end{aligned}
    \]
    \end{small}
    In the equation, we use the identities  
    \[
        \begin{aligned}
            & \int \bigg[ \frac{e_{a'}(c)}{e_{a'}^Q (c)} \frac{\pi_{c}^Q (m_j)}{\pi_{c, a', \pa_j (\mathcal{G}_M)}^Q (m_j)}\Big\{  \mu (c, a', m) - \mu^Q (c, a', m) \Big\} \\
            & ~~~~~~~~~~ + \varrho_j^Q (a', m_j, c \, ; \, \mathcal{G}_M) - \varrho_j (a', m_j, c \, ; \, \mathcal{G}_M) \bigg] \, \mathrm{d} P(x) \\
            & \overset{\text{expanding the formula of }  \varrho_j^Q}{=} \int \, \mathrm{d} P(x) \Bigg[ \frac{e_{a'}(c)}{e_{a'}^Q (c)} \pi_{c, a'}^Q \big(\pa_j (\mathcal{G}_M) \big) \Big\{  \mu (c, a', m) - \mu^Q (c, a', m) \Big\} \\
            & ~~~~~~~~~~~~~~~~~~~~~~~~~~~~~~~~~~~~ - \int \bigg[  \Big\{  \mu (c, a', m) - \mu^Q (c, a', m) \Big\} \pi_{c, a'}^Q \big(\pa_j (\mathcal{G}_M) \big) \\
            & ~~~~~~~~~~~~~~~~~~~~~~~~~~~~~~~~~ +  \mu (c, a', m) \Big\{  \pi_{c, a'} \big(\pa_j (\mathcal{G}_M) \big) - \pi_{c, a'}^Q \big(\pa_j (\mathcal{G}_M) \big) \Big\} \bigg] \, \mathrm{d} \pa_j (\mathcal{G}_M) \\
            & \overset{\text{rearranging}}{=} \int \, \mathrm{d} P(x) \bigg[ \frac{e_{a'}(c)}{e_{a'}^Q (c)} - 1\bigg]  \Big\{  \mu (c, a', m) - \mu^Q (c, a', m) \Big\} \pi_{c, a'}^Q \big(\pa_j (\mathcal{G}_M) \big)  \, \mathrm{d} \pa_j (\mathcal{G}_M) \\
            & ~~~~~~~~~~~~~~ - \int \mu (c, a', m) \Big\{  \pi_{c, a'} \big(\pa_j (\mathcal{G}_M) \big) - \pi_{c, a'}^Q \big(\pa_j (\mathcal{G}_M) \big) \Big\} \, \mathrm{d} \pa_j (\mathcal{G}_M),
        \end{aligned} 
    \]
    and
    \[
        \begin{aligned}
            & \int \mu (c, a', m) \Big\{  \pi_{c, a'} \big(\pa_j (\mathcal{G}_M) \big) - \pi_{c, a'}^Q \big(\pa_j (\mathcal{G}_M) \big) \Big\} \, \mathrm{d} \pa_j (\mathcal{G}_M) \\
            & + \int \, \frac{e_{a'}(c)}{e_{a'}^Q (c)} \left( \int_{\mathcal{M}_j} \mu^Q (c, a', \pa_j(\mathcal{G}_M), m_j) \pi_{c}^Q (m_j) \, \mathrm{d} m_j - \E_P \big[ \varrho_j^Q (a', M_j, c \, ; \, \mathcal{G}_M) \mid c \big]\right) \mathrm{d} P(x) \\
            & \overset{\text{expanding the expectation}}{=} \int \, \mathrm{d} P(x) \int \mu (c, a', m) \Big\{  \pi_{c, a'} \big(\pa_j (\mathcal{G}_M) \big) - \pi_{c, a'}^Q \big(\pa_j (\mathcal{G}_M) \big) \Big\} \, \mathrm{d} \pa_j (\mathcal{G}_M) \\
            & ~~~~~~~~~~~~~~~~~~~~~~~~~~~ - \int \frac{e_{a'}(c)}{e_{a'}^Q (c)} \, \mathrm{d} P(x) \int_{\mathcal{M}_j} \mu^Q (c, a', m) \pi_{c}^Q (m_j) \, \mathrm{d} m_j \\
            & ~~~~~~~~~~~~~~~~~~~~~~~~~~~ \int_{\mathcal{M}_{\pa_j}(\mathcal{G}_M)}  \Big\{ {\pi}_{C, a'} \big(\pa_j( \mathcal{G}_M) \big) - {\pi}_{C, a'}^Q \big(\pa_j( \mathcal{G}_M) \big) \Big\} \, \mathrm{d} \pa_j (\mathcal{G}_M) \\
            & \overset{\text{rearranging}}{=} \int \, \mathrm{d} P(x) \int \mu (c, a', m) \Big\{  \pi_{c, a'} \big(\pa_j (\mathcal{G}_M) \big) - \pi_{c, a'}^Q \big(\pa_j (\mathcal{G}_M) \big) \Big\} \, \mathrm{d} \pa_j (\mathcal{G}_M) \\
            & ~~~~~~~~~~~~~~~~~~~~~~~~~~~ - \int \frac{e_{a'}(c)}{e_{a'}^Q (c)} \, \mathrm{d} P(x) \int_{\mathcal{M}_j} \mu (c, a', m) \pi_{c}^Q (m_j) \, \mathrm{d} m_j \\
            & ~~~~~~~~~~~~~~~~~~~~~~~~~~~ \int_{\mathcal{M}_{\pa_j}(\mathcal{G}_M)}  \Big\{ {\pi}_{C, a'} \big(\pa_j( \mathcal{G}_M) \big) - {\pi}_{C, a'}^Q \big(\pa_j( \mathcal{G}_M) \big) \Big\} \, \mathrm{d} \pa_j (\mathcal{G}_M) \\
            & ~~~~~~~~~~~~~~~~~~~~~~~~~~~ + \int \frac{e_{a'}(c)}{e_{a'}^Q (c)} \, \mathrm{d} P(x) \int_{\mathcal{M}_j} \Big\{ \mu (c, a', m) - \mu^Q (c, a', m) \Big\} \pi_{c}^Q (m_j) \, \mathrm{d} m_j \\
            & ~~~~~~~~~~~~~~~~~~~~~~ \int_{\mathcal{M}_{\pa_j}(\mathcal{G}_M)}  \Big\{ {\pi}_{C, a'} \big(\pa_j( \mathcal{G}_M) \big) - {\pi}_{C, a'}^Q \big(\pa_j( \mathcal{G}_M) \big) \Big\} \, \mathrm{d} \pa_j (\mathcal{G}_M) \\
            & \overset{\text{rearranging}}{=} - \int \left( \frac{e_{a'}(c)}{e_{a'}^Q (c)} - 1 \right) \, \mathrm{d} P(x) \int_{\mathcal{M}_j} \mu (c, a', m) \pi_{c}^Q (m_j) \, \mathrm{d} m_j \\
            & ~~~~~~~~~~~~~~~~~~~~~~ \int_{\mathcal{M}_{\pa_j}(\mathcal{G}_M)}  \Big\{ {\pi}_{C, a'} \big(\pa_j( \mathcal{G}_M) \big) - {\pi}_{C, a'}^Q \big(\pa_j( \mathcal{G}_M) \big) \Big\} \, \mathrm{d} \pa_j (\mathcal{G}_M) \\
            & ~~~~~~ + \int \frac{e_{a'}(c)}{e_{a'}^Q (c)}\pi_{c}^Q (m_j) \, \mathrm{d} P(x)   \\
            & ~~~~~~ \int_{\mathcal{M}_{\pa_j}(\mathcal{G}_M)}  \Big(  {\pi}_{C, a'} \big(\pa_j( \mathcal{G}_M) \big) - {\pi}_{C, a'}^Q \big(\pa_j( \mathcal{G}_M) \big) \Big) \Big( \mu (c, a', m) - \mu^Q (c, a', m) \Big) \, \mathrm{d} \pa_j (\mathcal{G}_M),
        \end{aligned}
    \]
    and then the components $R_{j, a'}^{(\ell)}(\bcdot)$ with $\ell = 1, 2, 3, 4$ can be rewritten as following
    \begin{small}
    \[
        \begin{aligned}
            R_{j, a'}^{(1)} (Q, P) & = \int \Bigg[ \frac{e_{a'} (c)}{e_{a'}^Q (c)} \Big\{ \kappa(a', c) - \kappa^Q (a', c) \Big\} + \Big\{ \kappa^Q (a', c) - \kappa(a', c) \Big\} \Bigg] \, \mathrm{d} P(x) \\
            & = \int \bigg( \frac{1}{e_{a'}^Q (c)} - \frac{1}{e_{a'} (c)} \bigg) \Big( \kappa(a', c) - \kappa^Q (a', c) \Big) e_{a'}(c) \, \mathrm{d} P(x),
        \end{aligned}
    \]
    \end{small}
    \begin{small}
    \[
        \begin{aligned}
            & R_{j, a'}^{(2)} (Q, P \, ; \, \mathcal{G}_M) \\
            = & \int \, \mathrm{d} P(x) \bigg[ \frac{e_{a'}(c)}{e_{a'}^Q (c)} - 1\bigg]  \Big\{  \mu (c, a', m) - \mu^Q (c, a', m) \Big\} \pi_{c, a'}^Q \big(\pa_j (\mathcal{G}_M) \big)  \, \mathrm{d} \pa_j (\mathcal{G}_M) \\
            = & \int \bigg( \frac{1}{e_{a'}^Q (c)} - \frac{1}{e_{a'} (c)} \bigg) \Big(  \mu (c, a', m) - \mu^Q (c, a', m) \Big) e_{a'}(c)\pi_{c, a'}^Q \big(\pa_j (\mathcal{G}_M) \big) \, \mathrm{d} \pa_j (\mathcal{G}_M)  \, \mathrm{d} P(x),
        \end{aligned}
    \]
    \end{small}
    and 
    \begin{small}
    \[
        \begin{aligned}
            R_{j, a'}^{(3)} (Q, P \, ; \, \mathcal{G}_M) & = \int e_{a'}^Q (c) \, \mathrm{d} P(x) \int_{\mathcal{M}_j} \mu (c, a', m) \pi_{c}^Q (m_j) \, \mathrm{d} m_j \\
            & ~~ ~ \int_{\mathcal{M}_{\pa_j}(\mathcal{G}_M)}  \bigg( \frac{1}{e_{a'}^Q (c)} - \frac{1}{e_{a'} (c)} \bigg) \Big( {\pi}_{C, a'} \big(\pa_j( \mathcal{G}_M) \big) - {\pi}_{C, a'}^Q \big(\pa_j( \mathcal{G}_M) \big) \Big) \, \mathrm{d} \pa_j (\mathcal{G}_M),
        \end{aligned}
    \]
    \end{small}
    \begin{small}
    \[
        \begin{aligned}
            R_{j, a'}^{(4)} (Q, P \, ; \, \mathcal{G}_M) & = \int \frac{e_{a'}(c)}{e_{a'}^Q (c)}\pi_{c}^Q (m_j) \, \mathrm{d} P(x)   \\
            &  \int_{\mathcal{M}_{\pa_j}(\mathcal{G}_M)}  \Big(  {\pi}_{C, a'} \big(\pa_j( \mathcal{G}_M) \big) - {\pi}_{C, a'}^Q \big(\pa_j( \mathcal{G}_M) \big) \Big) \Big( \mu (c, a', m) - \mu^Q (c, a', m) \Big) \, \mathrm{d} \pa_j (\mathcal{G}_M).
        \end{aligned}
    \]
    \end{small}
    Therefore, by Assumption \ref{ass_conv_rate_QR}, we have for sufficient large $n$ such that 
    \[
        \Big| \widehat{f}(x_{S_2} \mid x_{S_1}) - {f}(x_{S_2} \mid x_{S_1}) \Big| \leq {f}(x_{S_2} \mid x_{S_1}) \quad \text{and} \quad \Big| \widehat{\E} \big[ x_{S_2} \mid x_{S_1}\big] - {\E}\big[ x_{S_2} \mid x_{S_1} \big] \Big| \leq {\E}\big[ x_{S_2} \mid x_{S_1} \big]
    \]
    almost surely. Then components $R_{j, a'}^{(\ell)}(\bcdot)$ with $\ell = 1, 2, 3, 4$ can be furthermore bounded by

    \[
        \begin{aligned}
            \big| R_{j, a'}^{(1)} (\widehat{P}, P) \big| & = \Bigg| \int \bigg( \frac{1}{\widehat{e}_{a'}(c)} - \frac{1}{e_{a'} (c)} \bigg) \Big( \kappa(a', c) - \widehat\kappa (a', c) \Big) e_{a'}(c) \, \mathrm{d} P(x) \Bigg| \\
            \alignedoverset{\text{by Assumption \ref{ass_pos}}}{\leq} \frac{1}{\varepsilon} \int \big| e_{a'}(c) - \widehat{e}_{a'}(c) \big| \big|  \kappa(a', c) - \widehat\kappa (a', c) \big| \, \mathrm{d} P(x) \\
            \alignedoverset{\text{by Assumption  \ref{ass_conv_rate_QR}}}{\asymp} n^{- \vartheta_{j, e}^* + \vartheta_{j, \mu}^*} = o_p (n^{-1 / 2}),
        \end{aligned}
    \]
    and similarly
    \[
        \begin{aligned}
            & \big| R_{j, a'}^{(2)} (\widehat{P}, P \, ; \, \mathcal{G}_M) \big| \\
            & = \Bigg| \int \bigg( \frac{1}{\widehat{e}_{a'} (c)} - \frac{1}{e_{a'} (c)} \bigg) \Big(  \mu (c, a', m) - \widehat\mu (c, a', m) \Big) e_{a'}(c)\pi_{c, a'}^Q \big(\pa_j (\mathcal{G}_M) \big) \, \mathrm{d} \pa_j (\mathcal{G}_M)  \, \mathrm{d} P(x) \Bigg| \\
             \alignedoverset{\text{by Assumption \ref{ass_pos}}}{\lesssim} \int \big| e_{a'}(c) - \widehat{e}_{a'}(c) \big| \big|  \mu (c, a', m) - \widehat\mu (c, a', m)  \big| \, \mathrm{d} P(x) \\
             \alignedoverset{\text{by Assumption \ref{ass_conv_rate_QR}}}{\asymp} n^{- \vartheta_{j, e}^* + \vartheta_{j, \mu}^*} = o_p (n^{-1 / 2}),
        \end{aligned}
    \]
    \[
        \begin{aligned}
            & \big| R_{j, a'}^{(3)} (\widehat{P}, P \, ; \, \mathcal{G}_M) \big| \\
            \alignedoverset{\text{by $\ell^2$ assmuption and Assumption \ref{ass_pos}}}{\lesssim} \int \big| e_{a'}(c) - \widehat{e}_{a'}(c) \big|\Big| {\pi}_{C, a'} \big(\pa_j( \mathcal{G}_M) \big) - \widehat{\pi}_{C, a'} \big(\pa_j( \mathcal{G}_M) \big) \Big| \, \mathrm{d} P(x) \\
            \alignedoverset{\text{by Assumption \ref{ass_conv_rate_QR}}}{\asymp} n^{- \vartheta_{j, e}^* + \vartheta_{j, \pi}^*} = o_p (n^{-1 / 2}),
        \end{aligned}
    \]
    and
    \[
        \begin{aligned}
            & \big| R_{a'}^{(4)} (\widehat{P}, P \, ; \, \mathcal{G}_M) \big| \\
            \alignedoverset{\text{by $\ell^2$ assmuption and Assumption \ref{ass_pos}}}{\lesssim} \int \Big| {\pi}_{C, a'} \big(\pa_j( \mathcal{G}_M) \big) - \widehat{\pi}_{C, a'} \big(\pa_j( \mathcal{G}_M) \big) \Big| \Big| \mu (c, a', m) - \widehat\mu (c, a', m) \Big| \, \mathrm{d} P(x) \\
            \alignedoverset{\text{by Assumption \ref{ass_conv_rate_QR}}}{\asymp} n^{- \vartheta_{j, \pi}^* + \vartheta_{j, \mu}^*} = o_p (n^{-1 / 2}).
        \end{aligned}
    \]
    Therefore, for the reminder, we have
    \begin{equation}\label{reminder2_final_QR}
        \begin{aligned}
    	& \frac{1}{\# \operatorname{MEC}(\mathcal{C}_M)} \sum_{\mathcal{G}_M \in \operatorname{MEC}(\mathcal{C}_M)} R_{j, 1}^{\text{QR}}(\widehat{P}, {P} \, ; \, \mathcal{G}_M) \\
            & \asymp \frac{1}{\# \operatorname{MEC}(\mathcal{C}_M)} \sum_{\mathcal{G}_M \in \operatorname{MEC}(\mathcal{C}_M)} R_{j, 0}^{\text{QR}}(\widehat{P}, {P} \, ; \, \mathcal{G}_M)\\\
            & =  o_p \big(n^{- 1 / 2} \big)
        \end{aligned} 
    \end{equation}
    for our quadruply robust estimators. Plug \eqref{reminder1_final_QR} and \eqref{reminder2_final_QR} into \eqref{TM_QR_avg_decompose}, we obtain
    \[
    	\begin{aligned}
    		& \sqrt{n} \big( \widehat{TM}_j^{\text{avg}, \, \text{QR}} - \overline{TM}_j \big) \\
            & = \frac{1}{\# \operatorname{MEC}(\mathcal{C}_M)} \sum_{\mathcal{G}_M \in \operatorname{MEC}(\mathcal{C}_M)} \Big[ \mathbb{G}_n \big\{ \upvarphi_{j, 1}^{\text{QR}}(X, P \, ; \, \mathcal{G}_M ) - \upvarphi_{j, 0}^{\text{QR}} (X, P \, ; \, \mathcal{G}_M ) \big\} + o_p(1) \Big] \\
            & \overset{\text{by $o_p(1)$ is uniform on $\operatorname{MEC}(\mathcal{C}_M)$}}{=} \mathbb{G}_n \Bigg\{ \frac{1}{\# \operatorname{MEC}(\mathcal{C}_M)} \sum_{\mathcal{G}_M \in \operatorname{MEC}(\mathcal{C}_M)} \big\{ \upvarphi_{j, 1}^{\text{QR}} (X, P \, ; \, \mathcal{G}_M ) - \upvarphi_{j, 0}^{\text{QR}} (X, P \, ; \, \mathcal{G}_M ) \big\} \Bigg\} \\
            & ~~~~~~~~~~~~~~~~~~~~~~~~~~~~~~~~~~~~~~~~~~~~~~~~~ + o_p(1) \\
            & \rightsquigarrow \, \mathcal{N} \left(0, \E \left[ \frac{1}{\# \operatorname{MEC}(\mathcal{C}_{0, M})} \sum_{\mathcal{G}_M \in \operatorname{MEC} (\mathcal{C}_{0, M})} S^{\text{eff}, \, \text{nonpar}} \big( TM_j (\mathcal{G}_M)\big) \right]^2 \right)
    	\end{aligned}
    \]
    which completes the proof of the semiparametric efficiency for $\widehat{TM}_j^{\text{avg}, \,\text{QR}}$.
    %Finally, the multivariate asymptotic normality  comes from standard increasing dimensional literature, \cite{van2014asymptotically} for example.
\end{proof}

\end{document}